\documentclass[a4paper,11pt]{article}

\usepackage{a4wide}
\usepackage{xspace}
\usepackage{graphicx}
\usepackage{subfigure}
\usepackage{psfrag}
\usepackage{amsmath,amssymb}
\usepackage{color}
\usepackage{hyperref}

\newtheorem{theorem}{Theorem}
\newtheorem{lemma}[theorem]{Lemma}
\newtheorem{corollary}[theorem]{Corollary}

\newtheorem{remark}{Remark}

\newenvironment{proof}{\smallskip\par\noindent\emph{Proof.\ }}%
{\par\medskip}

\newenvironment{mathdescription}%
  {\begin{list}{}%
    {\renewcommand{\makelabel}[1]%
      {\textbf{\mathversion{bold}{##1}\mathversion{normal}\hfil}}}}
  {\end{list}}

\newcommand{\figwidth}{}
\newcommand{\hide}[1]{}

\newcommand{\polyapprox}[1]{{\tilde {#1}}}
\newcommand{\compl}[1]{{{#1}^\alpha}}
\newcommand{\intrr}[1]{{{#1}^\circ}}
\newcommand{\tr}[1]{\mathcal{T}({#1})}
\newcommand{\sect}[1]{\mathcal{S}({#1})}
\newcommand{\proofbox}{\hfill$\Box$}
\newcommand{\ie}{i.e.,\xspace}
\newcommand{\resp}{resp.,\xspace}
\newcommand{\pconvex}{piecewise-convex\xspace}
\newcommand{\pconcave}{piecewise-concave\xspace}
\newcommand{\Pconvex}{Piecewise-convex\xspace}
\newcommand{\Pconcave}{Piecewise-concave\xspace}

\newcommand{\acsProjectNr}{IST-006413}

\newcommand{\acsAcronym}{ACS}
\newcommand{\acsTitle}{Algorithms for Complex Shapes}
\newcommand{\acsSubTitle}{with Certified Numerics and Topology}
\newcommand{\acsFullTitle}{\acsTitle\ \acsSubTitle}

\newcommand{\ACSacknowledgement}{Work partially supported by the IST
  Programme of the EU (FET Open) Project under Contract No
  \acsProjectNr\ -- (\acsAcronym\ - \acsFullTitle)}

\title{Guarding curvilinear art galleries with vertex or point guards}

\author{Menelaos I. Karavelas$^{\dagger,\ddagger}$\hfil{}
Elias P. Tsigaridas$^\star$\\[5pt]
\it{}$^\dagger$Department of Applied Mathematics,\\
\it{}University of Crete\\
\it{}GR-714 09 Heraklion, Greece,\\
\texttt{mkaravel@tem.uoc.gr}\\[5pt]
\it{}$^\ddagger$Institute of Applied and Computational Mathematics,\\
\it{}Foundation for Research and Technology - Hellas,\\
\it{}P.O. Box 1385, GR-711 10 Heraklion, Greece\\[5pt]
\it{}$^\star$LORIA-INRIA Lorraine,\\
\it{} 615 rue du Jardin Botanique, BP 101,\\
\it{}54602 Villers-l{\'e}-Nancy Cedex,  France,\\
\texttt{Elias.Tsigaridas@loria.fr}}

\begin{document}

\maketitle

\begin{abstract}
One of the earliest and most well known problems in computational
geometry is the so-called \emph{art gallery problem}. The goal is to
compute the minimum possible number guards placed on the vertices of a
simple polygon in such a way that they cover the interior of the
polygon.

In this paper we consider the problem of guarding an art gallery which
is modeled as a polygon with curvilinear walls. Our main focus is on
polygons the edges of which are convex arcs pointing towards the
exterior or interior of the polygon (but not both), named \pconvex
and \pconcave polygons. We prove that, in the case of
\pconvex polygons, if we only allow vertex guards,
$\lfloor\frac{4n}{7}\rfloor-1$ guards are sometimes necessary, and
$\lfloor\frac{2n}{3}\rfloor$ guards are always sufficient. Moreover,
an $O(n\log{}n)$ time and $O(n)$ space algorithm is described that
produces a vertex guarding set of size at most
$\lfloor\frac{2n}{3}\rfloor$. When we allow point guards the
afore-mentioned lower bound drops down to
$\lfloor\frac{n}{2}\rfloor$.
In the special case of monotone \pconvex polygons we can show
that $\lfloor\frac{n}{2}\rfloor$ vertex guards are always sufficient
and sometimes necessary; these bounds remain valid even if we allow
point guards.

In the case of \pconcave polygons, we show that $2n-4$ point
guards are always sufficient and sometimes necessary, whereas it might
not be possible to guard such polygons by vertex guards. We conclude
with bounds for other types of curvilinear polygons and future work.
\end{abstract}


\clearpage
\newcommand{\myparagraph}[1]{\par\noindent\textbf{#1}$\quad$}


\section{Introduction}
\label{sec:intro}

Consider a simple polygon $P$ with $n$ vertices. How many points with
omnidirectional visibility are required in order to see every point in
the interior of $P$? This problem, known as the \emph{art gallery
  problem} has been one of the earliest problems in Computational
Geometry. Applications areas include robotics
\cite{ks-errss,xcb-pvicb-86}, motion planning
\cite{lw-apcfp-79,m-aaspt-89}, computer vision and pattern recognition
\cite{sc-bwfmr-86,y-3damv-86,ae-caps-83,t-prgc-80}, graphics
\cite{m-wcohs-87,ci-tsc-84}, CAD/CAM \cite{b-beddc-88,ek-hstsr-89} and
wireless networks \cite{egs-gpepp-07}.
In the late 1980's to mid 1990's interest moved from linear polygonal
objects to curvilinear objects
\cite{uz-ics-89,cgl-iscd-89,cruz-pcs-94,cgru-ihcs-95} ---
see also the paper by Dobkin and Souvaine \cite{ds-cgcw-90} that
extends linear polygon algorithms to curvilinear polygons, as well as
the recent book by Boissonnat and Teillaud \cite{BoiTeil:ECG:book} for
a collection of results on non-linear computational geometry beyond
art gallery related problems.
In this context this paper addresses the classical art gallery problem
for various classes of polygonal regions the edges of which are arcs
of curves. To the best of our knowledge this is the first time that
the art gallery problem is considered in this context.

The first results on the art gallery problem or its variations date
back to the 1970's. Chv{\'a}tal \cite{c-ctpg-75} was the first to
prove that a simple polygon with $n$ vertices can be always guarded
with $\lfloor\frac{n}{3}\rfloor$ vertices; this bound is tight in the
worst case. The proof by Chv{\'a}tal was quite tedious and Fisk
\cite{f-spcwt-78} gave a much simpler proof by means of triangulating
the polygon and coloring its vertices using three colors in such a way
so that every triangle in the triangulation of the polygon does not
contain two vertices of the same color. The algorithm proposed by Fisk
runs in $O(T(n)+n)$ time, where $T(n)$ is the time to triangulate a
simple polygon. Following Chazelle's linear-time algorithm for
triangulating a simple polygon \cite{c-tsplt-90i,c-tsplt-91a}, the
algorithm proposed by Fisk runs in $O(n)$ time. Lee and Lin
\cite{ll-ccagp-86} showed that computing the minimum number of vertex
guards for a simple polygon is NP-hard, which was extended to
point guards by Aggarwal \cite{a-agpiv-84}. Soon afterwards other
types of polygons were considered. Kahn, Klawe and Kleitman
\cite{kkk-tgrfw-83} showed that orthogonal polygons of size $n$, \ie
polygons with axes-aligned edges, can be guarded with
$\lfloor\frac{n}{4}\rfloor$ vertex guards, which is also a lower
bound. Several $O(n)$ algorithms have been proposed for this variation
of the problem, notably by Sack \cite{s-rcg-84}, who gave the first
such algorithm, and later on by Lubiw \cite{l-dprcq-85}. Edelsbrunner,
O'Rourke and Welzl \cite{eow-sgrag-84} gave a linear time algorithm
for guarding orthogonal polygons with $\lfloor\frac{n}{4}\rfloor$
point guards.

Beside simple polygons and simple orthogonal polygons, polygons with
holes, and orthogonal polygons with holes have been investigated. As
far as the type of guards is concerned, \emph{edge guards} and
\emph{mobile guards} have been considered. An edge guard is an edge of
the polygon, and a point is visible from it if it is visible from at
least one point on the edge; mobile guards are essentially either
edges of the polygon, or diagonals of the polygon. Other types of
guarding problems have also been studied in the literature,
notably, the \emph{fortress problem} (guard the exterior of the
polygon against enemy raids) and the \emph{prison yard problem} (guard
both the interior and the exterior of the polygon which represents a
prison: prisoners must be guarded in the interior of the prison and
should not be allowed to escape out of the prison). For a detailed
discussion of these variations and the corresponding results the
interested reader should refer to the book by O'Rourke
\cite{o-agta-87}, the survey paper by Shermer \cite{s-rrag-92} and the
book chapter by Urrutia \cite{u-agip-00}.

In this paper we consider the original problem, that is the problem of
guarding a simple polygon. We are primarily interested in the case of
vertex guards, although results about point guards are also
described. In our case, polygons are not required to have linear
edges. On the contrary we consider polygons that have smooth curvilinear
edges. Clearly, these problems are NP-hard, since they are direct
generalizations of the corresponding original art gallery problems.
In the most general setting where we impose no restriction on
the type of edges of the polygon, it is very easy to see that there
exist curvilinear polygons that cannot be guarded with vertex guards, or
require an infinite number of point guards (see
Fig. \ref{fig:nonconvexinfinite}). Restricting the edges of the
polygon to be locally convex curves, pointing towards the exterior of
the polygon (\ie the polygon is a locally convex set, except
possibly at the vertices) we can construct polygons that require a
minimum of $n$ vertex or point guards, where $n$ is the number of
vertices of the polygon (see Fig. \ref{fig:locallyconvexall}); in fact
such polygons can always be guarded with their $n$ vertices.
The main focus of this paper is the class of polygons that are either
locally convex or locally concave (except possibly at the vertices),
the edges of which are convex arcs; we call such polygons
\emph{\pconvex} and \emph{\pconcave polygons},
respectively.

For the first class of polygons we show that it is always possible to
guard them with $\lfloor\frac{2n}{3}\rfloor$ vertex guards, where $n$
is the number of polygon vertices. On the other hand we describe
families of \pconvex polygons that require a minimum of
$\lfloor\frac{4n}{7}\rfloor-1$ vertex guards and
$\lfloor\frac{n}{2}\rfloor$ point guards. Aside from the combinatorial
complexity type of results, we describe an $O(n\log{}n)$ time and
$O(n)$ space algorithm which, given a \pconvex polygon,
computes a guarding set of size at most
$\lfloor\frac{2n}{3}\rfloor$. Our algorithm should be viewed as a
generalization of Fisk's algorithm \cite{f-spcwt-78}; in fact, when
applied to polygons with linear edges, it produces a guarding set of
size at most $\lfloor\frac{n}{3}\rfloor$.
For the purposes of our complexity analysis and results, we assume,
throughout the paper, that the curvilinear edges of our polygons are
arcs of algebraic curves of constant degree; as a result all
predicates required by the algorithms described in this paper take
$O(1)$ time in the Real RAM computation model.
The central idea for both obtaining the upper bound as well as for
designing our algorithm is to approximate the \pconvex polygon
by a linear polygon (a polygon with line segments as
edges). Additional auxiliary vertices are added on the boundary of the
curvilinear polygon in order to achieve this. The resulting linear polygon
has the same topology as the original polygon and captures the
essentials of the geometry of the \pconvex polygon; for
obvious reasons we term this linear polygon the
\emph{polygonal approximation}. Once the polygonal approximation has
been constructed, we compute a guarding set for it by applying a
slight modification of Fisk's algorithm \cite{f-spcwt-78}. The
guarding set just computed for the polygonal approximation turns out
to be a guarding set for the original curvilinear polygon. The final step
of both the proof and our algorithm consists in mapping the guarding
set of the polygonal approximation to another vertex guarding set
consisting of vertices of the original polygon only.

If we further restrict ourselves to \emph{monotone \pconvex
  polygons}, \ie \pconvex polygons that have the property
that there exists a line $L$, such that any line $L^\perp$
perpendicular to $L$ intersects the polygon at most twice, we can show
that $\lfloor\frac{n}{2}\rfloor+1$ vertex or $\lfloor\frac{n}{2}\rfloor$
point guards are always sufficient and sometimes necessary. Such a
line $L$ can be computed in $O(n)$ time (cf. \cite{ds-cgcw-90}). Given
$L$, it is very easy to compute a vertex guarding set of size
$\lfloor\frac{n}{2}\rfloor+1$, or a point guarding set of size
$\lfloor\frac{n}{2}\rfloor$: the problem of computing such 
a guarding set essentially reduces to merging two sorted arrays, thus
taking $O(n)$ time and $O(n)$ space. This result should be contrasted
against the case of monotone linear polygons where the corresponding
upper and lower bound on the number of vertex or point guards required
to guard the polygon matches that of general (\ie not necessarily
monotone) linear polygons. In other words, monotonicity seems to play
a crucial  role in the case of \pconvex polygons, which is not
the case for linear polygons.

For the second class of polygons, \ie the class of \pconcave
polygons, vertex guards may not be sufficient in order to guard the
interior of the polygon (see
Fig. \ref{fig:concave_no_vertex_guards}). We thus turn our attention
to point guards, and we show that $2n-4$ point guards are always
sufficient and sometimes necessary. Our method for showing the
sufficiency result is similar to the technique used to illuminate
sets of disjoint convex objects on the plane \cite{f-icd-77}. Given 
a \pconcave polygon $P$, we construct a new locally concave
polygon $Q$, contained inside $P$, and such that the tangencies
between edges of $Q$ are maximized. The problem of guarding $P$ then
reduces to the problem of guarding $Q$, which essentially consists of
a number of faces with pairwise disjoint interiors. The faces of $Q$
require, each, two point guards in order to be guarded, and are in
1--1 correspondence with the triangles of an appropriately defined
triangulation graph $\tr{R}$ of a polygon $R$ with $n$ vertices. Thus
the number point guards required to guard $P$ is at most two times
the number of faces of $\tr{R}$, \ie $2n-4$.

The rest of the paper is structured as follows. In Section
\ref{sec:preliminaries}  we introduce some notation and provide
various definitions. In Section \ref{sec:piececonvex} we present our
algorithm for computing a guarding set, of size
$\lfloor\frac{2n}{3}\rfloor$, for a \pconvex polygon with $n$
vertices. Section \ref{sec:piececonvex} is further subdivided into five
subsections. In Subsection \ref{sec:polyapprox} we define the
polygonal approximation of our curvilinear polygon and prove some geometric
and combinatorial properties. In Subsection \ref{sec:triangulation} we
show how to construct a, properly chosen, \emph{constrained}
triangulation of the polygonal approximation. In Subsection
\ref{sec:guardingset} we describe how to compute the guarding set for
the original curvilinear polygon from the guarding set of the polygonal
approximation due to Fisk's algorithm and prove the upper bound on the
cardinality of the guarding set. In Subsection
\ref{sec:guardingalgo} we show how to compute the guarding set in
$O(n\log{}n)$ time and $O(n)$ space. Finally, in Subsection
\ref{sec:lowerbound} is devoted to the presentation of the family of
polygons that attains the lower bound of
$\lfloor\frac{4n}{7}\rfloor-1$ vertex guards.
The special case of guarding monotone \pconvex
polygons is discussed in Section \ref{sec:monotonepiececonvex}.
We show that $\lfloor\frac{n}{2}\rfloor+1$ vertex (or
$\lfloor\frac{n}{2}\rfloor$ point) guards are
always necessary and sometimes sufficient, and present an $O(n)$ time
and $O(n)$ space algorithm for computing such a guarding set.
In Section \ref{sec:piececoncave} we present our results for
\pconcave polygons, namely, that $2n-4$ point guards are
always necessary and sometimes sufficient for this class of polygons.
Section \ref{sec:otherlb} contains further results. More
precisely, we present bounds for locally convex polygons, monotone
locally convex polygons and general polygons. The final section of the
paper, Section \ref{sec:summary}, summarizes our results and discusses
open problems.


\section{Definitions}
\label{sec:preliminaries}

\begin{figure}[t]
\begin{center}
  \subfigure[\label{fig:linearpoly}]
            {\includegraphics[width=0.2\textwidth]{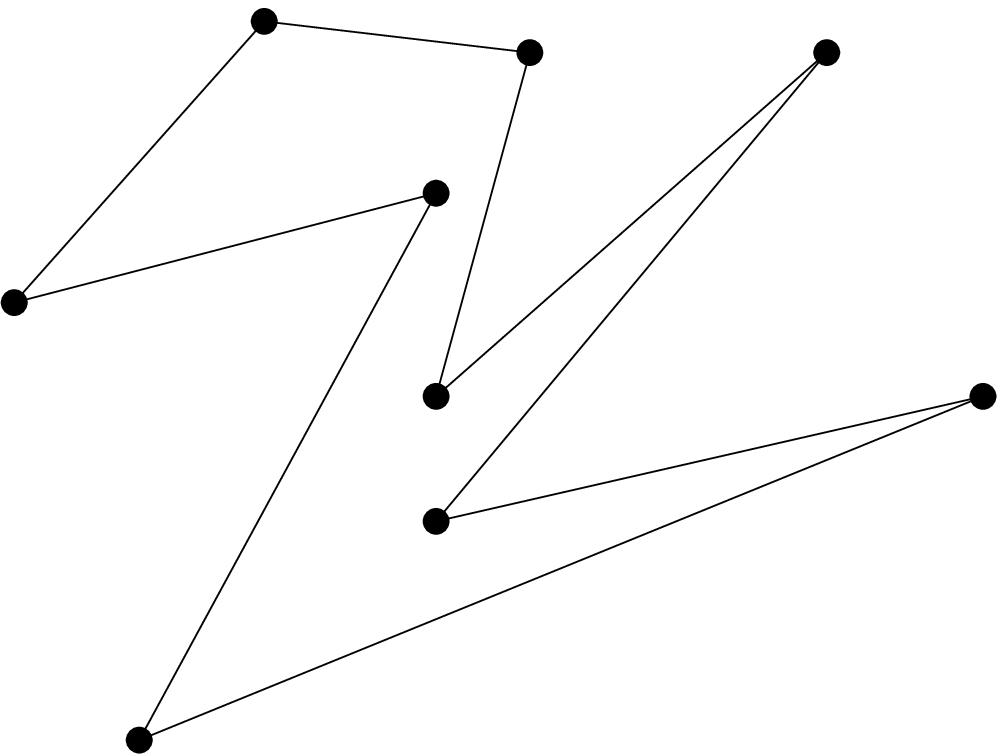}}
            \hfil
  \subfigure[\label{fig:convexpoly}]
            {\includegraphics[width=0.2\textwidth]{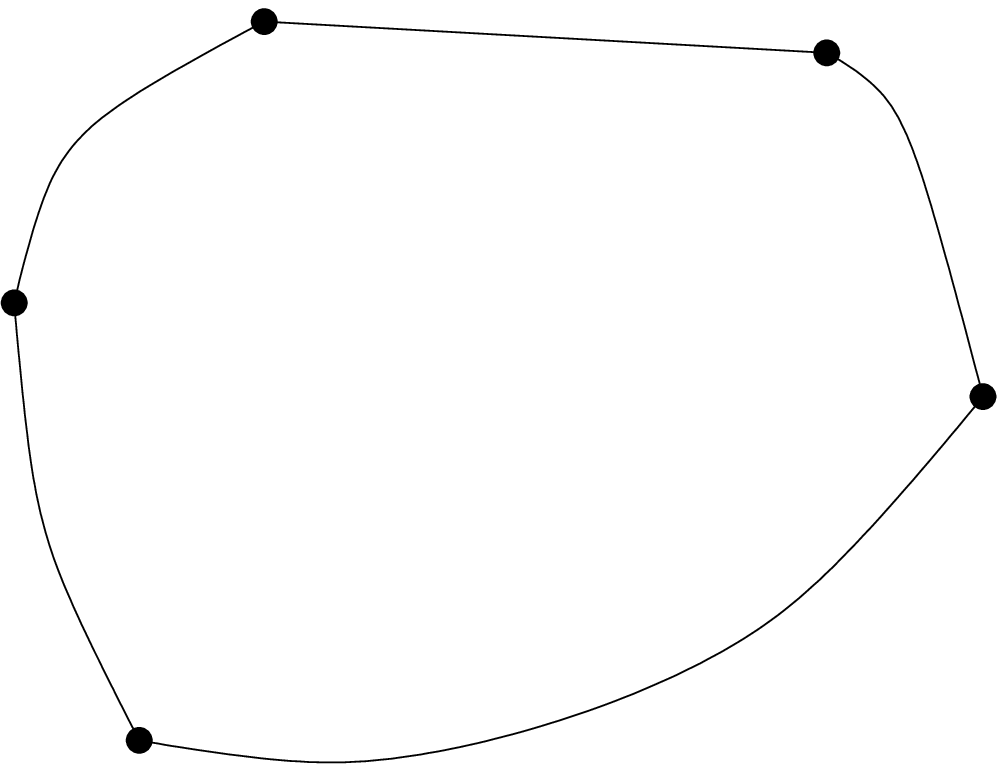}}
            \hfil
  \subfigure[\label{fig:piececonvexpoly}]
            {\includegraphics[width=0.2\textwidth]{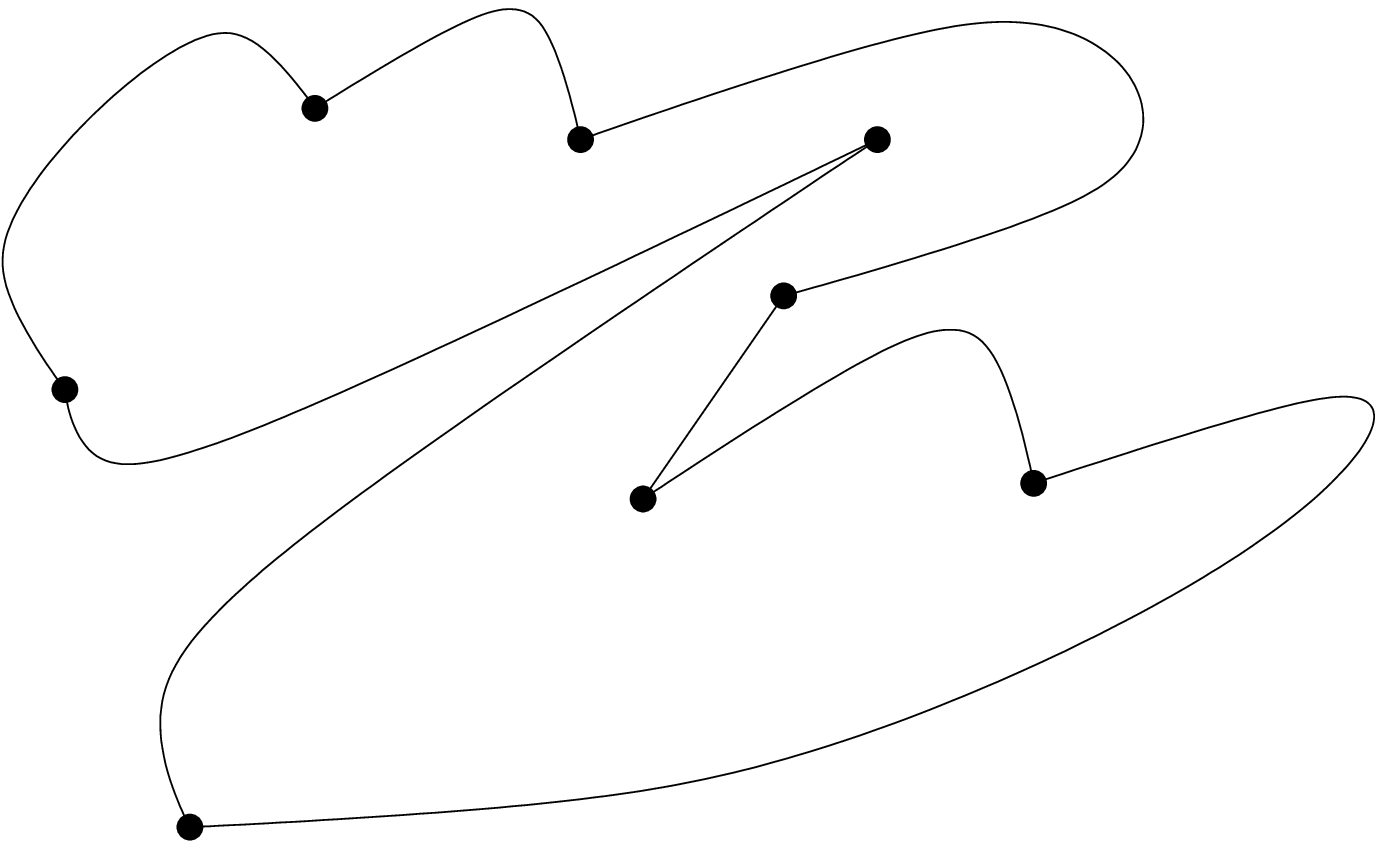}}
            \\
  \subfigure[\label{fig:locallyconvexpoly}]
            {\includegraphics[width=0.2\textwidth]{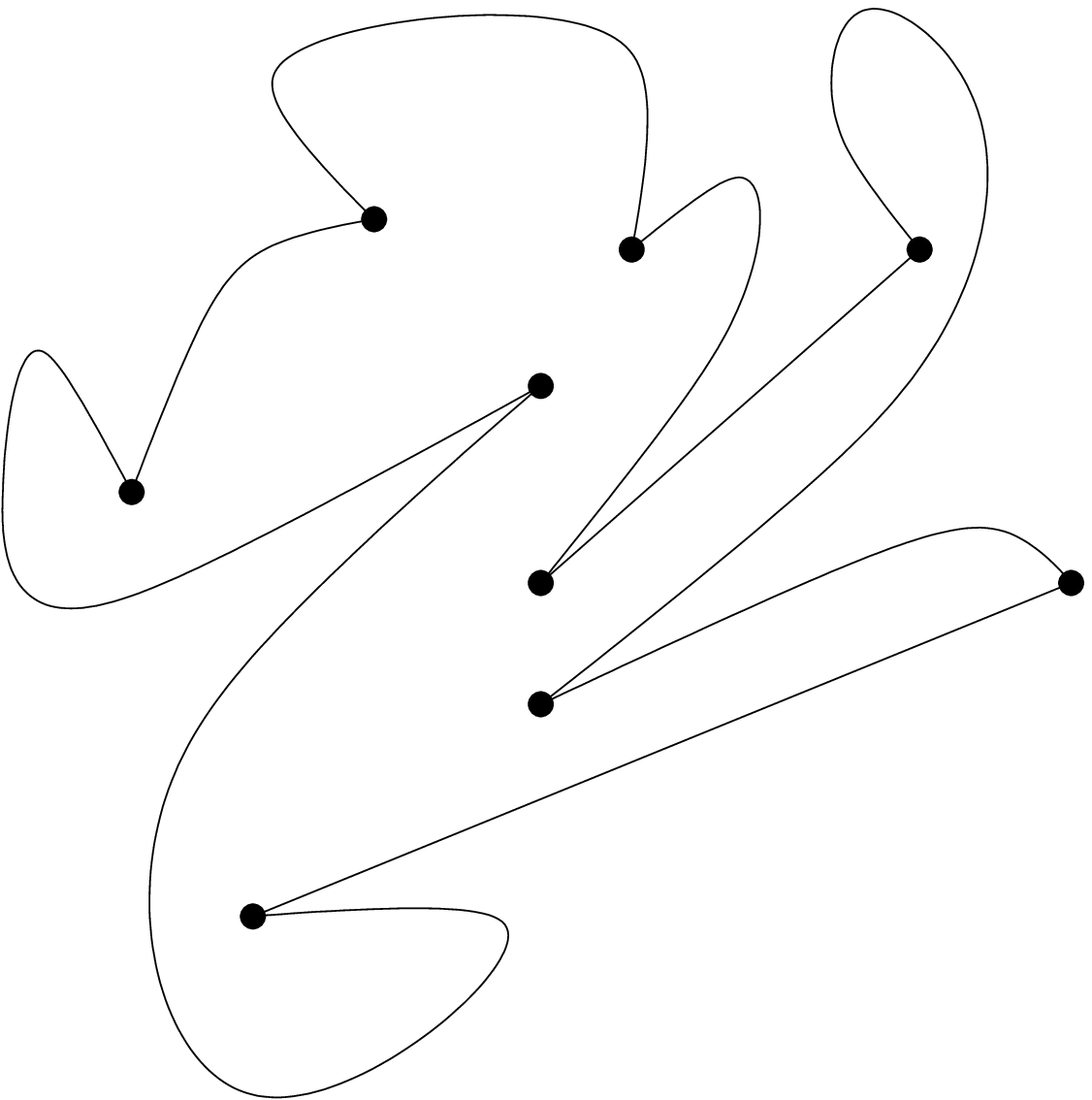}}
            \hfil
  \subfigure[\label{fig:piececoncavepoly}]
            {\includegraphics[width=0.2\textwidth]{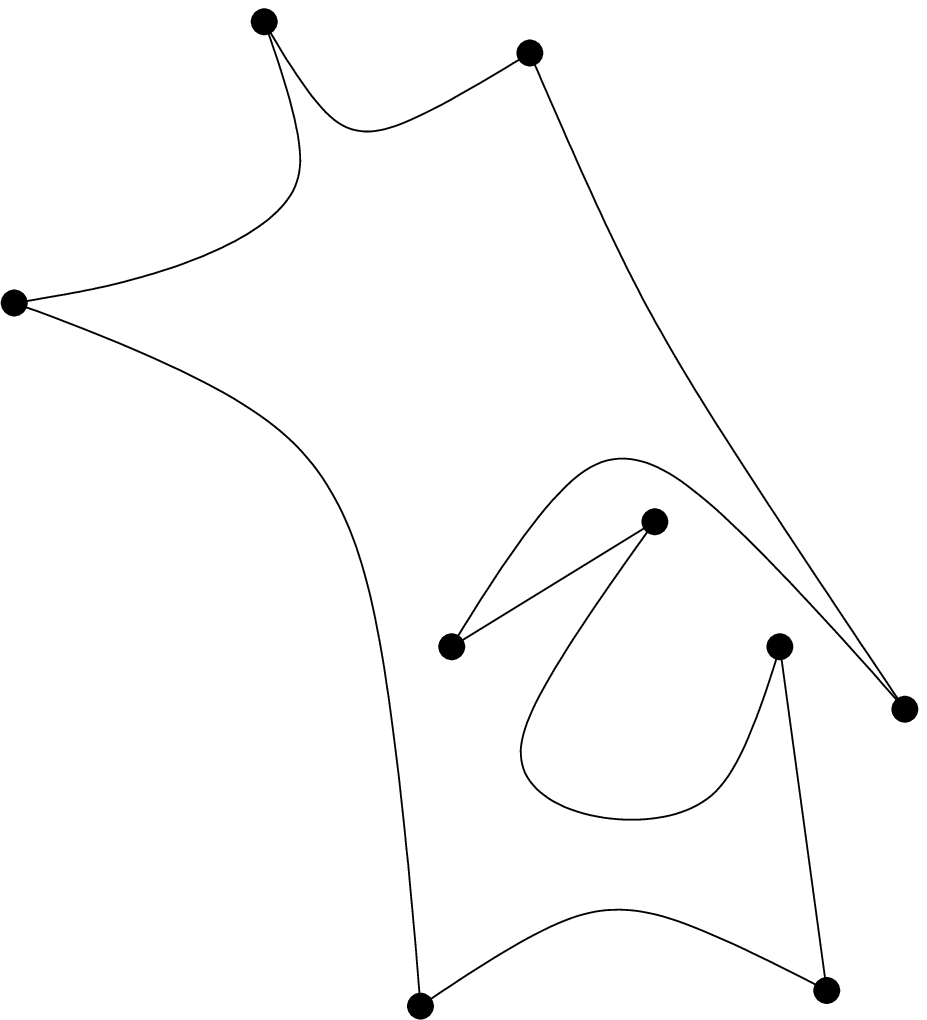}}
            \hfil
  \subfigure[\label{fig:nonconvexpoly}]
            {\includegraphics[width=0.2\textwidth]{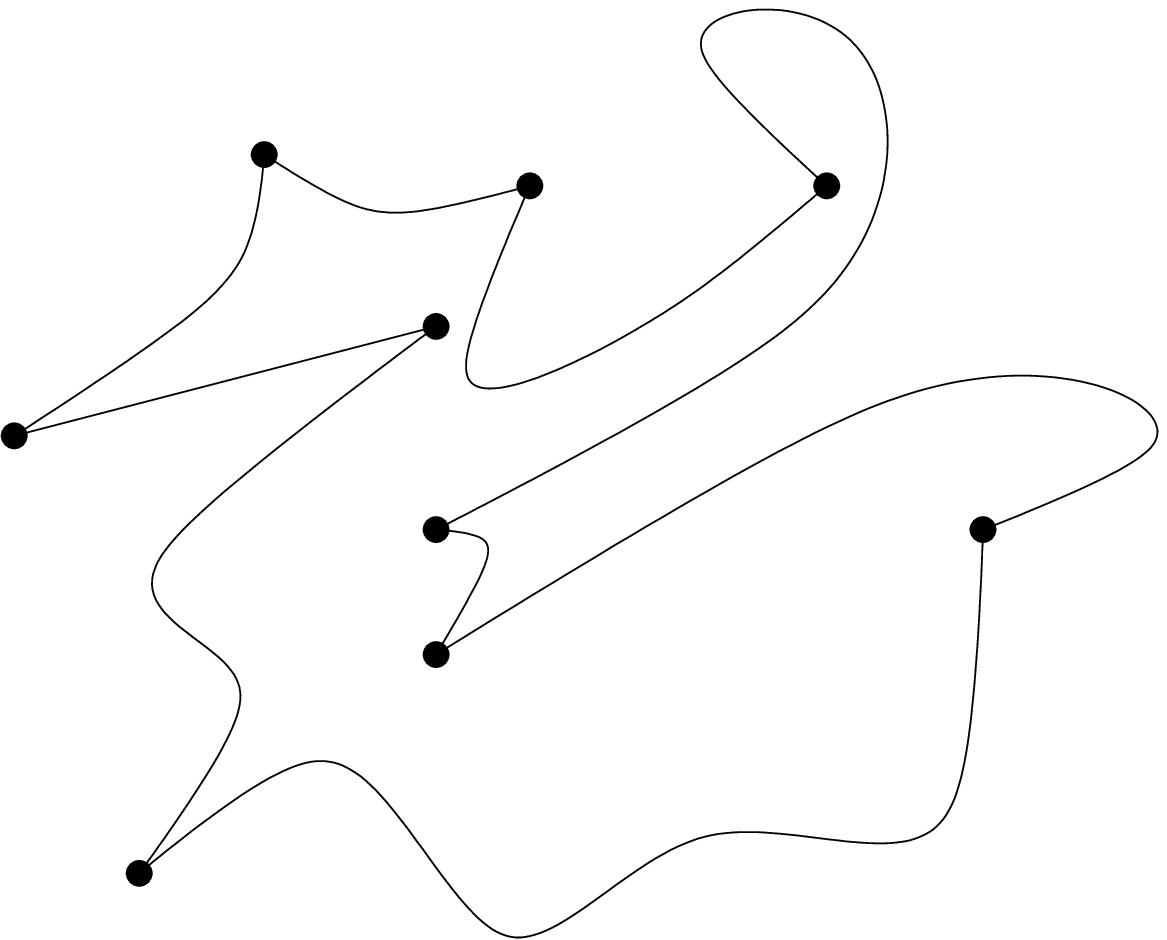}}
  \caption{Different types of curvilinear polygons: (a) a linear polygon,
    (b) a convex polygon, (c) a \pconvex polygon, (d) a
    locally convex polygon, (e) a \pconcave polygon and (f) a
    general polygon.}
  \label{fig:polygontypes}
\end{center}
\end{figure}

\myparagraph{Curvilinear arcs.}
Let $S$ be a sequence of points $v_1,\ldots,v_n$ and $E$ a set of
curvilinear arcs $a_1,\ldots,a_n$, such that $a_i$ has as endpoints
the points $v_i$ and $v_{i+1}$\footnote{Indices are considered to be
  evaluated modulo $n$.}. We will assume that
the arcs $a_i$ and $a_j$, $i\ne j$, do not intersect, except when
$j=i-1$ or $j=i+1$, in which case they intersect only at the points
$v_i$ and $v_{i+1}$, respectively . We define a \emph{curvilinear polygon} $P$
to be the closed region delimited by the arcs $a_i$. The points $v_i$
are called the vertices of $P$. An arc $a_i$ is
a \emph{convex arc} if every line on the plane intersects $a_i$ at
either at most two points or along a linear segment. 
If $q$ is a point in the interior of $a_i$, an
$\varepsilon$-neighborhood $n_\varepsilon(q)$ of $q$ is defined to be
the intersection of $a_i$ with a disk centered at $q$ with radius
$\varepsilon$. An arc $a_i$ is a \emph{locally convex arc} if for
every point $q$ in the interior of $a_i$, there exists an
$\varepsilon_q$ such that for every $0<\varepsilon \le \varepsilon_q$,
the $\varepsilon$-neighborhood of $q$ lies entirely in one of the two
halfspaces defined by the line $\ell$ tangent to $a_i$ at $q$; note
that if $\ell$ is not uniquely defined, then the
containment-in-halfspace property mentioned just above has to hold for
any such line $\ell$. Finally, note that a convex arc is also a
locally convex arc.

Our definition does not really require that the arcs $a_i$ are
smooth. In fact the arcs $a_i$ can be polylines, in which case the
results presented in this paper are still valid. What might be
different, however, is our complexity analyses, since we have assumed
that the $a_i$'s have constant complexity. In the remainder of this
paper, and unless otherwise stated, we will assume that the arcs $a_i$
are $G^1$-continuous and have constant complexity.

\myparagraph{Curvilinear polygons.}
A polygon $P$ is a \emph{linear polygon} if its edges are line
segments (see Fig. \ref{fig:linearpoly}).
A polygon $P$ consisting of curvilinear arcs as edges is called a
\emph{convex polygon} if every line on the plane intersects its
boundary at either at most two points or along a line segment (see
Fig. \ref{fig:convexpoly}).
A polygon is called a \emph{\pconvex polygon}, if every arc is
a convex arc and for every point $q$ in the interior of an arc $a_i$
of the polygon, the interior of the polygon is locally on the same
side as the arc $a_i$ with respect to the line tangent to $a_i$ at
$q$ (see Fig. \ref{fig:piececonvexpoly}).
A polygon is called a \emph{locally convex polygon} if the
boundary of the polygon is a locally convex curve, except possibly at
its vertices (see Fig. \ref{fig:locallyconvexpoly}). 
Note that a convex polygon is a \pconvex polygon and that a
\pconvex polygon is also a locally convex polygon. 
A polygon $P$ is called a \emph{\pconcave polygon}, if every
arc of $P$ is convex and for every point $q$ in the interior of
a non-linear arc $a_i$, the interior of $P$ lies locally on both sides
of the line tangent to $a_i$ at $q$ (see
Fig. \ref{fig:piececoncavepoly}).
Finally, a polygon is said to be a \emph{general polygon} if we impose
no restrictions on the type of its edges (see
Fig. \ref{fig:nonconvexpoly}).
We will use the term \emph{curvilinear polygon} to refer to a polygon the
edges of which are either line or curve segments.

\myparagraph{Guards and guarding sets.}
In our setting, a \emph{guard} or \emph{point guard} is a point in the
interior or on the boundary of a curvilinear polygon $P$. A guard of $P$
that is also a vertex of $P$ is called a \emph{vertex guard}.
We say that a curvilinear polygon $P$ is \emph{guarded} by a set $G$ of
guards if every point in $P$ is visible from at least one point in
$G$. The set $G$ that has this property is called a
\emph{guarding set} for $P$. A guarding set that consists solely of
vertices of $P$ is called a \emph{vertex guarding set}.


\section{\Pconvex polygons}\label{sec:piececonvex}

In this section we present an algorithm which, given a \pconvex
polygon $P$ of size $n$, it computes a vertex guarding set $G$ of
size $\lfloor\frac{2n}{3}\rfloor$. The basic steps of the algorithm are
as follows:
\begin{enumerate}
\item Compute the polygonal approximation $\polyapprox{P}$ of $P$.
\item\label{step:trPA}
  Compute a constrained triangulation $\tr{\polyapprox{P}}$ of
  $\polyapprox{P}$.
\item Compute a guarding set $G_{\polyapprox{P}}$ for
  $\polyapprox{P}$, by coloring the vertices of $\tr{\polyapprox{P}}$
  using three colors.
\item\label{algo:compg}
  Compute a guarding set $G_P$ for $P$ from the guarding set
  $G_{\polyapprox{P}}$.
\end{enumerate}


\subsection{Polygonalization of a \pconvex polygon}
\label{sec:polyapprox}

Let $a_i$ be a convex arc with endpoints $v_i$ and
$v_{i+1}$. We call the convex region $r_i$ delimited by $a_i$
and the line segment $v_iv_{i+1}$ a \emph{room}. A room is
called degenerate if the arc $a_i$ is a line segment. A line segment
$pq$, where $p,q\in a_i$ is called a \emph{chord}, and the region
delimited by the chord $pq$ and $a_i$ is called a \emph{sector}. The
chord of a room $r_i$ is defined to be the line segment $v_iv_{i+1}$
connecting the endpoints of the corresponding arc $a_i$.
A degenerate sector is a sector with empty interior.
We distinguish between two types of rooms (see Fig. \ref{fig:rooms}):
\begin{enumerate}
\item \emph{empty rooms}: these are non-degenerate rooms that do not
  contain any vertex of $P$ in the interior of $r_i$ or in the
  interior of the chord $v_iv_{i+1}$.
\item \emph{non-empty rooms}: these are non-degenerate rooms that
  contain at least one vertex of $P$ in the interior of $r_i$ or in
  the interior of the chord $v_iv_{i+1}$.
\end{enumerate}

\begin{figure}[t]
\begin{center}
\psfrag{r1}[][]{$r_{ne}'$}
\psfrag{r2}[][]{$r_{e}'$}
\psfrag{r3}[][]{$r_{ne}''$}
\psfrag{r4}[][]{$r_{e}''$}
\includegraphics[width=0.5\textwidth]{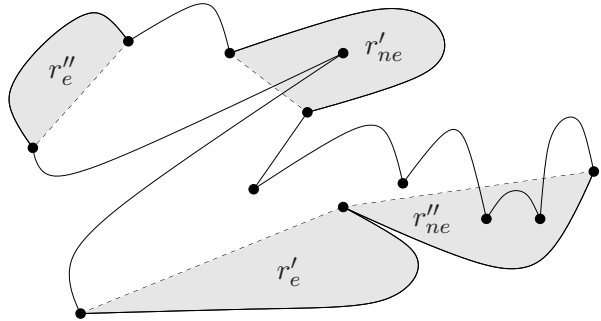}
\caption{The two types of rooms in a \pconvex polygon:
  $r_{e}'$ and $r_{e}''$ are empty rooms, whereas $r_{ne}'$ and
  $r_{ne}''$ are non-empty rooms.}
\label{fig:rooms}
\end{center}
\end{figure}

In order to polygonalize $P$ we are going to add new vertices in the
interior of non-linear convex arcs. To distinguish between the two
types of vertices, the $n$ vertices of $P$ will be called
\emph{original vertices}, whereas the additional vertices will be
called \emph{auxiliary vertices}.

\begin{figure}[hbt]
\begin{center}
\psfrag{v1}[][]{$v_1$}
\psfrag{v2}[][]{$v_2$}
\psfrag{v3}[][]{$v_3$}
\psfrag{v4}[][]{$v_4$}
\psfrag{v5}[][]{$v_5$}
\psfrag{v6}[][]{$v_6$}
\psfrag{v7}[][]{$v_7$}
\psfrag{m5}[][]{$m_5$}
\psfrag{a3}[][]{$a_3$}
\psfrag{a5}[][]{$a_5$}
\psfrag{r3}[][]{$r_3$}
\psfrag{r5}[][]{$r_5$}
\psfrag{w3,1}[][]{$w_{3,1}$}
\psfrag{w5,1}[][]{$w_{5,1}$}
\psfrag{w5,2}[][]{$w_{5,2}$}
\includegraphics[width=0.45\textwidth]{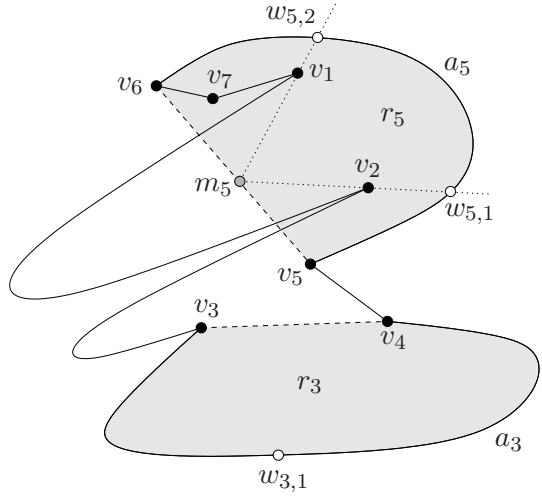}
\caption{The auxiliary vertices (white points) for rooms $r_3$ (empty)
  and $r_5$ (non-empty). $w_{3,1}$ is a point in the interior of $a_3$.
  $m_5$ is the midpoint of $v_5$ and $v_6$, whereas $w_{5,1}$ and
  $w_{5,2}$ are the intersections of the lines $m_5v_2$ and $m_5v_1$
  with the arc $a_5$, respectively. In this example
  $R_5=\{v_1,v_2,v_7\}$, whereas $C_5^*=\{v_1,v_2\}$.}
\label{fig:auxiliaryvertices}
\end{center}
\end{figure}

More specifically, for each empty room $r_i$ we add a vertex $w_{i,1}$
(anywhere) in the interior of the arc $a_i$ (see
Fig. \ref{fig:auxiliaryvertices}).
For each non-empty room $r_i$, let $X_i$ be the set of vertices of $P$
that lie in the interior of the chord $v_iv_{i+1}$ of $r_i$, and $R_i$
be the set of vertices of $P$ that are contained in the interior of
$r_i$ or belong to $X_i$ (by assumption $R_i\ne\emptyset$). If
$R_i\ne{}X_i$, let $C_i$ be the set of vertices on the convex hull of
the vertex set $(R_i\setminus{}X_i)\cup\{v_i,v_{i+1}\}$; if $R_i=X_i$,
let $C_i=X_i\cup\{v_i,v_{i+1}\}$. Finally, let
$C_i^*=C_i\setminus\{v_i,v_{i+1}\}$. Clearly, $v_i$ and $v_{i+1}$
belong to the set $C_i$ and, furthermore, $C_i^*\ne\emptyset$.

Let $m_i$ be the midpoint of $v_iv_{i+1}$ and $\ell_i^\perp(p)$ the
line perpendicular to $v_iv_{i+1}$ passing through a point $p$. If
$C_i^*\ne{}X_i$, then, for each $v_k\in{}C_i^*$, let $w_{i,j_k}$,
$1\le{}j_k\le{}|C_i^*|$, be the (unique) intersection of the line
$m_iv_k$ with the arc $a_i$; if $C_i^*=X_i$, then, for each
$v_k\in{}C_i^*$, let $w_{i,j_k}$, $1\le{}j_k\le{}|C_i^*|$, be the
(unique) intersection of the line $\ell_i^\perp(v_k)$ with the arc
$a_i$.

Now consider the sequence $\polyapprox{S}$ of the original vertices of
$P$ augmented by the auxiliary vertices added to empty and non-empty
rooms; the order of the vertices in $\polyapprox{S}$ is the order in
which we encounter them as we traverse the boundary of $P$ in the
counterclockwise order. The linear polygon defined by the sequence
$\polyapprox{S}$ of vertices is denoted by $\polyapprox{P}$ (see
Fig. \ref{fig:polygonalapprox}). It is easy to show that:

\begin{figure}[t]
\begin{center}
\psfrag{v1}[][]{$v_1$}
\psfrag{v2}[][]{$v_2$}
\psfrag{v3}[][]{$v_3$}
\psfrag{v4}[][]{$v_4$}
\psfrag{v5}[][]{$v_5$}
\psfrag{v6}[][]{$v_6$}
\psfrag{v7}[][]{$v_7$}
\psfrag{w1,1}[][]{$w_{1,1}$}
\psfrag{w2,1}[][]{$w_{2,1}$}
\psfrag{w3,1}[][]{$w_{3,1}$}
\psfrag{w5,1}[][]{$w_{5,1}$}
\psfrag{w5,2}[][]{$w_{5,2}$}
\subfigure[\label{fig:polygonalapprox}]
{\includegraphics[width=0.47\textwidth]{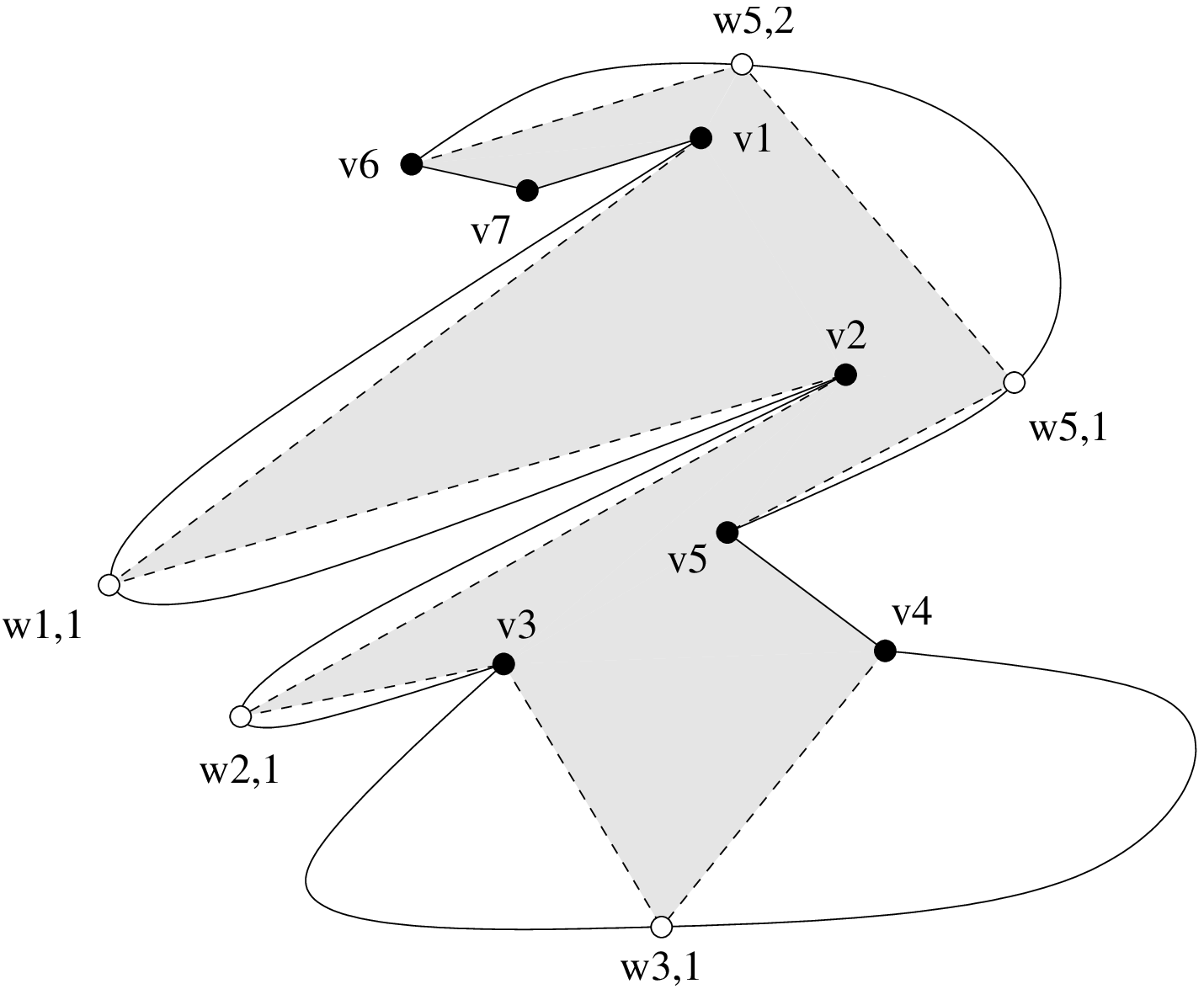}}\hfill%
\subfigure[\label{fig:triangulatedapprox}]
{\includegraphics[width=0.47\textwidth]{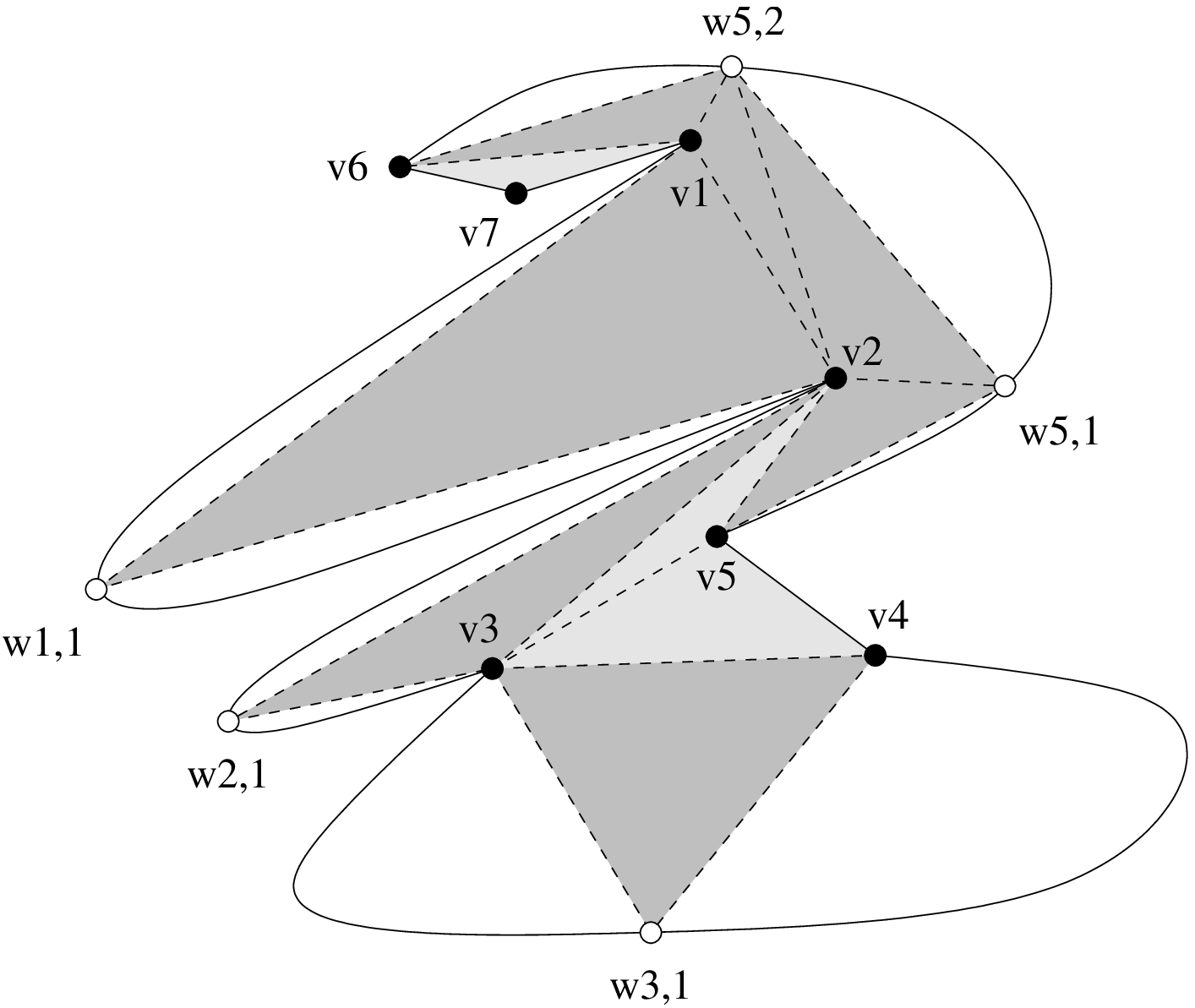}}
\caption{(a) The polygonal approximation $\polyapprox{P}$, shown in gray,
  of the \pconvex polygon $P$ with vertices $v_i$,
  $i=1,\ldots,7$. (b) The constrained triangulation
  $\tr{\polyapprox{P}}$ of $\polyapprox{P}$. The dark gray triangles
  are the constrained triangles.
  The polygonal region $v_5w_{5,1}w_{5,2}v_6v_1v_2v_5$ is a
  crescent. The triangles $w_{5,1}v_2v_5$ and  $v_1w_{5,2}v_6$ are
  boundary crescent triangles. The triangle $v_2w_{5,2}v_1$ is an
  upper crescent triangle, whereas the triangle $v_2w_{5,1}w_{5,2}$ is
  a lower crescent triangle.}
\label{fig:triangulatedpolygonalapprox}
\end{center}
\end{figure}

\begin{lemma}
The linear polygon $\polyapprox{P}$ is a simple polygon.
\end{lemma}

\begin{proof}
It suffices show that the line segments replacing the curvilinear segments
of $P$ do not intersect other edges of $P$ or $\polyapprox{P}$.

Let $r_i$ be an empty room, and let $w_{i,1}$ be the point added in the
interior of $a_i$. The interior of the line segments $v_iw_{i,1}$ and
$w_{i,1}v_{i+1}$ lie in the interior of $r_i$. Since $P$ is a
\pconvex polygon, and $r_i$ is an empty room, no edge of $P$
could potentially intersect $v_iw_{i,1}$ or $w_{i,1}v_{i+1}$. Hence
replacing $a_i$ by the polyline $v_iw_{i,1}v_{i+1}$ gives us a new
\pconvex polygon.

Let $r_i$ be a non-empty room. Let $w_{i,1},\ldots,w_{i,K_i}$ be the
points added on $a_i$, where $K_i$ is the cardinality of $C_i^*$. By
construction, every point $w_{i,k}$ is visible from $w_{i,k+1}$,
$k=1,\ldots{}K_i-1$, and every point $w_{i,k}$ is visible from
$w_{i,k-1}$, $k=2,\ldots{}K_i$. Moreover, $w_{i,1}$ is visible from
$v_i$ and $w_{i,K_i}$ is visible from $v_{i+1}$. Therefore, the interior
of the segments in the polyline $v_iw_{i,1}\ldots{}w_{i,K_i}v_{i+1}$ lie
in the interior of $r_i$ and do not intersect any arc in $P$. Hence,
substituting $a_i$ by the polyline $v_iw_{i,1}\ldots{}w_{i,K_i}v_{i+1}$
gives us a new \pconvex polygon.

As a result, the linear polygon $\polyapprox{P}$ is a simple
polygon.\proofbox
\end{proof}

We call the linear polygon $\polyapprox{P}$, defined by
$\polyapprox{S}$, the \emph{straight-line polygonal approximation} of
$P$, or simply the \emph{polygonal approximation} of $P$. An obvious
result for $\polyapprox{P}$ is the following:

\begin{corollary}\label{cor:inclusion}
If $P$ is a \pconvex polygon the polygonal approximation
$\polyapprox{P}$ of $P$ is a linear polygon that is contained inside
$P$.
\end{corollary}

We end this section by proving a tight upper bound on the size of the
polygonal approximation of a \pconvex polygon. We start by
stating and proving an intermediate result, namely that the sets
$C_i^*$ are pairwise disjoint.

\begin{lemma}\label{lem:cistardisjoint}
Let $i$, $j$, with $1\le{}i<j\le{}n$. Then $C_i^*\cap{}C_j^*=\emptyset$.
\end{lemma}

\renewcommand{\figwidth}{0.3\textwidth}
\newcommand{\figtextsize}{\scriptsize}

\begin{proof}
If one of the rooms $r_i$ and $r_j$ is a degenerate or an empty room,
the result is obvious.

Consider two non-empty rooms $r_i$ and $r_j$. For simplicity of
presentation we assume that $R_i\ne{}X_i$ and $R_j\ne{}X_j$; the
proof easily carries on to the case $R_i=X_i$ or $R_j=X_j$.

Suppose that there exists a vertex $u\in{}P$ that is contained in
$C_i^*\cap{}C_j^*$. Let $v_i$, $v_{i+1}$, and $v_j$, $v_{j+1}$ be the
endpoints of the arcs $a_i$ and $a_j$, and $m_i$, $m_j$ the midpoints
of the chords $v_iv_{i+1}$, $v_jv_{j+1}$, respectively. Let $u_i$ be
the intersection of the line $m_iu$ with the convex arc $a_i$ and
$u_j$ be the intersection of the line $m_ju$ with the convex arc $a_j$,
respectively. Consider the following cases.
\begin{mathdescription}
\item[$v_j,v_{j+1}\not\in{}R_i,v_i,v_{i+1}\not\in{}R_j$.]
  This is the easy case (see Fig. \ref{fig:proof-lem3:case1}). Since
  $u\in{}C_i^*\cap{}C_j^*$ we have that
  $r_i\cap{}r_j\ne\emptyset$. Moreover, it is either the case that
  $a_j$ intersects the chord $v_iv_{i+1}$ or $a_i$ intersects the
  chord $v_jv_{j+1}$. Without loss of generality we can assume that
  $a_j$ intersects the chord $v_iv_{i+1}$. In this case the boundary
  of $r_i\cap{}r_j$ that lies in the interior of $r_i$ is a subarc of
  $a_j$. But then the segment $uu_i$ has to intersect $a_j$, which
  contradicts the fact that $u\in{}C_i^*$.
  \begin{figure}[b]
    \begin{center}
      \psfrag{ai}[][]{\figtextsize$a_i$}
      \psfrag{aj}[][]{\figtextsize$a_j$}
      \psfrag{vi}[][]{\figtextsize$v_i$}
      \psfrag{vi+1}[][]{\figtextsize$v_{i+1}$}
      \psfrag{vj}[][]{\figtextsize$v_j$}
      \psfrag{vj+1}[][]{\figtextsize$v_{j+1}$}
      \psfrag{u}[][]{\figtextsize$u$}
      \psfrag{ui}[][]{\figtextsize$u_i$}
      \psfrag{mi}[][]{\figtextsize$m_i$}
      \includegraphics[width=\figwidth]{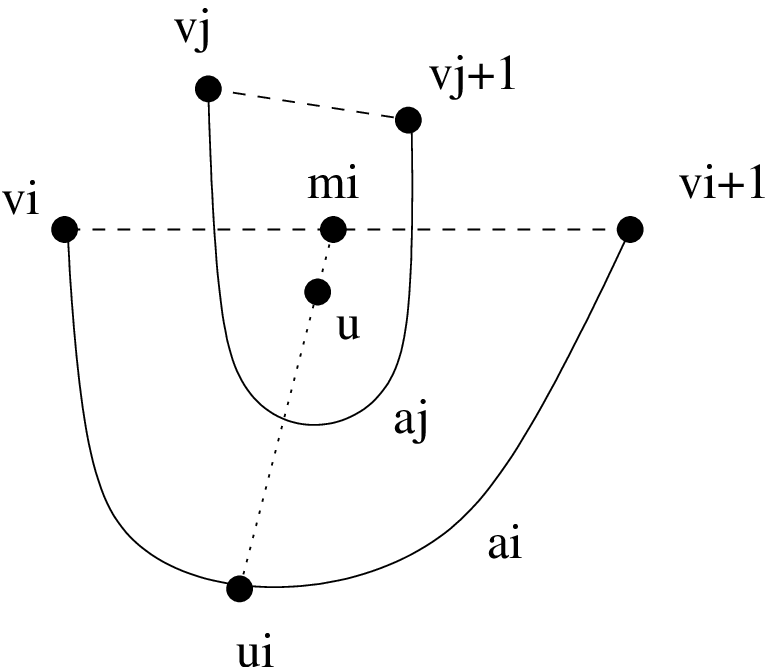}
    \end{center}
    \caption{Proof of Lemma \ref{lem:cistardisjoint}. The case
      $v_j,v_{j+1}\not\in{}R_i,v_i,v_{i+1}\not\in{}R_j$.}
    \label{fig:proof-lem3:case1}
  \end{figure}
\item[$v_j,v_{j+1}\in{}R_i$.]
  Since $u$ belongs to $C_i^*$, the line segment $uu_i$ cannot contain
  any vertices of $P$ and it cannot intersect any edge of $P$ (since
  otherwise $u$ would not belong to $C_i^*$). For this reason, and
  since $u$ belongs to $C_j^*$, $uu_i$ has to intersect the chord of
  $r_j$. We distinguish between the following two cases (see
  Fig. \ref{fig:proof-lem3:case2}):
  \begin{enumerate}
    \item \emph{The chord $v_jv_{j+1}$ intersects the interior of $uu_i$}.
      Depending on whether the supporting line of
      $v_jv_{j+1}$ intersects the chord $v_iv_{i+1}$ of $r_i$ or not,
      $u$ will be either contained in the interior of one of the
      triangles $v_iv_{i+1}v_j$ and $v_iv_{i+1}v_{j+1}$ (this happens
      if the supporting line of $v_jv_{j+1}$ intersects $v_iv_{i+1}$
      --- see Fig. \ref{fig:proof-lem3:case2:a}),
      or inside the convex quadrilateral $v_iv_{i+1}v_jv_{j+1}$ (this
      happens if the supporting line of $v_jv_{j+1}$ does not
      intersect $v_iv_{i+1}$ --- see Fig.
      \ref{fig:proof-lem3:case2:b}). In either case, $u$ is
      in the interior of a convex polygon, the vertices of which are
      in $R_i\cup\{v_i,v_{i+1}\}$, and, thus, it cannot belong to
      $C_i^*$, hence a contradiction.
      \begin{figure}[t]
        \begin{center}
          \psfrag{ai}[][]{\figtextsize$a_i$}
          \psfrag{aj}[][]{\figtextsize$a_j$}
          \psfrag{vi}[][]{\figtextsize$v_i$}
          \psfrag{vi+1}[][]{\figtextsize$v_{i+1}$}
          \psfrag{vj}[][]{\figtextsize$v_j$}
          \psfrag{vjp}[][]{\figtextsize$v_j'$}
          \psfrag{vj+1}[][]{\figtextsize$v_{j+1}$}
          \psfrag{vj+1p}[][]{\figtextsize$v_{j+1}'$}
          \psfrag{u}[][]{\figtextsize$u$}
          \psfrag{x}[][]{\figtextsize$x$}
          \psfrag{ui}[][]{\figtextsize$u_i$}
          \psfrag{mi}[][]{\figtextsize$m_i$}
          \subfigure[\label{fig:proof-lem3:case2:a}]%
          {\includegraphics[width=\figwidth]{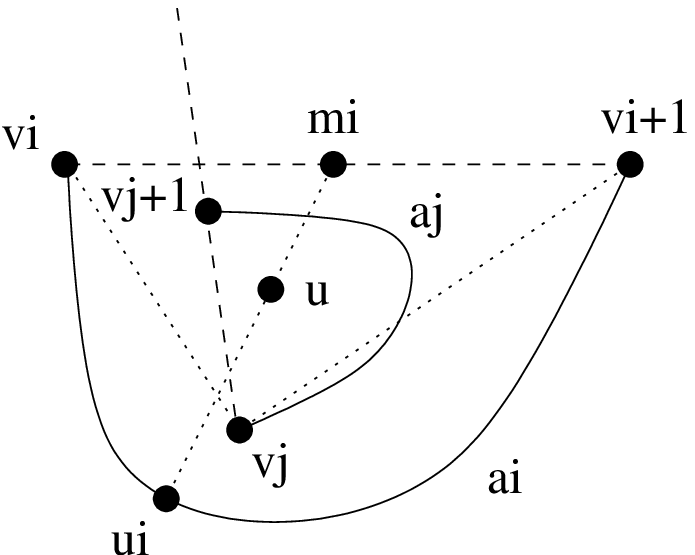}}%
          \hfill%
          \subfigure[\label{fig:proof-lem3:case2:b}]%
          {\includegraphics[width=\figwidth]{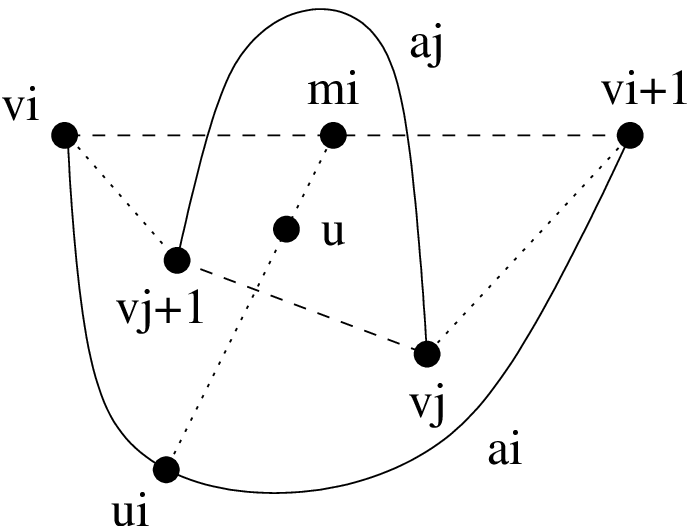}}%
          \hfill%
          \subfigure[\label{fig:proof-lem3:case2:c}]%
          {\includegraphics[width=\figwidth]{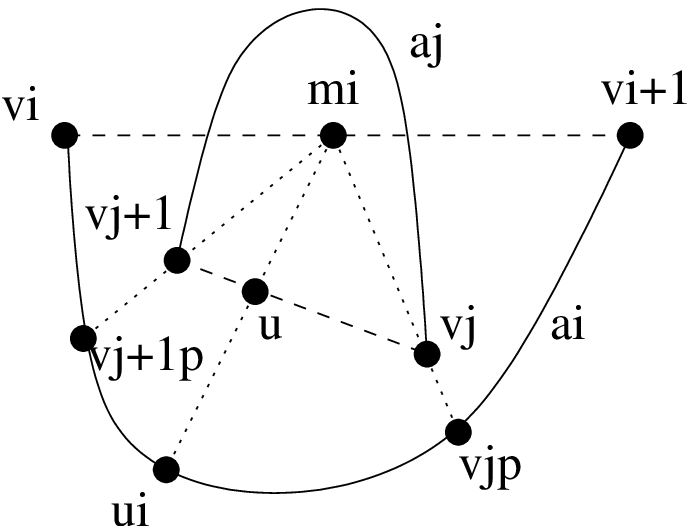}}
        \end{center}
        \caption{Proof of Lemma \ref{lem:cistardisjoint}. The case
          $v_j,v_{j+1}\in{}R_i$. (a) the chord $v_jv_{j+1}$ intersects
          the interior of $uu_i$ and $u$ is contained inside the
          triangle $v_iv_{i+1}v_j$. (b) the chord $v_jv_{j+1}$
          intersects the interior of $uu_i$ and $u$ is contained
          inside the convex quadrilateral $v_iv_{i+1}v_jv_{j+1}$.
          (c) the chord $v_jv_{j+1}$ intersects $uu_i$ at $u$.}
        \label{fig:proof-lem3:case2}
      \end{figure}
    \item \emph{The chord $v_jv_{j+1}$ intersects $uu_i$ at $u$}.
      We can assume without loss of generality that $v_{i+1}$, $v_j$
      are to the right and $v_i$, $v_{j+1}$ to the left of the
      oriented line $u_iu$ (see Fig. \ref{fig:proof-lem3:case2:c}). 
      Notice that both $v_j$ and $v_{j+1}$ have 
      to belong to $C_i^*$, since otherwise $u$ would not belong to
      $C_i^*$. Let $v_j'$ and $v_{j+1}'$ be the intersections of the
      lines $m_iv_j$ and $m_iv_{j+1}$ with $a_i$. Consider the path
      $\pi$ from $u$ to $v_i$ on the boundary $\partial{}P$ of
      $P$, that does not contain the edge $a_j$. $\pi$ has to
      intersect either the interior of the line segment $v_jv_j'$ or
      the interior of the line segment $v_{j+1}v_{j+1}'$; either case
      yields a contradiction with the fact that both $v_j$ and
      $v_{j+1}$ belong to $C_i^*$.
  \end{enumerate}
\item[$v_i,v_{i+1}\in{}R_j$.]
  This case is symmetric to the previous one.
\item[$|\{v_j,v_{j+1}\}\cap{}R_i|=1$.]
  Without loss of generality we may assume that $v_j\in{}R_i$ and
  $v_{j+1}\not\in{}R_i$. Consider the following two cases (see
  Fig. \ref{fig:proof-lem3:case4}):
  \begin{enumerate}
  \item \emph{The chord $v_jv_{j+1}$ intersects the chord $v_iv_{i+1}$}.
    If $v_jv_{j+1}$ intersects the interior of $v_iv_{i+1}$ (see
    Fig. \ref{fig:proof-lem3:case4:a}), then $u$
    has to lie in the interior of the triangle $v_iv_{i+1}v_j$, which
    contradicts the fact that $u\in{}C_i^*$.

    Suppose now that $v_jv_{j+1}$ intersects one of the endpoints of
    $v_iv_{i+1}$, and let us assume that this endpoint is $v_i$ (see
    Fig. \ref{fig:proof-lem3:case4:b}). $u$
    has to lie in the interior of $v_iv_j$, since otherwise it would
    have been in the interior of the triangle $v_iv_{i+1}v_j$, which
    contradicts the fact that $u\in{}C_i^*$. Moreover, $v_i$ (\resp
    $v_j$) has to belong to $R_j$ (\resp $R_i$), since otherwise
    $u\not\in{}C_j^*$ (\resp $u\not\in{}C_i^*$). Let $v_j'$ be the
    intersection of $m_iv_j$ with $a_i$ and $v_i'$ be the intersection
    with $a_j$ of the line perpendicular to $v_jv_{j+1}$ at $v_i$.
    Consider the paths $\pi_1$ and $\pi_2$ on $\partial{}P$ from $u$
    to $v_{i+1}$ and $v_{j+1}$, respectively. One of these two
    paths has to intersect either the interior of the line segment
    $v_iv_i'$ or the interior of line segment $v_jv_j'$; either case
    yields a contradiction with the fact that $v_i$ belongs to $C_j^*$
    and $v_j$ belongs to $C_i^*$.
    \begin{figure}[t]
      \begin{center}
        \psfrag{ai}[][]{\figtextsize$a_i$}
        \psfrag{aj}[][]{\figtextsize$a_j$}
        \psfrag{vi}[][]{\figtextsize$v_i$}
        \psfrag{vi+1}[][]{\figtextsize$v_{i+1}$}
        \psfrag{vj}[][]{\figtextsize$v_j$}
        \psfrag{vj+1}[][]{\figtextsize$v_{j+1}$}
        \psfrag{vip}[][]{\figtextsize$v_i'$}
        \psfrag{vi+1p}[][]{\figtextsize$v_{i+1}'$}
        \psfrag{vjp}[][]{\figtextsize$v_j'$}
        \psfrag{vj+1p}[][]{\figtextsize$v_{j+1}'$}
        \psfrag{u}[][]{\figtextsize$u$}
        \psfrag{ui}[][]{\figtextsize$u_i$}
        \psfrag{mi}[][]{\figtextsize$m_i$}
        \psfrag{mj}[][]{\figtextsize$m_j$}
        \subfigure[\label{fig:proof-lem3:case4:a}]%
        {\includegraphics[width=0.27\textwidth]{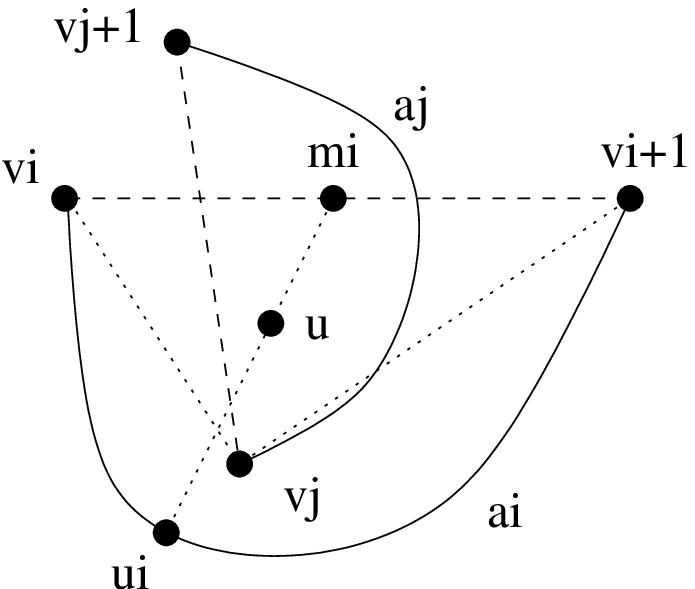}}%
        \hfill%
        \subfigure[\label{fig:proof-lem3:case4:b}]%
        {\includegraphics[width=0.27\textwidth]{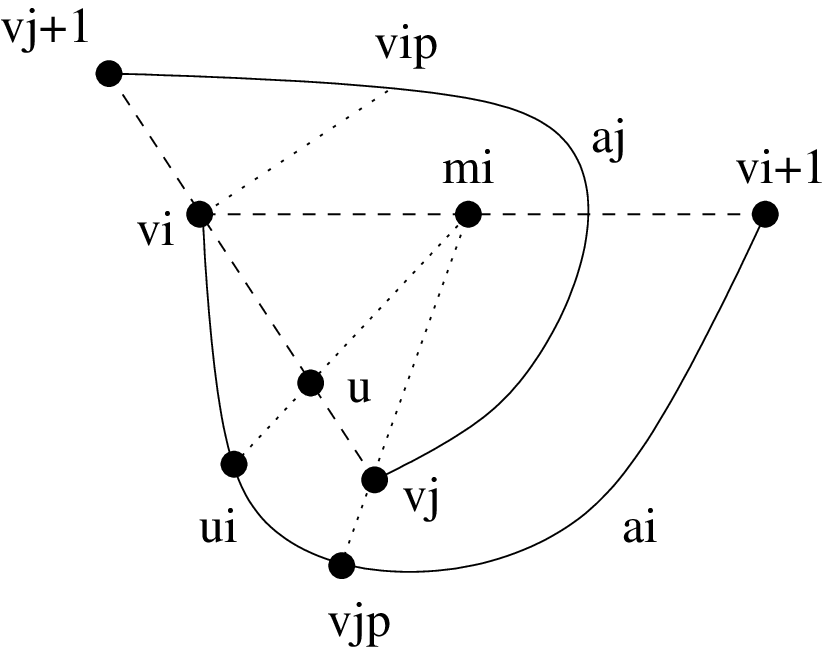}}%
        \hfill%
        \subfigure[\label{fig:proof-lem3:case4:c}]%
        {\includegraphics[width=0.35\textwidth]{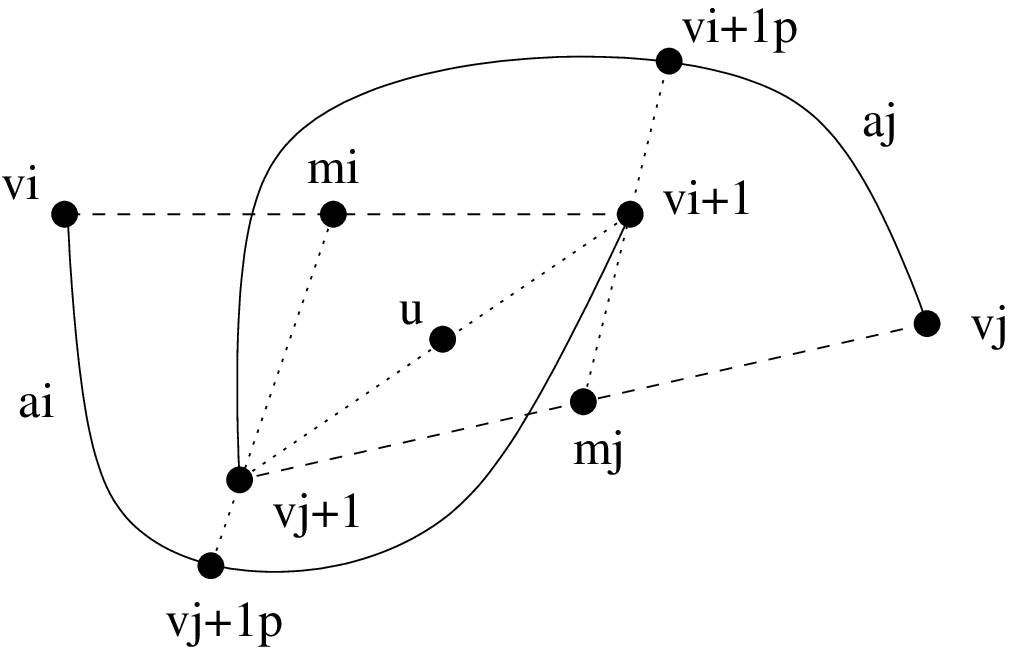}}
      \end{center}
      \caption{Proof of Lemma \ref{lem:cistardisjoint}. The case
        $|\{v_j,v_{j+1}\}\cap{}R_i|=1$. (a) the chord $v_jv_{j+1}$
        intersects the chord $v_iv_{i+1}$ and $v_jv_{j+1}$ intersects
        the interior of $v_iv_{i+1}$. (b) the chord $v_jv_{j+1}$
        intersects the chord $v_iv_{i+1}$ and $v_jv_{j+1}$ intersects
        $v_iv_{i+1}$ at $v_i$. (c) the chord $v_jv_{j+1}$
        intersects $a_i$.}
      \label{fig:proof-lem3:case4}
    \end{figure}
  \item \emph{The chord $v_jv_{j+1}$ intersects the edge $a_i$}.
    In this case we also have that either $v_i\in{}R_j$ or
    $v_{i+1}\in{}R_j$, but not both. Without loss of generality we
    may assume that $v_{i+1}\in{}R_j$ (see Fig. 
    \ref{fig:proof-lem3:case4:c}). Since $u$ belongs to both
    $C_i^*$ and $C_j^*$, it has to lie on the line segment
    $v_{i+1}v_{j+1}$. Moreover, $v_{j+1}$ (\resp $v_{i+1}$) has to belong to
    $C_i^*$ (\resp $C_j^*$), since otherwise $u$ would not belong to
    $C_i^*$ (\resp $C_j^*$). Let $v_{i+1}'$ and $v_{j+1}'$ be the
    intersections of the lines $m_jv_{i+1}$ and $m_iv_{j+1}$ with the arcs
    $a_j$ and $a_i$, respectively. Consider the paths $\pi_1$ and
    $\pi_2$ on $\partial{}P$ from $u$ to $v_i$ and $v_j$,
    respectively. One of these two paths has to intersect either the
    interior of the line segment $v_{i+1}v_{i+1}'$ or the interior of
    the line segment $v_{j+1}v_{j+1}'$. In the former case, we get a
    contradiction with the fact that $v_{i+1}$ belongs to $C_j^*$; in
    the latter case we get a contradiction with the fact that $v_{j+1}$
    belongs to $C_i^*$.
  \end{enumerate}
\item[$|\{v_i,v_{i+1}\}\cap{}R_j|=1$.]
  This case is symmetric to the previous one.\proofbox
\end{mathdescription}
\end{proof}

An immediate consequence of Lemma \ref{lem:cistardisjoint} is the
following corollary that gives us a tight bound on the size of the
polygonal approximation $\polyapprox{P}$ of $P$.

\begin{corollary}\label{cor:polyapproxsize}
If $n$ is the size of a \pconvex polygon $P$, the size of its
polygonal approximation $\polyapprox{P}$ is at most $3n$. This bound
is tight (up to a constant).
\end{corollary}

\begin{proof}
Let $a_i$ be a convex arc of $P$, and let $r_i$ be the corresponding
room. If $a_i$ is an empty room, then $\polyapprox{P}$ contains one
auxiliary vertex due to $a_i$. Hence $\polyapprox{P}$ contains at most
$n$ auxiliary vertices attributed to empty rooms in $P$.
If $a_i$ is a non-empty room, then $\polyapprox{P}$ contains $|C_i^*|$
auxiliary vertices due to $a_i$. By Lemma \ref{lem:cistardisjoint} the
sets $C_i^*$, $i=1,\ldots,n$ are pairwise disjoint, which implies that
$\sum_{i=1}^n|C_i^*|\le|P|=n$.
Therefore $\polyapprox{P}$ contains the $n$ vertices of $P$, contains
at most $n$ vertices due to empty rooms in $P$ and at most $n$
vertices due to non-empty rooms in $P$. We thus conclude that the size
of $\polyapprox{P}$ is at most $3n$.

The upper bound of the paragraph above is tight up to a
constant. Consider the \pconvex polygon $P$ of
Fig. \ref{fig:polyapproxlb}. It consists of $n-1$ empty rooms and one
non-empty room $r_1$, such that $|C_1^*|=n-2$. It is easy to see that
$|\polyapprox{P}|=3n-3$.\proofbox
\end{proof}

\begin{figure}[t]
\begin{center}
\psfrag{m1}[][]{\small$m_1$}
\psfrag{v1}[][]{\small$v_1$}
\psfrag{v2}[][]{\small$v_2$}
\psfrag{v3}[][]{\small$v_3$}
\psfrag{v4}[][]{\small$v_4$}
\psfrag{v5}[][]{\small$v_5$}
\psfrag{v6}[][]{\small$v_6$}
\psfrag{vn-3}[][]{\small$v_{n-3}$}
\psfrag{vn-2}[][]{\small$v_{n-2}$}
\psfrag{vn-1}[][]{\small$v_{n-1}$}
\psfrag{vn}[][]{\small$v_n$}
\includegraphics[width=\textwidth]{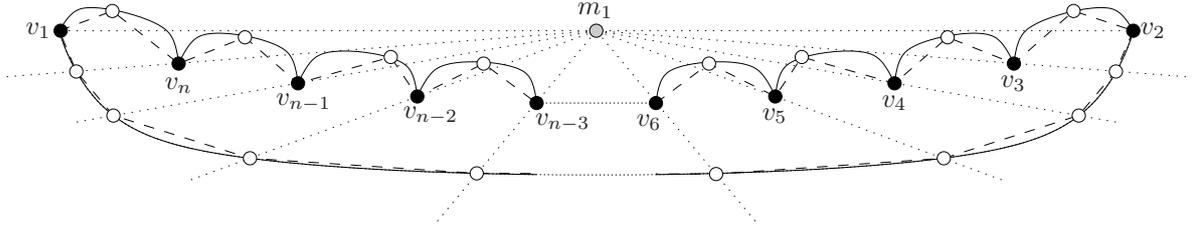}
\caption{A \pconvex polygon $P$ of size $n$ (solid
  curve), the polygonal approximation $\polyapprox{P}$ of which
  consists of $3n-3$ vertices (dashed polyline).}
\label{fig:polyapproxlb}
\end{center}
\end{figure}


\subsection{Triangulating the polygonal approximation}
\label{sec:triangulation}

Let $P$ be a \pconvex polygon and $\polyapprox{P}$ is its
polygonal approximation. We are going to construct a
\emph{constrained triangulation} of $\polyapprox{P}$, \ie we are
going to triangulate $\polyapprox{P}$, while enforcing some triangles
to be part of this triangulation. Let
$\compl{P}=\polyapprox{P}\setminus{}P$ be the set of auxiliary
vertices in $\polyapprox{P}$. The main idea behind the way this
particular triangulation is constructed is to enforce that:
\begin{enumerate}
\item all triangles of $\tr{\polyapprox{P}}$, that contain a vertex in
  $\compl{P}$, also contain at least one vertex of $P$, \ie
  no triangles contain only auxiliary vertices,
\item every vertex in $\compl{P}$ belongs to at least one triangle in
  $\tr{\polyapprox{P}}$ the other two vertices of which are both
  vertices of $P$, and
\item the triangles of $\tr{\polyapprox{P}}$ that contain vertices of
  $\polyapprox{P}$ can be guarded by vertices of $P$.
\end{enumerate}
These properties are going to be exploited in Step
\ref{algo:compg} of the algorithm presented in Section
\ref{sec:piececonvex}.

More precisely, we are going to enforce the way the triangles of
$\tr{\polyapprox{P}}$ are created in the neighborhoods of the vertices
in $\compl{P}$. By enforcing the triangles in these neighborhoods, we
effectively triangulate parts of $\polyapprox{P}$. The remaining
untriangulated parts of $\polyapprox{P}$ consist of one of more
disjoint polygons, which can then be  triangulated by means of any 
$O(n\log{}n)$ polygon triangulation algorithm. In other words, the
triangulation of $\polyapprox{P}$ that we want to construct is a
constrained triangulation, in the sense that we pre-specify some of
the edges of the triangulation. In fact, as we will see below we
pre-specify triangles, rather than edges, which are going to be
referred to as \emph{constrained triangles}.

Let us proceed to define the constrained triangles in
$\tr{\polyapprox{P}}$. If $r_i$ is an empty room, and $w_{i,1}$ is the
point added on $a_i$, add the edges $v_iv_{i+1}$, $v_iw_{i,1}$ and
$w_{i,1}v_{i+1}$, thus formulating the constrained triangle
$v_iw_{i,1}v_{i+1}$ (see Fig. \ref{fig:triangulatedapprox}).
If $r_i$ is a non-empty room, $\{c_1,\ldots,c_{K_i}\}$ the vertices in
$C_i^*$, $K_i=|C_i^*|$, and $\{w_{i,1},\ldots,w_{i,K_i}\}$ the
vertices added on $a_i$, add the following edges, if they do not
already exist:
\begin{enumerate}
\item $c_k,c_{k+1}$, $k=1,\ldots,K_i-1$; $v_ic_1$; $c_{K_i}v_{i+1}$;
\item $c_iw_{i,k}$, $k=1,\ldots,K_i$;
\item $c_iw_{i,k+1}$, $k=1,\ldots,K_i-1$;
\item $w_{i,k},w_{i,k+1}$, $k=1,\ldots,K_i-1$; $v_iw_{i,1}$; $w_{i,K_i}v_{i+1}$.
\end{enumerate}
These edges formulate $2K_i$ constrained triangles, namely,
$c_kc_{k+1}w_{i,k+1}$, $k=1,\ldots,K_i-1$, $c_kw_{i,k}w_{i,k+1}$,
$k=1,\ldots,K_i-1$, $v_ic_1w_{i,1}$ and $v_{i+1}c_{K_i}w_{i,K_i}$. We call the
polygonal region delimited by these triangles a \emph{crescent}. The
triangles $v_ic_1w_{i,1}$ and $v_{i+1}c_{K_i}w_{i,K_i}$ are called
\emph{boundary crescent triangles}, the triangles $c_kc_{k+1}w_{i,k+1}$,
$k=1,\ldots,K_i-1$ are called \emph{upper crescent triangles} and the
triangles $c_kw_{i,k}w_{i,k+1}$, $k=1,\ldots,K_i-1$ are called
\emph{lower crescent triangles}.

Note that almost all points in $\compl{P}$ belong to exactly one
triangle the other two points of which are in $P$; the only exception
are the points $w_{i,K_i}$ which belong to exactly two such
triangles.

As we have already mentioned, having created the constrained triangles
mentioned above, there may exist additional possibly disjoint
polygonal non-triangulated regions of $\polyapprox{P}$. The
triangulation procedure continues by triangulating these additional
polygonal non-triangulated regions; any $O(n\log{}n)$ polygon
triangulation algorithm may be used.


\subsection{Computing a guarding set for the original polygon}
\label{sec:guardingset}

To compute a guarding set for $P$ we will perform the following two steps:
\begin{enumerate}
\item
  Compute a guarding set $G_{\polyapprox{P}}$ for $\polyapprox{P}$.
\item
  From the guarding set $G_{\polyapprox{P}}$ for $\polyapprox{P}$
  compute a guarding set $G_P$ for $P$ of size at most
  $\lfloor\frac{2n}{3}\rfloor$, consisting of vertices of $P$ only.
\end{enumerate}

Assume that we have colored the vertices of $\polyapprox{P}$ with
three colors, so that every triangle in $\tr{\polyapprox{P}}$ does not
contain two vertices of the same color. This can be easily done by the
standard three-coloring algorithm for linear polygons presented in
\cite{m-phe-75,f-spcwt-78}. Let red, green and blue be the three
colors, and let $K_A$ be the set of vertices of red color, $\Pi_A$ be
the set of vertices of green color and $M_A$ be the set of vertices of
blue color in a subset $A$ of $\polyapprox{P}$. 
Clearly, all three sets $K_{\polyapprox{P}}$, $\Pi_{\polyapprox{P}}$ and
$M_{\polyapprox{P}}$ are guarding sets for $\polyapprox{P}$. In fact,
they are also guarding sets for $P$, as the following theorem suggests
(see also Fig. \ref{fig:vertexcoloring}).

\begin{figure}[t]
\begin{center}
\psfrag{v1}[][]{$v_1$}
\psfrag{v2}[][]{$v_2$}
\psfrag{v3}[][]{$v_3$}
\psfrag{v4}[][]{$v_4$}
\psfrag{v5}[][]{$v_5$}
\psfrag{v6}[][]{$v_6$}
\psfrag{v7}[][]{$v_7$}
\psfrag{w1,1}[][]{$w_{1,1}$}
\psfrag{w2,1}[][]{$w_{2,1}$}
\psfrag{w3,1}[][]{$w_{3,1}$}
\psfrag{w5,1}[][]{$w_{5,1}$}
\psfrag{w5,2}[][]{$w_{5,2}$}
\includegraphics[width=0.47\textwidth]{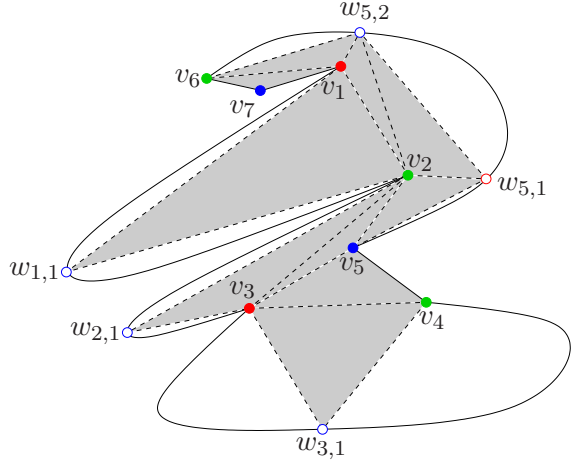}
\caption{The three guarding sets for $\polyapprox{P}$, are also
  guarding sets for $P$, as Theorem \ref{thm:gsapprox2orig} suggests.}
\label{fig:vertexcoloring}
\end{center}
\end{figure}

\begin{theorem}\label{thm:gsapprox2orig}
Each one of the sets $K_{\polyapprox{P}}$, $\Pi_{\polyapprox{P}}$ and
$M_{\polyapprox{P}}$ is a guarding set for $P$.
\end{theorem}

\begin{proof}
Let $G_{\polyapprox{P}}$ be one of $K_{\polyapprox{P}}$,
$\Pi_{\polyapprox{P}}$ and $M_{\polyapprox{P}}$. 
By construction, $G_{\polyapprox{P}}$ guards all triangles in
$\tr{\polyapprox{P}}$. To show that $G_{\polyapprox{P}}$ is a guarding
set for $P$, it suffices to show that $G_{\polyapprox{P}}$ also guards
the non-degenerate sectors defined by the edges of $\polyapprox{P}$
and the corresponding convex subarcs of $P$.

Let $s_i$ be a non-degenerate sector associated with the convex arc
$a_i$. We consider the following two cases:
\begin{enumerate}
\item
  The room $r_i$ is an empty room. Then $s_i$ is adjacent to the
  triangle $v_iw_{i,1}v_{i+1}$ of $\tr{\polyapprox{P}}$. Note that
  since $a_i$ is a convex arc, all three points $v_i$, $v_{i+1}$ and
  $w_{i,1}$ guard $s_i$. Since one of them has to be in
  $G_{\polyapprox{P}}$, we conclude that $G_{\polyapprox{P}}$ guards $s_i$.
\item
  The room $r_i$ is a non-empty room. Then $s_i$ is adjacent to either
  a boundary crescent triangle or a lower crescent triangle in
  $\tr{\polyapprox{P}}$ . Let $T$ be this triangle, and let $x$, $y$
  and $z$ be its vertices. Since $a_i$ is a convex arc, all three $x$,
  $y$ and $z$ guard $s_i$. Therefore, since one of the three vertices
  $x$, $y$ and $z$ is in $G_{\polyapprox{P}}$, we conclude that
  $G_{\polyapprox{P}}$ guards $s_i$.
\end{enumerate}
Therefore every non-degenerate sector in $\compl{P}$ is guarded by at
least one vertex in $G_{\polyapprox{P}}$, which implies that
$G_{\polyapprox{P}}$ is a guarding set for $P$.\proofbox
\end{proof}

Let as now assume, without loss of generality that, among $K_P$, $\Pi_P$
and $M_P$, $K_P$ has the smallest cardinality and that $\Pi_P$ has the
second smallest cardinality, \ie $|K_P|\le|\Pi_P|\le|M_P|$. We are
going to define a mapping $f$ from $K_{\compl{P}}$
to the power set $2^{\Pi_P}$ of $\Pi_P$. Intuitively, $f$ maps a
vertex $x$ in $K_{\compl{P}}$ to all the
neighboring vertices of $x$ in $\tr{\polyapprox{P}}$ that belong to
$\Pi_P$. We are going to give a more precise definition of $f$
below (consult Fig. \ref{fig:mapping}). Let
$x\in{}K_{\compl{P}}$. We distinguish between the
following cases: 
\begin{enumerate}
\item
  $x$ is an auxiliary vertex added to an empty room $r_i$ (see
  Fig. \ref{fig:mapempty}). Then $x$ is one of the vertices of the
  constrained triangle $v_iv_{i+1}x$ contained inside $r_i$. One of
  $v_i$, $v_{i+1}$ must be a vertex in $\Pi_P$, say $v_{i+1}$. Then we set
  $f(x)=\{v_{i+1}\}$.
\item
  $x$ is an auxiliary vertex added to a non-empty room $r_i$. Consider
  the following subcases:
\begin{enumerate}
\item
  $x$ is not the last auxiliary vertex on $a_i$, as we walk along
  $a_i$ in the counterclockwise sense (see
  Fig. \ref{fig:mapnotlast}). Then $x$ is incident to a single
  triangle in $\tr{\polyapprox{P}}$ the other two vertices of which
  are vertices in $P$. Let $y$ and $z$ be these other two
  vertices. One of $y$ and $z$ has to be a green vertex, say $y$. Then
  we set $f(x)=\{y\}$.
\item
  $x$ is the last auxiliary vertex on $a_i$ as we walk along $a_i$ in
  the counterclockwise sense (see Figs. \ref{fig:maplast1} and
  \ref{fig:maplast2}). Then $x$ is incident to a boundary crescent
  triangle and an upper crescent triangle. Let $xv_{i+1}y$ be the
  boundary crescent triangle and $xyz$ the upper crescent
  triangle. Clearly, all three vertices $v_{i+1}$, $y$ and $z$ are vertices
  of $P$. If $y\in{}\Pi_P$ (this is the case in
  Fig. \ref{fig:maplast1}), then we set $f(x)=\{y\}$. Otherwise (this is
  the case in Fig. \ref{fig:maplast2}), both $v_{i+1}$ and
  $z$ have to be green vertices, in which case we set $f(x)=\{v_{i+1},z\}$.
\end{enumerate}
\end{enumerate}

\begin{figure}[t]
\begin{center}
\psfrag{x}[][]{\small$x$}
\psfrag{y}[][]{\small$y$}
\psfrag{z}[][]{\small$z$}
\psfrag{vi}[][]{\small$v_i$}
\psfrag{vi+1}[][]{\small$v_{i+1}$}
\subfigure[\label{fig:mapempty}]
          {\includegraphics[width=0.4\textwidth]{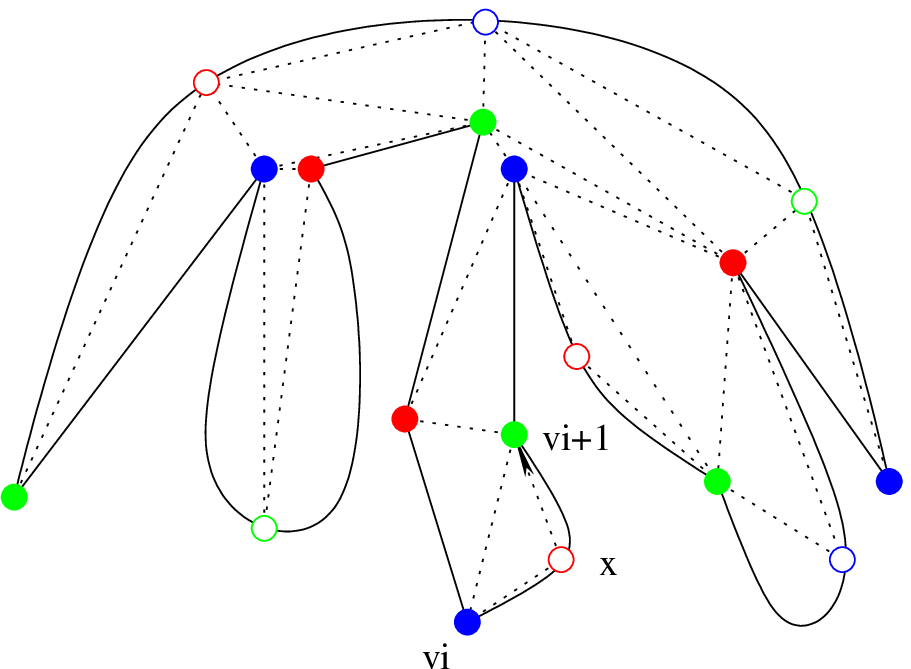}}\hfil%
\subfigure[\label{fig:mapnotlast}]
          {\includegraphics[width=0.4\textwidth]{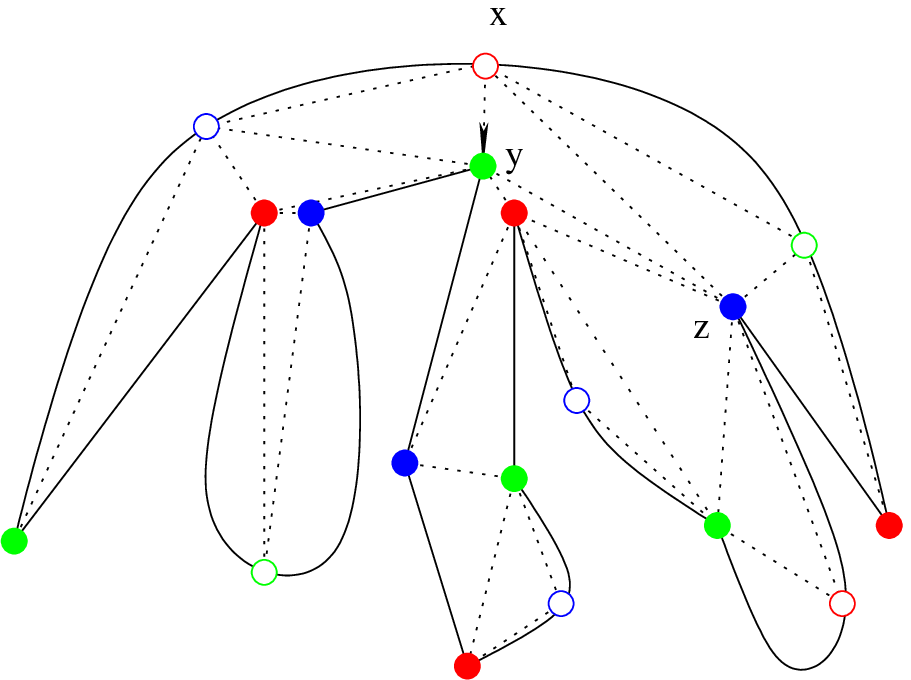}}
\subfigure[\label{fig:maplast1}]
          {\includegraphics[width=0.4\textwidth]{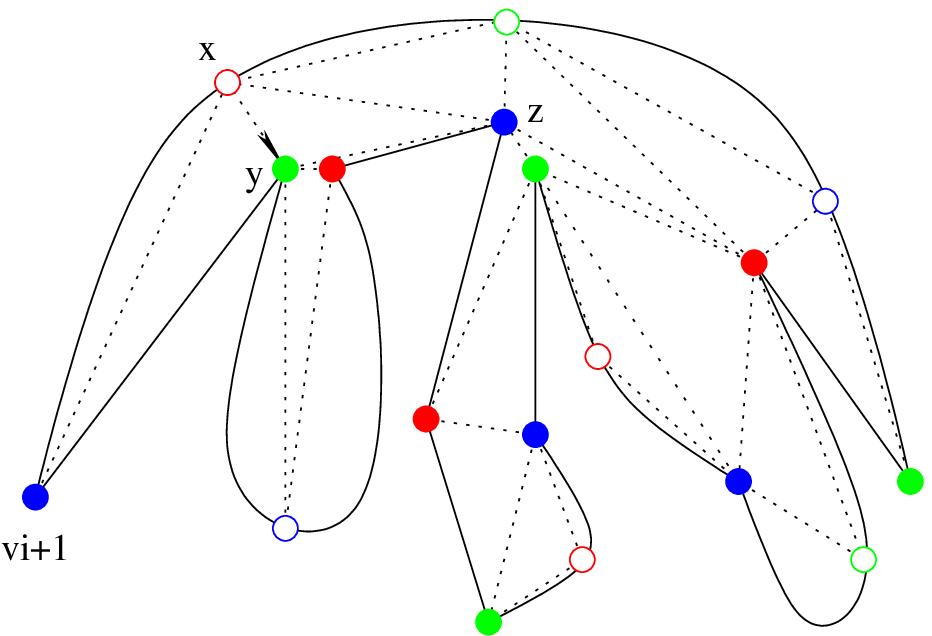}}\hfil%
\subfigure[\label{fig:maplast2}]
          {\includegraphics[width=0.4\textwidth]{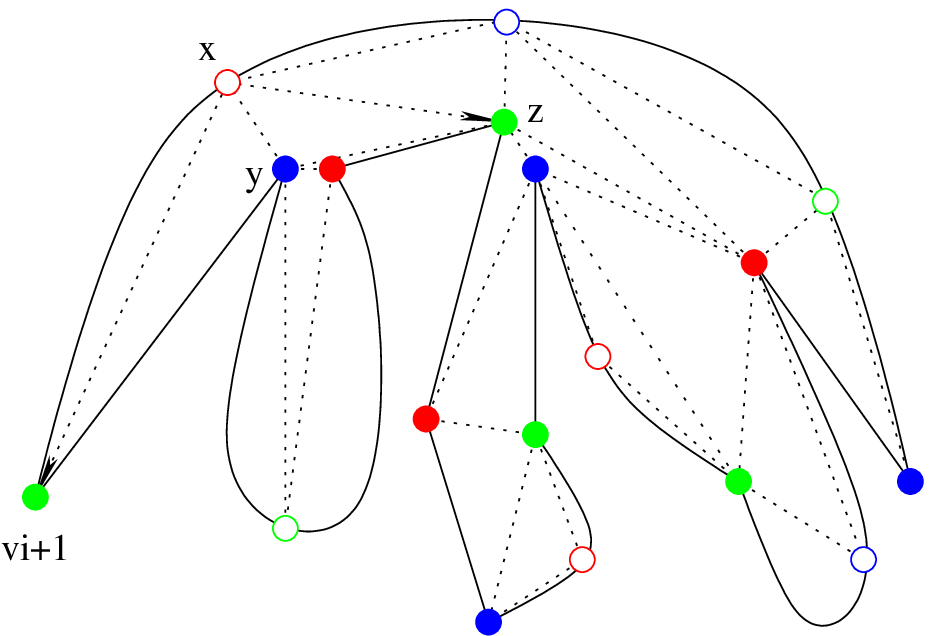}}
\caption{The three cases in the definition of the mapping $f$. Case
  (a): $x$ is a auxiliary vertex in an empty room. Case (b): $x$ is an
  auxiliary vertex in a non-empty room and is not the last auxiliary vertex
  added on the curvilinear arc. Cases (c) and (d): $x$ is the last
  auxiliary vertex added on the curvilinear arc of a non-empty room (in (c)
  only one of its neighbors in $P$ is green, whereas in (d) two of its
  neighbors in $P$ are green).}
\label{fig:mapping}
\end{center}
\end{figure}

Now define the set
$G_P=K_P\cup\left(\bigcup_{{x\in{}K_{\compl{P}}}}f(x)\right)$.
We claim that $G_P$ is a guarding set for $P$.

\begin{lemma}\label{lem:guardingset}
The set
$G_P=K_P\cup\left(\bigcup_{{x\in{}K_{\compl{P}}}}f(x)\right)$
is a guarding set for $P$.
\end{lemma}

\begin{proof}
Let us consider the triangulation $\tr{\polyapprox{P}}$ of
$\polyapprox{P}$. The regions in $\compl{P}$ are sectors defined
by a curvilinear arc, which is a subarc of an edge of $P$ and the
corresponding chord connecting the endpoints of this subarc. Let us
consider the set of triangles in $\tr{\polyapprox{P}}$ and the set
$\sect{P}$ of sectors in $\compl{P}$. To show that $G_P$ is a guarding
set for $P$, it suffices show that every triangle in
$\tr{\polyapprox{P}}$ and every sector in $\sect{P}$ is guarded by at
least one vertex in $G_P$.

If $T$ is a triangle in $\tr{\polyapprox{P}}$ that is defined over
vertices of $P$, one of its vertices is colored red and belongs to
$K_P\subseteq{}G_P$. Hence, $T$ is guarded.

Consider now a triangle $T$ that is defined inside an empty room
$r_i$. If the auxiliary vertex of $T$ is not red, then one of the two
endpoints of $a_i$ has to be red, and thus it belongs to $G_P$. Hence
both $T$ and the two sectors adjacent to it in $r_i$ are guarded. If
the auxiliary vertex is red, then one of the other two vertices of $T$
is green and belongs to $G_P$; again, $T$ is guarded.

Suppose now that $T$ is a boundary crescent triangle, and let
$s$ be the sector adjacent to it (consult
Fig. \ref{fig:guardingsetproof1}). Let $x$ be the endpoint of $a_i$
contained in $T$, $y$ be the second point of $T$ that belongs to $P$ 
and $z$ the point in $\compl{P}$. Note that all three vertices
guard the sector $s$. If $x$ (\resp $y$) is a red vertex it will also
be a vertex in $G_P$. Hence, in this case both $T$ and $s$ are guarded
by $x$ (\resp $y$). If $z$ is the red vertex in $T$, either $x$ or $y$
has to be a green vertex. Hence either $x$ or $y$ will be in $G_P$, and
thus again both $T$ and $s$ will be guarded.

\begin{figure}[t]
\begin{center}
\psfrag{x}[][]{\small$x$}
\psfrag{y}[][]{\small$y$}
\psfrag{z}[][]{\small$z$}
\psfrag{w}[][]{\small$w$}
\psfrag{s}[][]{\small$s$}
\psfrag{T}[][]{\small$T$}
\psfrag{T1}[][]{\small$T'$}
\psfrag{ai}[][]{\small$a_i$}
\psfrag{mi}[][]{\small$m_i$}
\subfigure[\label{fig:guardingsetproof1}]
          {\includegraphics[width=0.3\textwidth]{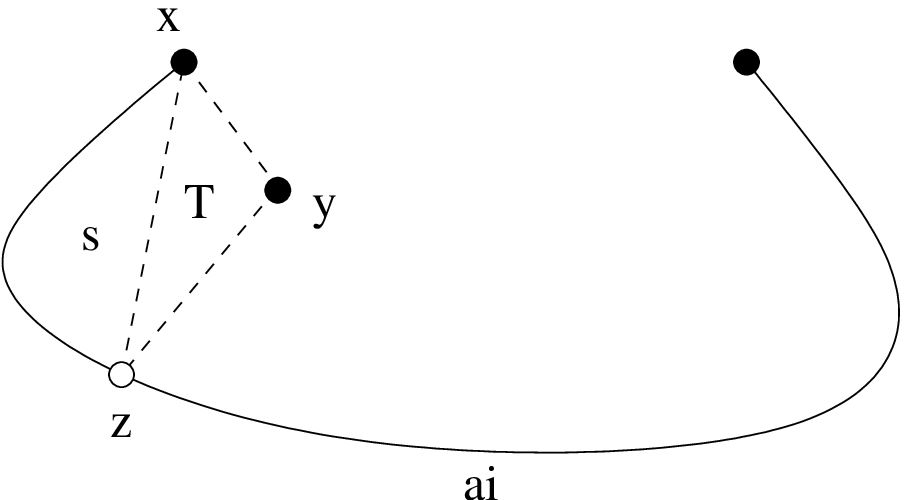}}\hfill%
\subfigure[\label{fig:guardingsetproof2}]
          {\includegraphics[width=0.3\textwidth]{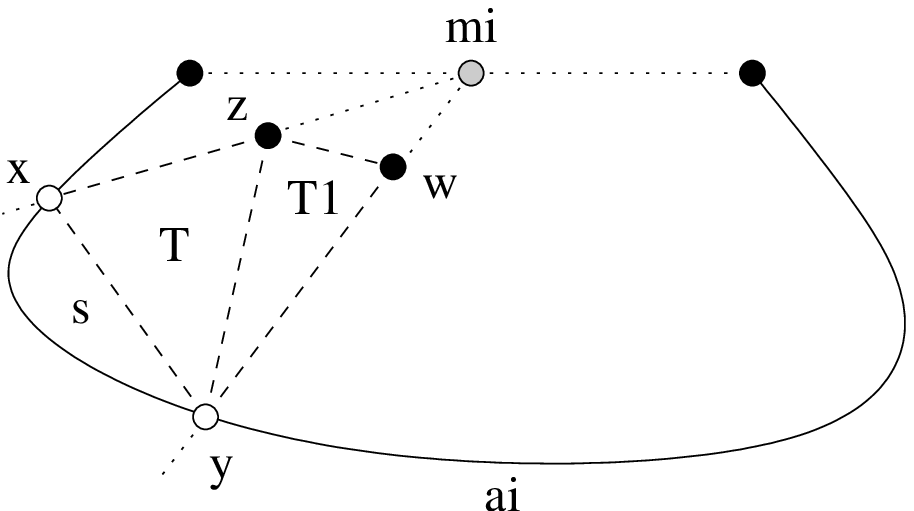}}\hfill%
\subfigure[\label{fig:guardingsetproof3}]
          {\includegraphics[width=0.3\textwidth]{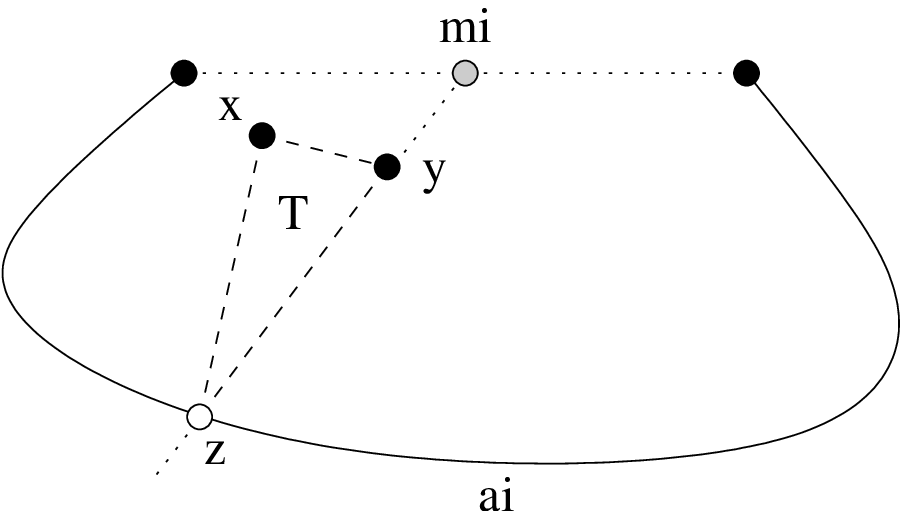}}%
\caption{Three of the five cases in the proof of Lemma
  \ref{lem:guardingset}: (a) the triangle $T$ is a boundary crescent
  triangle; (b) the triangle $T$ is a lower crescent triangle; (c) the
  triangle $T$ is an upper crescent triangle.}
\label{fig:guardingsetproof}
\end{center}
\end{figure}

If $T$ is a lower crescent triangle, let $s$ be the sector adjacent to
it (consult Fig. \ref{fig:guardingsetproof2}).
Let $x$, $y$ be the endpoints of the chord of $s$ on $a_i$ and
let $z$ be the point of $P$ in $T$. Let us also assume we encounter
$x$ and $y$ in that order as we walk along $a_i$ in the
counterclockwise sense, which implies that $x$ is the intersection of
the line $zm_i$ and the arc $a_i$. Finally, let $T'$ be the upper
crescent triangle incident to the edge $yz$, and let $w$ be the third
vertex of $T'$, beyond $y$ and $z$. It is interesting to note that all
four vertices $x$, $y$, $z$ and $w$ guard $T$, $T'$ and $s$. Moreover,
$x$ and $w$ have to be of the same color. In order to show that $T$
and $s$ are guarded by $G_P$, it suffices to show that one of $x$, $y$,
$z$ and $w$ belongs to $G_P$. Consider the following cases:
\begin{enumerate}
\item
  $z$ is a red vertex. Since $z\in{}K_P$, we get that $z\in{}G_P$.
\item
  $x$ is a red vertex. But then $w$ is also a red vertex. Since
  $w\in{}K_P$, we conclude that $w$ belongs to $G_P$ as well.
\item
  $y$ is a red vertex. Then either $z$ is a green vertex or both $x$
  and $w$ are green vertices. If $z$ is a green vertex, then
  $\{z\}\subseteq{}f(y)$, which implies that $z\in{}G_P$. If $z$ is a
  blue vertex, then both $x$ and $w$ are green vertices, and in
  particular $\{w\}\subseteq{}f(y)$. Hence $w\in G_P$.
\end{enumerate}

Finally, consider the case that $T$ is an upper crescent triangle, let
$x$ and $y$ be the vertices of $P$ in $T$ and let $z$ be the vertex of
$T$ in $\compl{P}$ (consult Fig. \ref{fig:guardingsetproof3}). Let us
also assume that $z$ is the intersection of
the line $ym_i$ with $a_i$. To show that $T$ is guarded by $G_P$, it
suffices to show that one of $x$ and $y$ belongs to
$G_P$. Consider the following cases:
\begin{enumerate}
\item
  $x$ is red vertex. Since $x\in{}K_P$ we have that $x\in{}G_P$.
\item
  $y$ is red vertex. Since $y\in{}K_P$ we have that $y\in{}G_P$.
\item
  $z$ is a red vertex. If $x$ is a green vertex, then
  $\{x\}\subseteq{}f(z)$. Hence $x\in{}G_P$. If $x$ is blue vertex,
  then $y$ has to be a green vertex, and
  $\{y\}\subseteq{}f(z)$. Therefore, $y\in{}G_P$.\proofbox
\end{enumerate}
\end{proof}

Since $f(x)\subseteq{}\Pi_P$ for every $x$ in
$K_{\compl{P}}$ we get that
$\bigcup_{x\in{}K_{\compl{P}}}f(x)\subseteq\Pi_P$.
But this, in turn implies that $G_P\subseteq K_P\cup\Pi_P$. Since
$K_P$ and $\Pi_P$ are the two sets of smallest cardinality among
$K_P$, $\Pi_P$ and $M_P$, we can easily verify that
$|K_P|+|\Pi_P|\le\lfloor\frac{2n}{3}\rfloor$. Hence,
$|G_P|\le|K_P|+|\Pi_P|\le\lfloor\frac{2n}{3}\rfloor$,
which yields the following theorem.

\begin{theorem}\label{thm:ubound}
Let $P$ be a \pconvex polygon with $n\ge{}2$ vertices. $P$ can
be guarded with at most $\lfloor\frac{2n}{3}\rfloor$ vertex guards.
\end{theorem}

We close this subsection by making two remarks:
\begin{remark}\sl
  The bound on the size of the vertex guarding set in Theorem
  \ref{thm:ubound} is tight: our algorithm will produce a vertex
  guarding set of size exactly $\lfloor\frac{2n}{3}\rfloor$ when
  applied to the \pconvex polygon of
  Fig. \ref{fig:polyapproxlb} or the crescent-like \pconvex
  polygon of Fig. \ref{fig:crescent}.
\end{remark}

\begin{remark}\sl
  When the input to our algorithm is a linear polygon all rooms are
  degenerate; consequently, no auxiliary vertices are created, and the
  guarding set computed corresponds to the set of colored vertices of
  smallest cardinality, hence producing a vertex guarding set of size at
  most $\lfloor\frac{n}{3}\rfloor$. In that respect, it can be
  considered as a generalization of Fisk's algorithm \cite{f-spcwt-78}
  to the class of \pconvex polygons.
\end{remark}


\subsection{Time and space complexity}
\label{sec:guardingalgo}

In this section we will show how to compute a vertex guarding set $G_P$, of
size at most $\lfloor\frac{2n}{3}\rfloor$, for $P$, in $O(n\log n)$ time and
$O(n)$ space. The algorithm presented at the beginning of this section
consists of four phases:
\begin{enumerate}
\item\label{step:pap}
  The computation of the polygonal approximation $\polyapprox{P}$ of
  $P$.
\item\label{step:ctr}
  The computation of the constrained triangulation
  $\tr{\polyapprox{P}}$ of $\polyapprox{P}$.
\item\label{step:gspap}
  The computation of a guarding set $G_{\polyapprox{P}}$ for
  $\polyapprox{P}$.
\item\label{step:gsp}
  The computation of a guarding set $G_P$ for $P$ from the guarding set
  $G_{\polyapprox{P}}$.
\end{enumerate}

Step \ref{step:ctr} of the algorithm presented above can be done in
$O(T(n))$ time and $O(n)$ space, where $T(n)$ is the time complexity
of any $O(n\log{}n)$ polygon triangulation algorithm: we need linear
time and space to create the constrained triangles of
$\tr{\polyapprox{P}}$, whereas the subpolygons created after the
introduction of the constrained triangles may be triangulated in
$O(T(n))$ time and linear space.

Step \ref{step:gspap} of the algorithm takes also linear time and
space with respect to the size of the polygon $\polyapprox{P}$. By
Corollary \ref{cor:polyapproxsize}, $|\polyapprox{P}|\le 3n$, which
implies that the guarding set $G_{\polyapprox{P}}$ can be computed in
$O(n)$ time and space.

Step \ref{step:gsp} also requires $O(n)$ time. Computing $G_P$ from
$G_{\polyapprox{P}}$ requires determining for each vertex $v$ of
$K_{\compl{P}}$ all the vertices of $\Pi_P$
adjacent to it. This takes time proportional to the degree $deg(v)$ 
of $v$ in $\tr{\polyapprox{P}}$, \ie a total of
$\sum_{v\in{}K_{\compl{P}}}deg(v)=O(|\polyapprox{P}|)=O(n)$
time. The space requirements for performing Step
\ref{step:gsp} is $O(n)$.

To complete our time and space complexity analysis, we need to show
how to compute the polygonal approximation $\polyapprox{P}$ of $P$ in
$O(n\log n)$ time and linear space. 
In order to compute the polygonal approximation $\polyapprox{P}$ or
$P$, it suffices to compute for each room $r_i$ the set of vertices
$C_i^*$. If $C_i^*=\emptyset$, then $r_i$ is empty, otherwise we have
the set of vertices we wanted. From $C_i^*$ we can compute the points
$w_{i,k}$ and the linear polygon $\polyapprox{P}$ in $O(n)$ time and
space.

The underlying idea is to split $P$ into $y$-monotone \pconvex
subpolygons. For each room $r_i$ within each such $y$-monotone
subpolygon, corresponding to a convex arc $a_i$ with endpoints $v_i$
and $v_{i+1}$, we will then compute the corresponding set
$C_i^*$. This will be done by first computing a subset $S_i$ of the
set $R_i$ of the points inside the room $r_i$, such that
$S_i\supseteq{}C_i^*$, and then apply an optimal time and space convex
hull algorithm to the set $S_i\cup\{v_i,v_{i+1}\}$ in order to compute
$C_i$, and subsequently from that $C_i^*$. In the discussion that
follows, we assume that for each convex arc $a_i$ of $P$ we associate
a set $S_i$, which is initialized to be the empty set. The sets $S_i$
will be progressively filled with vertices of $P$, so that in the end
they fulfill the containment property mentioned above.

Splitting $P$ into $y$-monotone \pconvex subpolygons
can be done in two steps:
\begin{enumerate}
\item
  First we need to split each convex arc $a_i$ into $y$-monotone
  pieces. Let $P'$ be the \pconvex polygon we get by
  introducing the $y$-extremal points for each $a_i$. Since each $a_i$
  can yield up to three $y$-monotone convex pieces, we conclude that
  $|P'|\le 3n$. Obviously splitting the convex arcs $a_i$ into
  $y$-monotone pieces takes $O(n)$ time and space. A vertex added to
  split a convex arc into $y$-monotone pieces will be called an
  \emph{added extremal vertex}.
\item
  Second, we need to apply the standard algorithm for computing
  $y$-monotone subpolygons out of a linear polygon to $P'$
  (cf. \cite{lp-lppsi-77} or \cite{bkos-cgaa-00}). The algorithm in
  \cite{lp-lppsi-77} (or \cite{bkos-cgaa-00}) is
  valid not only for line segments, but also for \pconvex
  polygons consisting of $y$-monotone arcs (such as $P'$). Since
  $|P'|\le{}3n$, we conclude that computing the $y$-monotone subpolygons
  of $P'$ takes $O(n\log n)$ time and requires $O(n)$ space.
\end{enumerate}
Note that a non-split arc of $P$ belongs to exactly one $y$-monotone
subpolygon. $y$-monotone pieces of a split arc of $P$ may belong to at
most three $y$-monotone subpolygons (see Fig. \ref{fig:split2monotone}).

\begin{figure}[t]
\begin{center}
\psfrag{Q1}[][]{\small$Q_1$}
\psfrag{Q2}[][]{\small$Q_2$}
\psfrag{Q3}[][]{\small$Q_3$}
\psfrag{Q4}[][]{\small$Q_4$}
\psfrag{Q5}[][]{\small$Q_5$}
\psfrag{Q6}[][]{\small$Q_6$}
\psfrag{Q7}[][]{\small$Q_7$}
\psfrag{Q8}[][]{\small$Q_8$}
\psfrag{Q9}[][]{\small$Q_9$}
\psfrag{Q10}[][]{\small$Q_{10}$}
\includegraphics[width=0.65\textwidth]{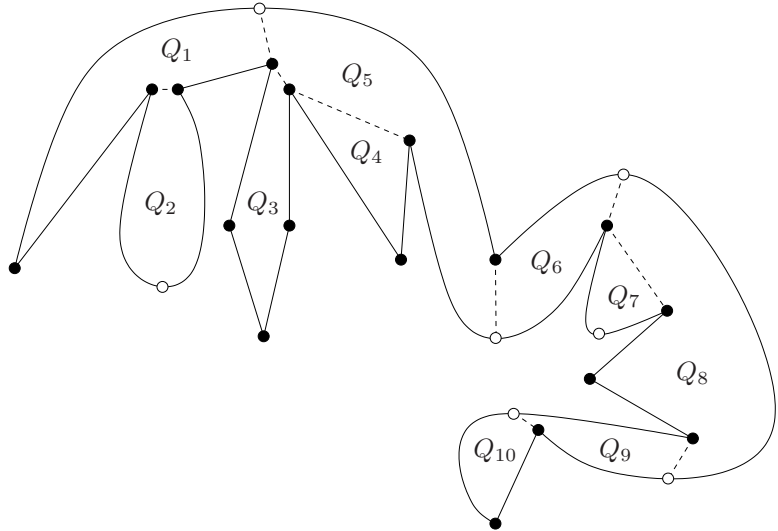}
\caption{Decomposition of a \pconvex polygon into ten
  $y$-monotone subpolygons. The white points are added extremal
  vertices that have been added in order to split non-$y$-monotone
  arcs to $y$-monotone pieces. The bridges are shown as dashed segments.}
\label{fig:split2monotone}
\end{center}
\end{figure}

At the beginning of our algorithm we associate to each arc $a_i$ of $P$
a set of vertices $S_i$, which is initialized to the empty set.
Suppose now that we have a $y$-monotone polygon $Q$. The edges of $Q$
are either convex arcs of $P$, or pieces of convex arcs of $P$, or
line segments between mutually visible vertices of $P$, added in order
to form the $y$-monotone subpolygons of $P$; we call these line
segments \emph{bridges} (see Fig. \ref{fig:split2monotone}).
For each non-bridge edge $e_i$ of $Q$, we want to compute the set
$C_i^*$. This can be done by sweeping $Q$ in the negative
$y$-direction (\ie by moving the sweep line from $+\infty$ to
$-\infty$). The events of the sweep correspond to the $y$ coordinates
of the vertices of $Q$, which are all known before-hand and can be put
in a decreasing sorted list. The first event of the sweep corresponds
to the top-most vertex of $Q$: since $Q$ consists of $y$-monotone
convex arcs, the $y$-maximal point of $Q$ is necessarily a vertex. The
last event of the sweep corresponds to the bottom-most vertex of $Q$,
which is also the $y$-minimal point of $Q$. We distinguish between
four different types of events:
\begin{enumerate}
\item the first event: corresponds to the top-most vertex of $Q$,
\item the last event: corresponds to the bottom-most vertex of $Q$,
\item a left event: corresponds to a vertex of the left $y$-monotone
  chain of $Q$, and
\item a right event: corresponds to a vertex of the right $y$-monotone
  chain of $Q$.
\end{enumerate}
Our sweep algorithm proceeds as follows. Let $\ell$ be the sweep
line parallel to the $x$-axis at some $y$. For each $y$ in between the
$y$-maximal and $y$-minimal values of $Q$, $\ell$ intersects $Q$ at two
points which belong to either a left edge $e_l$ (\ie an edge on the
left $y$-monotone chain of $Q$) or is a left vertex $v_l$ (\ie a
vertex on the left $y$-monotone chain of $Q$), and either a right edge
$e_r$ (\ie an arc on the right $y$-monotone chain of $Q$) or a right
vertex $v_r$ (\ie a vertex on the right $y$-monotone chain of
$Q$). We are going to associate the current left edge $e_l$ at position
$y$ to a point set $S_L$ and the current right edge at position $y$ to
a point set $S_R$. If the edge $e_l$ (\resp $e_r$) is a non-bridge
edge, the set $S_L$ (\resp $S_R$) will contain vertices of
$Q$ that are inside the room of the convex arc of $P$ corresponding
$e_l$ (\resp $e_r$).

When the $y$-maximal vertex $v_{max}$ is encountered, \ie during the
first event, we initialize $S_L$ and $S_R$ to be the empty set.
When a left event is encountered due a vertex $v$, let $e_{l,up}$ be the
left edge above $v$ and $e_{l,down}$ be the left edge below $v$
and let $e_r$ be the current right edge (\ie the right edge at the
$y$-position of $v$). If $e_{l,up}$ is an non-bridge edge, and $a_i$
is the corresponding convex arc of $P$, we augment the set $S_i$ by
the vertices in $S_L$. Then, irrespectively of whether or not
$e_{l,up}$ is a bridge edge, we re-initialize $S_L$ to be the
empty set. Finally, if $e_r$ is a non-bridge edge, and $a_k$ is the
corresponding convex arc in $P$, we check if $v$ is inside the room
$r_k$ or lies in the interior of the chord of $r_k$; if this is the
case we add $v$ to $S_R$.
When a right event is encountered our sweep algorithm behaves
symmetrically. If the right event is due to a vertex $v$, let $e_{r,up}$ be
right edge of $Q$ above $v$ and $e_{r,down}$ be the right edge of $Q$ below
$v$ and let $e_l$ be the current left edge of $Q$. If $e_{l,up}$ is a
non-bridge edge, and $a_i$ is the corresponding convex arc of $P$, we
augment $S_i$ by the vertices in $S_R$. Then, irrespectively of
whether or not $e_{r,up}$ is a bridge edge or not, we re-initialize
$S_R$ to be the empty set. Finally, If $e_l$ is a non-bridge edge, and
$a_k$ is the corresponding convex arc of $P$, we check if $v$ is
inside the room $r_k$ or lies in the interior of the chord of
$r_k$; if this is the case we add $v$ to $S_L$.
When the last event is encountered due to the $y$-minimal vertex $v_{min}$,
let $e_l$ and $e_r$ be the left and right edges of $Q$ above $v_{min}$,
respectively. If $e_l$ is a non-bridge edge, let $a_i$ be the
corresponding convex arc in $P$. In this case we simply augment $S_i$
by the vertices in $S_L$. Symmetrically, if $e_r$ is a non-bridge
edge, let $a_j$ be the corresponding convex arc in $P$. In this case
we simply augment $S_j$ by the vertices in $S_R$.

We claim that our sweep-line algorithm computes a set $S_i$ such that
$S_i\supseteq{}C_i^*$. To prove this we need the following
intermediate result:

\begin{lemma}\label{lem:sidecomp}
Given a non-empty room $r_i$ of $P$, with $a_i$ the corresponding
convex arc, the vertices of the set $C_i^*$ belong to the $y$-monotone
subpolygons of $P'$ computed via the algorithm in \cite{lp-lppsi-77}
(or \cite{bkos-cgaa-00}), which either contain the entire arc $a_i$ or
$y$-monotone pieces of $a_i$.
\end{lemma}

\begin{proof}
Let $r_i$ be a non-empty room, $a_i$ the corresponding convex arc and
let $u$ be a vertex of $P$ in $C_i^*$ that is not a vertex of any of
the $y$-monotone subpolygons of $P'$ (computed by the algorithm in
\cite{lp-lppsi-77} or \cite{bkos-cgaa-00}) that contain either the
entire arc $a_i$ or $y$-monotone pieces of $a_i$. Let $v_{max}$
(\resp $v_{min}$) be the vertex of $P$ of maximum (\resp minimum)
$y$-coordinate in $C_i$ (ties are broken lexicographically). Let
$\ell_u$ be the line parallel to the $x$-axis passing through
$u$. Consider the following cases:
\begin{enumerate}
\item $u\in{}C_i^*\setminus\{v_{min},v_{max}\}$. In this case $u$ will
  be a vertex in either the left $y$-monotone chain of $C_i$ or
  a vertex in the right $y$-monotone chain of $C_i$. Without loss of
  generality we can assume that $u$ is a vertex in the right
  $y$-monotone chain of $C_i$ (see Figs. \ref{fig:proof-lem8-case1a}
  and \ref{fig:proof-lem8-case1b}). 
  Let $u'$ be the intersection of $\ell_u$ with $a_i$. Let $Q$ (\resp
  $Q'$) be the $y$-monotone subpolygon of $P'$ that contains $u$
  (\resp $u'$); by our assumption $Q\ne{}Q'$. Finally, let $u_+$
  (\resp $u_-$) be the vertex of $C_i$ above (\resp below) $u$ in the
  right $y$-monotone chain of $C_i$.
  
  The line segment $uu'$ cannot intersect any edges of $P$, since this
  would contradict the fact that $u\in{}C_i^*$. Similarly, $uu'$
  cannot contain any vertices of $P'$: if $v$ is a vertex of $P$ in
  the interior of $uu'$, $u$ would be inside the triangle $vu_+u_-$,
  which contradicts the fact that $u\in{}C_i^*$, whereas if $v$ is a
  vertex of $P'\setminus{}P$ in the interior of $uu'$, $P$ would not
  be locally convex at $v$, a contradiction with the fact that $P$ is
  a \pconvex polygon. As a result, and since $Q\ne{}Q'$, there
  exists a bridge edge $e$ intersecting $uu'$. Let $w_+$, $w_-$ be the
  two endpoints of $e$ in $P'$, where $w_+$ lies above the line
  $\ell_u$ and $w_-$ lies below the line $\ell_u$. In fact neither $w_+$
  nor $w_-$ can be a vertex in $P'\setminus{}P$, since the
  algorithm in \cite{lp-lppsi-77} (or \cite{bkos-cgaa-00}) will
  connect a vertex in $P'\setminus{}P$ inside a room $r_k$ with either
  the $y$-maximal or the $y$-minimal vertex of $C_k$ only. Let
  $\ell_+$ (\resp $\ell_-$) be the line passing through the vertices
  $u$ and $u_+$ (\resp $u$ and $u_-$). Finally, let $s$ be the sector
  delimited by the lines $\ell_+$, $\ell_-$ and $a_i$. Now, if $w_+$
  lies inside $s$, then $u$ will be inside the triangle
  $w_+u_+u_-$ (see Fig. \ref{fig:proof-lem8-case1a}). Analogously, if
  $w_-$ lies inside $s$, then $u$ will be inside the triangle
  $w_-u_+u_-$. In both cases we get a contradiction with the fact that
  $u\in{}C_i^*$. If neither $w_+$ nor $w_-$ lie inside $s$, then both
  $w_+$ and $w_-$ have to be vertices inside $r_i$, and moreover $u$
  will lie inside the convex quadrilateral $w_+u_+u_-w_-$; again this
  contradicts the fact that $u\in{}C_i^*$ (see
  Fig. \ref{fig:proof-lem8-case1b}).
\begin{figure}[t]
  \psfrag{vi}[][]{\scriptsize$v_i$}
  \psfrag{vi+1}[][]{\scriptsize$v_{i+1}$}
  \psfrag{ai}[][]{\scriptsize$a_i$}
  \psfrag{lu}[][]{\scriptsize$\ell_u$}
  \psfrag{l+}[][]{\scriptsize$\ell_+$}
  \psfrag{l-}[][]{\scriptsize$\ell_-$}
  \psfrag{u}[][]{\scriptsize$u$}
  \psfrag{up}[][]{\scriptsize$u'$}
  \psfrag{u+}[][]{\scriptsize$u_+$}
  \psfrag{u-}[][]{\scriptsize$u_-$}
  \psfrag{w+}[][]{\scriptsize$w_+$}
  \psfrag{w-}[][]{\scriptsize$w_-$}
  \psfrag{s}[][]{\scriptsize$s$}
  \psfrag{vmaxp}[][]{\scriptsize$v_{max}'$}
  \subfigure[\label{fig:proof-lem8-case1a}]%
  {\includegraphics[width=0.3\textwidth]{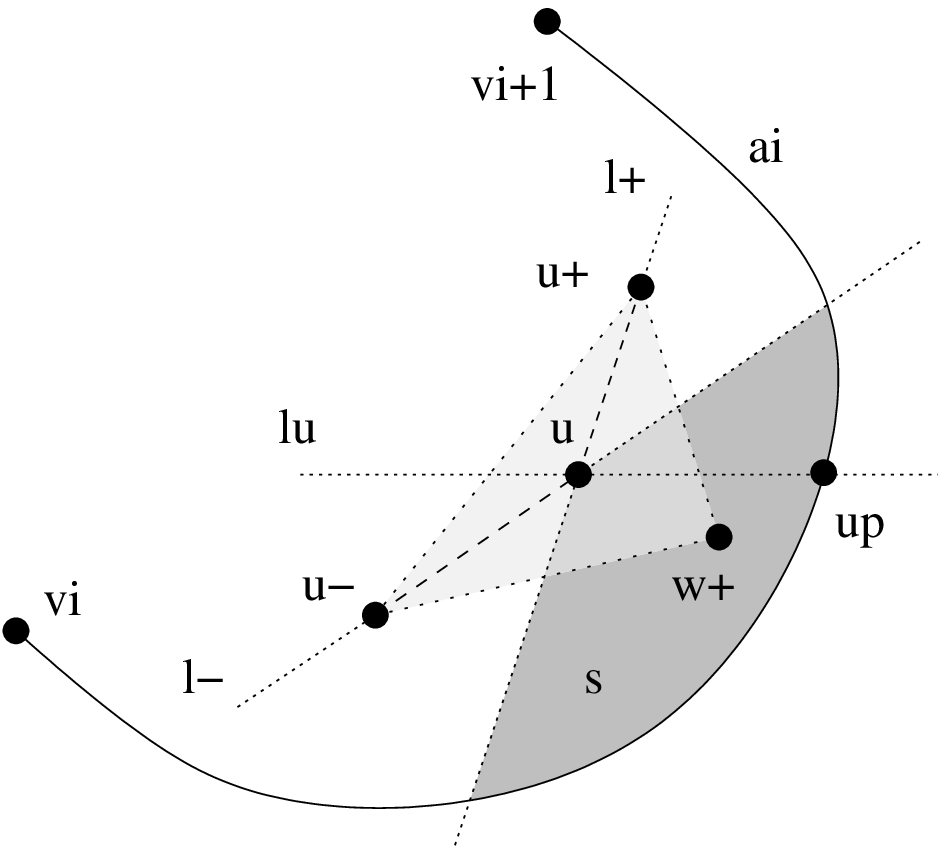}}%
  \hfill\subfigure[\label{fig:proof-lem8-case1b}]%
  {\includegraphics[width=0.3\textwidth]{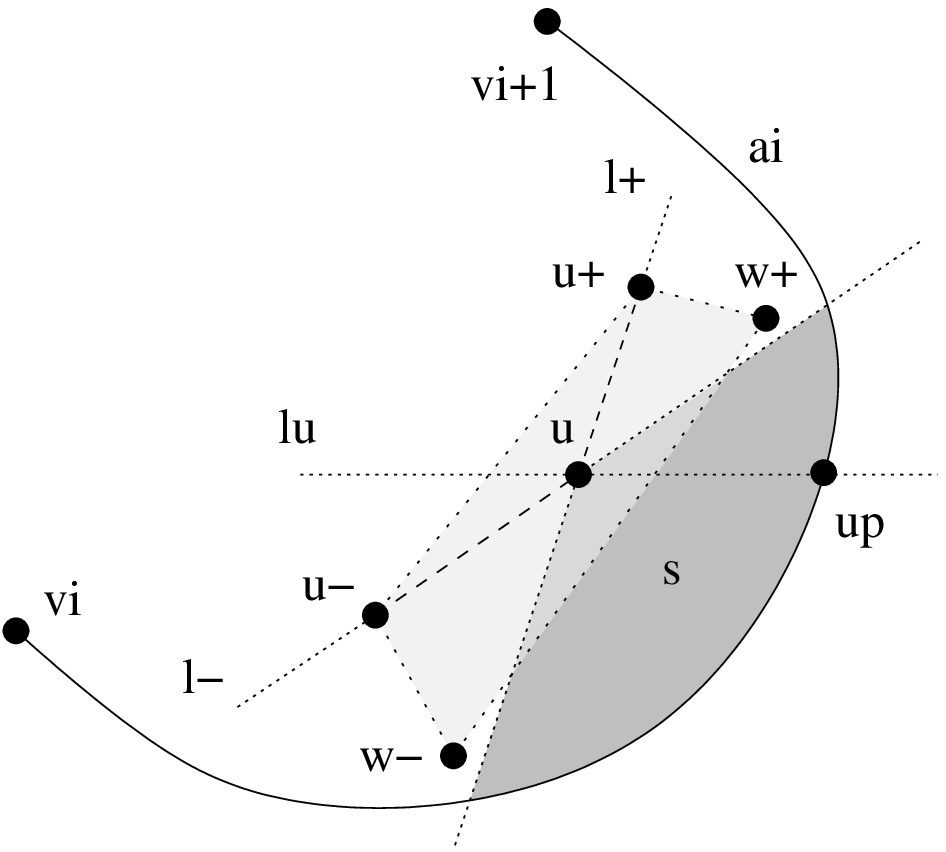}}%
  \hfill\subfigure[\label{fig:proof-lem8-case2}]%
  {\includegraphics[width=0.28\textwidth]{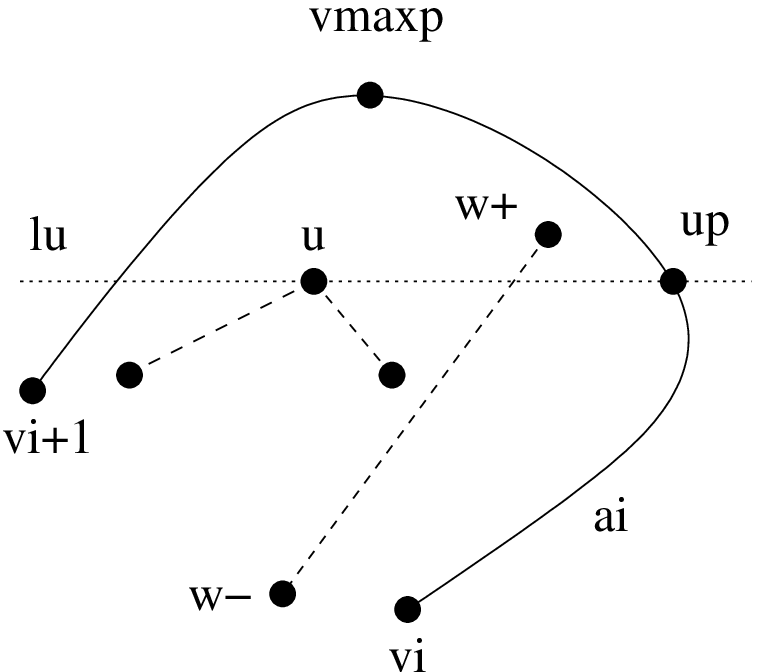}}%
  \caption{Proof of Lemma \ref{lem:sidecomp}. (a) The case
    $u\in{}C_i^*\setminus\{v_{min},v_{max}\}$, with $w_+\in{}s$.
    (b) The case $u\in{}C_i^*\setminus\{v_{min},v_{max}\}$, with
    $w_+,w_-\not\in{}s$. (c) The case $u\equiv{}v_{max}$.}
  \label{fig:proof-lem8}
\end{figure}
\item $u\equiv{}v_{max}$. By the maximality of the $y$-coordinate of
  $u$ in $C_i$, we have that the $y$-coordinate of $u$ is larger than
  or equal to the $y$-coordinates of both $v_i$ and $v_{i+1}$.
  Therefore, the line $\ell_u$ intersects the arc $a_i$ exactly twice,
  and, moreover, $a_i$ has a $y$-maximal vertex of $P'\setminus{}P$ in
  its interior, which we denote by $v_{max}'$ (see
  Fig. \ref{fig:proof-lem8-case2}). Let $u'$ be the intersection of
  $\ell_u$ with $a_i$ that lies to the right of $u$, and let $Q$
  (\resp $Q'$) be the $y$-monotone subpolygon of $P'$ that contains
  $u$ (\resp $u'$). By assumption $Q\ne{}Q'$, which implies that there
  exists a bridge edge $e$ intersecting the line segment
  $uu'$. Notice, that, as in the case
  $u\in{}C_i^*\setminus\{v_{min},v_{max}\}$, the line segment $uu'$
  cannot intersect any edges of $P$, or cannot contain any vertex $v$
  of $P'\setminus{}P$; the former would contradict the fact that
  $u\in{}C_i^*$, whereas as the latter would contradict the
  fact that $P$ is \pconvex. Furthermore, $uu'$ cannot contain
  vertices of $P$ since this would contradict the maximality of the
  $y$-coordinate of $u$ in $C_i$.

  Let $w_+$ and $w_-$ be the endpoints of $e$ above and below
  $\ell_u$, respectively. Notice that $e$ cannot have $v_{max}'$ as
  endpoint, since the only bridge edge that has $v_{max}'$ as endpoint
  is the bridge edge $v_{max}'u$. But then $w_+$ must be a vertex of
  $P$ lying inside $r_i$; this contradicts the maximality of the
  $y$-coordinate of $u$ among the vertices in $C_i$.
\item $u\equiv{}v_{min}$. This case is entirely symmetric to the case
  $u\equiv{}v_{max}$.\proofbox
\end{enumerate}
\end{proof}

An immediate corollary of the above lemma is the following:

\begin{corollary}
\label{cor:sweepsuperset}
For each convex arc $a_i$ of $P$, the set $S_i$ computed by the
sweep algorithm described above is a superset of the set $C_i^*$.
\end{corollary}

Let us now analyze the time and space complexity of Step
\ref{step:pap} of the algorithm sketched at the beginning of this
subsection. Computing the polygonal approximation $\polyapprox{P}$ of
$P$ requires subdividing $P$ into $y$-monotone subpolygons. This
subdivision takes $O(n\log n)$ time and $O(n)$ space. Once we have
the subdivision of $P$ into $y$-monotone subpolygons we need to
compute the sets $S_i$ for each convex arc $a_i$ of $P$. The sets
$S_i$ can be implemented as red-black trees. Inserting an element in
some $S_i$ takes $O(\log{}n)$ time. During the course of our algorithm
we perform only insertions on the $S_i$'s. A vertex $v$ of $P$ is
inserted at most $deg(v)$ times in some $S_i$, where $deg(v)$ is the
degree of $v$ in the $y$-monotone decomposition of $P$. Since the sum
of the degrees of the vertices of $P$ in the $y$-monotone
decomposition of $P$ is $O(n)$, we conclude that the total size of the
$S_i$'s is $O(n)$ and that we perform $O(n)$ insertions on the
$S_i$'s. Therefore we need $O(n\log{}n)$ time and $O(n)$ space to
compute the $S_i$'s. Finally, since $\sum_{i=1}^n|S_i|=O(n)$, the sets
$C_i^*$ can also be computed in total $O(n\log{}n)$ time and $O(n)$
space. The analysis above thus yields the following:

\begin{theorem}
Let $P$ be a \pconvex polygon with $n\ge{}2$ vertices. We can
compute a guarding set for $P$ of size at most
$\lfloor\frac{2n}{3}\rfloor$ in $O(n\log{}n)$ time and $O(n)$ space.
\end{theorem}


\subsection{The lower bound construction}
\label{sec:lowerbound}

\begin{figure}[!b]
\begin{center}
\psfrag{v1}[][]{\small$v_1$}
\psfrag{v2}[][]{\small$v_2$}
\psfrag{v3}[][]{\small$v_3$}
\psfrag{v4}[][]{\small$v_4$}
\psfrag{v5}[][]{\small$v_5$}
\psfrag{v6}[][]{\small$v_6$}
\psfrag{v7}[][]{\small$v_7$}
\psfrag{m}[][]{\small$m$}
\psfrag{a3}[][]{\small$a_3$}
\psfrag{a4}[][]{\small$a_4$}
\subfigure[\label{fig:windmill}]
          {\includegraphics[width=0.49\textwidth]{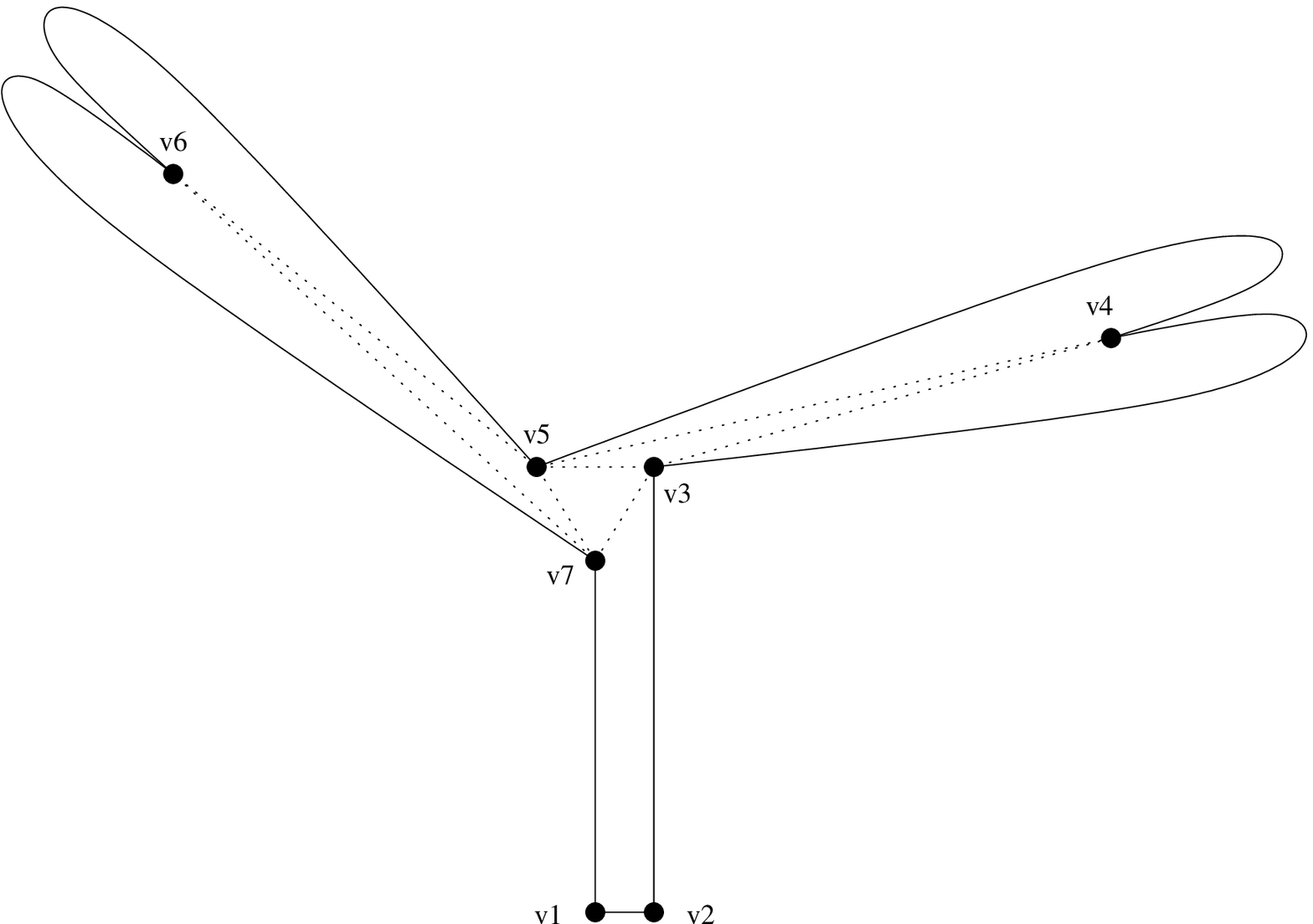}}\hfill%
\subfigure[\label{fig:windmillear}]
          {\includegraphics[width=0.49\textwidth]{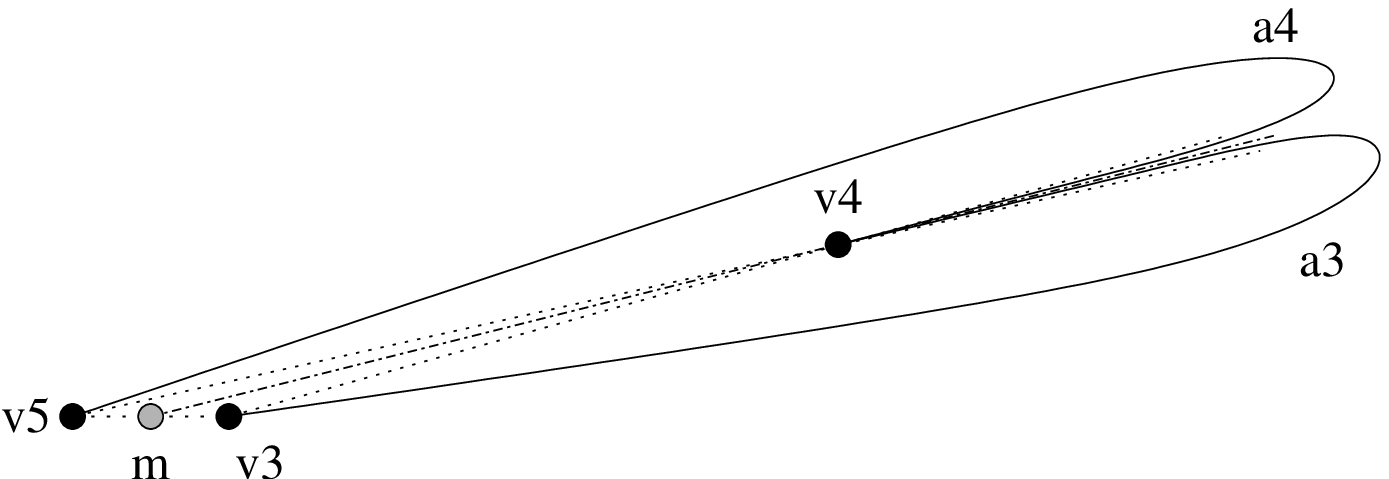}}
\caption{The windmill-like \pconvex polygon $W$ that requires
  at least three vertex guards in order to be guarded. The only
  triplets of guards that guard $W$ are $\{v_3,v_4,v_6\}$,
  $\{v_3,v_5,v_6\}$, $\{v_3,v_5,v_7\}$, $\{v_4,v_5,v_7\}$ and
  $\{v_4,v_6,v_7\}$.}
\label{fig:windmillall}
\end{center}
\end{figure}

\begin{figure}[t]
\begin{center}
\psfrag{vi}[][]{\small$v_i$}
\psfrag{vi+1}[][]{\small$v_{i+1}$}
\psfrag{vi+2}[][]{\small$v_{i+2}$}
\psfrag{vi+3}[][]{\small$v_{i+3}$}
\psfrag{u}[][]{\small$u$}
\includegraphics[width=0.9\textwidth]{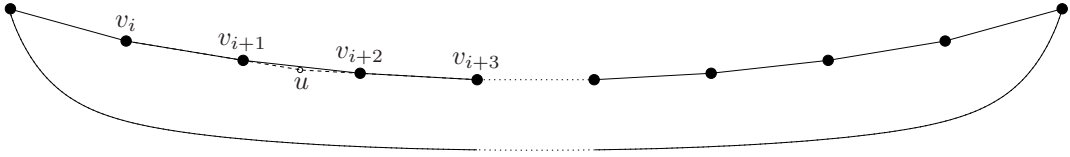}
\caption{The crescent-like \pconvex polygon $C$, that requires
  a guarding set of at least $\lfloor\frac{n}{2}\rfloor$ vertex guards.}
\label{fig:crescent}
\end{center}
\end{figure}

In this section we are going to present a \pconvex polygon
which requires a minimum of $\lfloor\frac{4n}{7}\rfloor-1$ vertex
guards in order to be guarded.

Let us first consider the windmill-like \pconvex polygon $W$
with seven vertices of Fig. \ref{fig:windmill}, a detail of which is
shown in Fig. \ref{fig:windmillear}. The \emph{double ear} defined by
the vertices $v_3$, $v_4$ and $v_5$ and the convex arcs $a_3$ and
$a_4$ is constructed in such a way so that neither
$v_3$ nor $v_5$ can guard both rooms $r_3$ and $r_4$ by
itself. This is achieved by ensuring that $a_3$ (\resp $a_4$)
intersects the line $v_4v_5$ (\resp $v_3v_4$) twice. Note that both
$a_3$ and $a_4$ intersect the line $mv_4$ only at $v_4$, where $m$ is
the midpoint of the line segment $v_3v_5$. The double ear defined
by the vertices $v_5$, $v_6$ and $v_7$ and the convex arcs $a_5$ and
$a_6$ is constructed in an analogous way. Moreover, the vertices
$v_1$, $v_2$, $v_4$ and $v_6$ are placed in such a way so that they do
not (collectively) guard the interior of the triangle $v_3v_5v_7$ (for
example the lengths of the edges $v_1v_7$ and $v_2v_3$ are considered
to be big enough, so that $v_2$ does not see too much of the triangle
$v_3v_5v_7$). As a result of this construction, $W$ cannot be guarded
by two vertex guards, but can be guarded with three. There are
actually only five possible guarding triplets: $\{v_3,v_4,v_6\}$,
$\{v_3,v_5,v_6\}$, $\{v_3,v_5,v_7\}$, $\{v_4,v_5,v_7\}$ and
$\{v_4,v_6,v_7\}$. Any guarding set that contains either $v_1$ or
$v_2$ has cardinality at least four. The vertices $v_1$ and $v_2$ will
be referred to as \emph{base vertices}.

Consider now the crescent-like polygon $C$ with $n$ vertices of
Fig. \ref{fig:crescent}. The vertices of $C$ are in strictly convex
position. This fact has the following implication: if $v_i$, $v_{i+1}$,
$v_{i+2}$ and $v_{i+3}$ are four consecutive vertices of $C$, and $u$
is the point of intersection of the lines $v_iv_{i+1}$ and
$v_{i+2}v_{i+3}$, then the triangle $v_{i+1}uv_{i+2}$ is guarded
if and only if either $v_{i+1}$ or $v_{i+2}$ is in the guarding set of $C$.
As a result, it is easy to see that $C$ cannot be guarded with less
than $\lfloor\frac{n}{2}\rfloor$ vertices, since in this case there
will be at least one edge both endpoints of which would not be in the
guarding set for $C$.

In order to construct the \pconvex polygon that gives us the
lower bound mentioned at the beginning of this section, we are going
to merge several copies of $W$ with $C$. 
More precisely, consider the \pconvex polygon $P$ of Fig.
\ref{fig:lbconstruction} with $n=7k$ vertices. It consists of copies of
the polygon $W$ merged with $C$ at every other linear edge of $C$,
through the base points of the $W$'s.

In order to guard any of the windmill-like subpolygons, we
need at least three vertices per such polygon, none which can be a
base point. This gives a total of $3k$ vertices. On the other hand, in
order to guard the crescent-like part of $P$ we need at least $k-1$
guards among the base points. To see that, notice that there are $k-1$
linear segments connecting base points; if we were to use less than
$k-1$ guards, we would have at least one such line segment $e$, both
endpoints of which would not participate in the guarding set of
$G$. But then, as in the case of $C$, there would be a triangle,
adjacent to $e$, which would not be guarded. Therefore, in order to
guard $P$ we need a minimum of $4k-1=\lfloor\frac{4n}{7}\rfloor-1$
guards, which yields the following theorem.

\begin{figure}[!t]
\begin{center}
\includegraphics[width=0.99\textwidth]{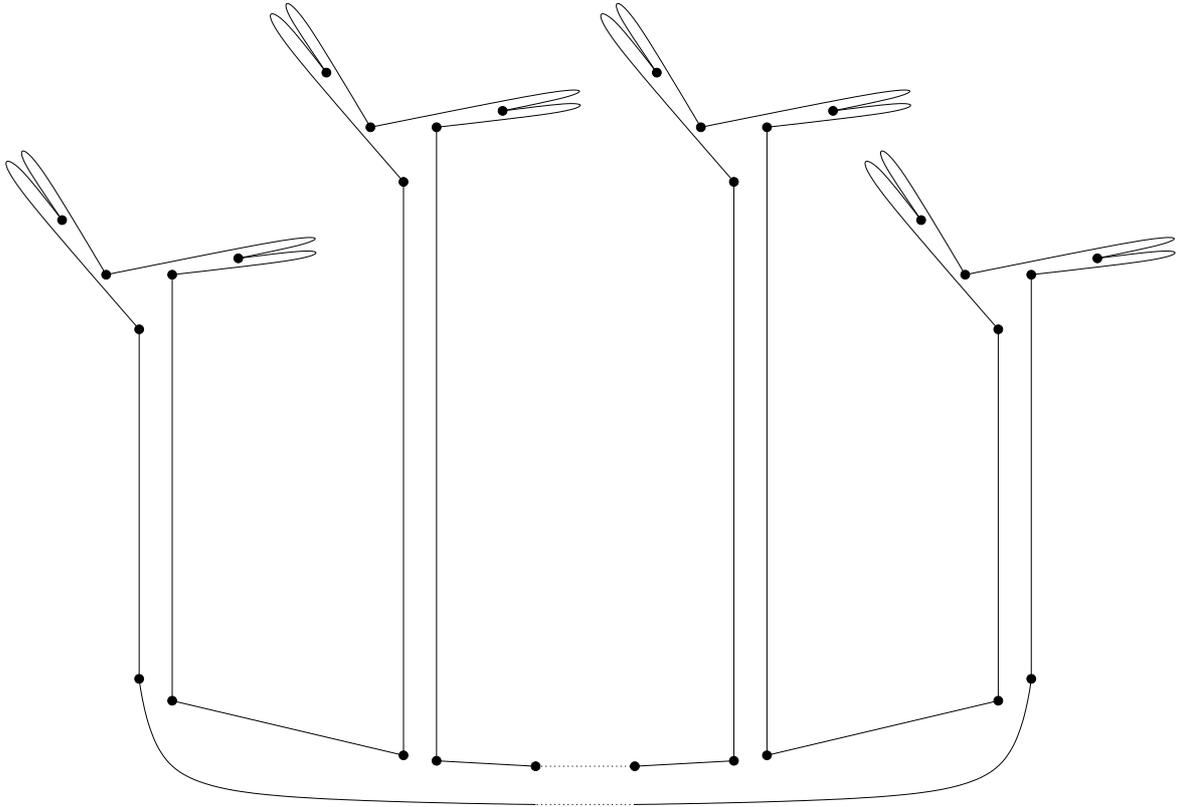}
\caption{The lower bound construction.}
\label{fig:lbconstruction}
\end{center}
\end{figure}

\begin{theorem}
There exists a family of \pconvex polygons with $n$ vertices
any vertex guarding set of which has cardinality at least
$\lfloor\frac{4n}{7}\rfloor-1$.
\end{theorem}


\section{Monotone \pconvex polygons}
\label{sec:monotonepiececonvex}

In this section we focus on the subclass of \pconvex polygons
that are monotone.
Let us recall the definition of monotone polygons from Section
\ref{sec:intro}: a curvilinear polygon $P$ is called monotone if there
exists a line $L$ such that any line $L^\perp$ perpendicular to $L$
intersects $P$ at most twice.

In the case of linear polygons monotonicity does not yield
better bounds on the worst case number of point or vertex guards
needed in order to guard the polygon. In both cases, monotone or
possibly non-monotone linear polygons, $\lfloor\frac{n}{3}\rfloor$
point or vertex guards are always sufficient and sometimes
necessary.
In the context of \pconvex polygons the situation is
different. Unlike general (\ie not necessarily monotone) \pconvex
polygons, which require at least $\lfloor\frac{4n}{7}\rfloor-1$
vertex guards and can always be guarded with
$\lfloor\frac{2n}{3}\rfloor$ vertex guards, monotone \pconvex
polygons can always be guarded with $\lfloor\frac{n}{2}\rfloor+1$
vertex or $\lfloor\frac{n}{2}\rfloor$ point guards. These bounds are
tight, since there exist monotone \pconvex polygons that
require that many vertex (see Figs. \ref{fig:monotonevertexguardsoddlb}
and \ref{fig:monotonevertexguardsevenlb}) or point guards (see
Fig. \ref{fig:monotonepointguardslb}). This section is devoted to
the presentation of these tight bounds.

\myparagraph{Vertex guards.}
Let us consider a monotone \pconvex polygon $P$, and let us
assume without loss of generality that $P$ is monotone with respect to
the $x$-axis (see also Fig. \ref{fig:monotonepiececonvex}). Let $u_j$,
$1\le{}j\le{}n$, be the $j$-th vertex of $P$ when considered in the
list of vertices sorted with respect to their $x$-coordinate (ties are
broken lexicographically). Let also $u_0$ (\resp $u_{n+1}$) be the
left-most (\resp right-most) point of $P$.
Let $\ell_j$, $0\le{}j\le{}n+1$ be the vertical line
passing through the point $u_j$ of $P$, and let
$\mathcal{L}=\{\ell_0, \ell_1,\ell_2,\ldots,\ell_{n+1}\}$ be the 
collection of these lines. An immediate consequence of the fact that
$P$ is monotone and \pconvex is the following corollary:

\begin{figure}[!t]
\begin{center}
\psfrag{u0}[][]{$u_0$}
\psfrag{u2p}[][]{$u_2'$}
\psfrag{u10}[][]{$u_{10}$}
\psfrag{v1}[][]{$v_1$}
\psfrag{v2}[][]{$v_2$}
\psfrag{v3}[][]{$v_3$}
\psfrag{v4}[][]{$v_4$}
\psfrag{v5}[][]{$v_5$}
\psfrag{v6}[][]{$v_6$}
\psfrag{v7}[][]{$v_7$}
\psfrag{v8}[][]{$v_8$}
\psfrag{v9}[][]{$v_9$}
\includegraphics[width=0.95\textwidth]{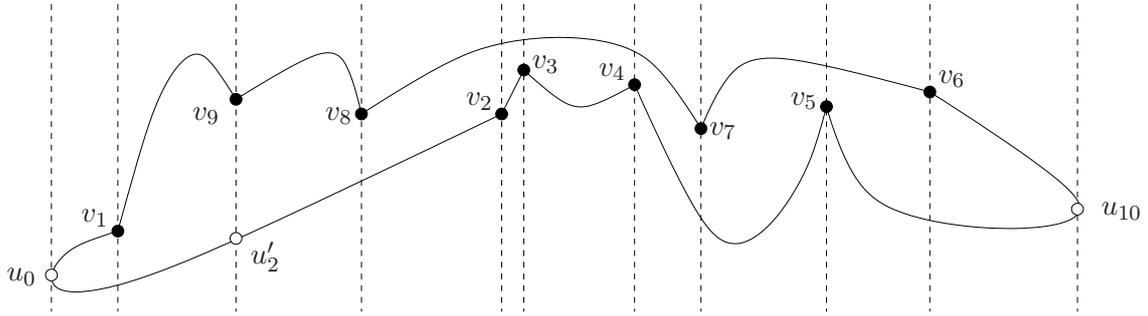}
\caption{A monotone \pconvex polygon $P$ with $n=9$ vertices
  and its vertical decomposition into four-sided convex slabs. The
  points $u_0$ and $u_{10}$ are the left-most and right-most points
  of $P$; $u_2'$ is the projection of $u_2\equiv{}v_9$, along
  $\ell_2$, on the opposite chain of $P$. $P$ can be either guarded with:
  (1) $\lfloor\frac{n}{2}\rfloor+1=5$ vertices, 
  namely the vertex set
  $\{u_1,u_3,u_5,u_7,u_9\}\equiv\{v_1,v_8,v_3,v_7,v_6\}$, or
  $\lfloor\frac{n}{2}\rfloor=4$ points, namely the point set
  $\{u_2',u_4,u_6,u_8\}\equiv\{u_2',v_2,v_4,v_5\}$.}
\label{fig:monotonepiececonvex}
\end{center}
\end{figure}

\begin{corollary}\label{cor:vertical_decomp}
The collection of lines $\mathcal{L}$ decomposes the interior of $P$
into at most $n+1$ convex regions $\kappa_j$, $j=0,\ldots,n$, that
are free of vertices or edges of $P$.
\end{corollary}

In addition to the fact that the region $\kappa_j$,
$1\le{}j\le{}n-1$, is convex, $\kappa_j$ has on its boundary both
vertices $u_j$ and $u_{j+1}$. This immediately implies that both $u_j$
and $u_{j+1}$ see the entire region $\kappa_j$. As far as $\kappa_0$
and $\kappa_n$ are concerned, they have $u_1$ and $u_n$ on their
boundary, respectively. As a result, $u_1$ sees $\kappa_0$, whereas
$u_n$ sees $\kappa_n$. Hence, in order to guard $P$ it suffices to take
every other vertex $u_j$, starting from $u_1$, plus $u_n$ (if not
already taken). The set
$G=\{u_{2m-1}, 1\le{}m\le{}\lfloor\frac{n}{2}\rfloor\}\cup\{u_n\}$
is, thus, a vertex guarding set for $P$ of size
$\lfloor\frac{n}{2}\rfloor+1$.

A line $L$ with respect to which $P$ is monotone can be computed in
$O(n)$ time if it exists \cite{ds-cgcw-90}. Given $L$, we can compute
the vertex guarding set $G$ for $P$ in $O(n)$ time and $O(n)$ space:
project the vertices of $P$ on $L$ and merge the two sorted (with
respect to their ordering on $L$) lists of vertices in the upper and
lower chain of $P$; then report every other vertex in the merged
sorted list starting from the first vertex, plus the last vertex of
$P$, if it has not already been reported.

\begin{figure}[t]
\begin{center}
  \includegraphics[width=0.99\textwidth]{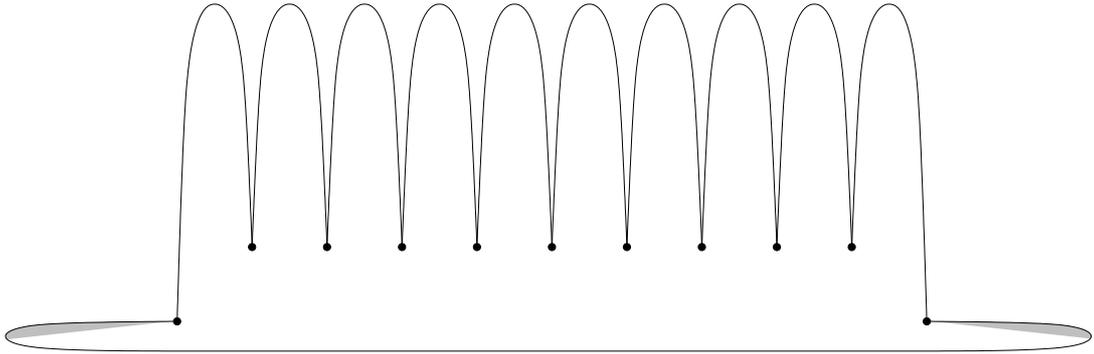}
  \caption{A monotone \pconvex polygon $M_1$ with an odd
    number of vertices that requires $\lfloor\frac{n}{2}\rfloor+1$
    vertex guards in order to be guarded: the shaded regions require
    that at least one of the two endpoints of the bottom-most edge of
    the polygon to be in the guarding set.}
  \label{fig:monotonevertexguardsoddlb}
\end{center}
\end{figure}

\begin{figure}[t]
  \begin{center}
    \psfrag{r1}[][]{\figtextsize$s_1$}
    \psfrag{r2}[][]{\figtextsize$s_2$}
    \psfrag{r3}[][]{\figtextsize$s_3$}
    \psfrag{r4}[][]{\figtextsize$s_4$}
    \psfrag{r5}[][]{\figtextsize$s_5$}
    \psfrag{x1}[][]{\figtextsize$x_1$}
    \psfrag{x2}[][]{\figtextsize$x_2$}
    \psfrag{x3}[][]{\figtextsize$x_3$}
    \psfrag{x4}[][]{\figtextsize$x_4$}
    \psfrag{x5}[][]{\figtextsize$x_5$}
    \psfrag{x6}[][]{\figtextsize$x_6$}
    \psfrag{x7}[][]{\figtextsize$x_7$}
    \psfrag{x8}[][]{\figtextsize$x_8$}
    \includegraphics[width=\textwidth]{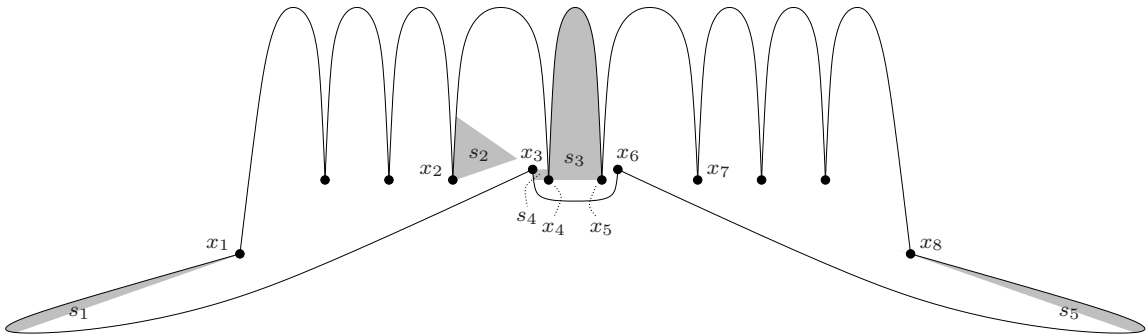}
  \end{center}
  \caption{A monotone \pconvex polygon $M_2$ with an even
    number of vertices that requires $\lfloor\frac{n}{2}\rfloor+1$
    vertex guards in order to be guarded.}
  \label{fig:monotonevertexguardsevenlb}
\end{figure}

The polygons $M_1$ and $M_2$ yielding the lower bound are shown in
Figs. \ref{fig:monotonevertexguardsoddlb} and
\ref{fig:monotonevertexguardsevenlb}. $M_1$ has an odd number of
vertices, whereas $M_2$ has an even number of vertices. Let $G_1$
(\resp $G_2$) be the vertex guarding set for $M_1$ (\resp $M_2$).
Let us first consider $M_1$ (see Fig. \ref{fig:monotonevertexguardsoddlb}). 
Notice that each prong of $M_1$ is fully guarded by either of its two
endpoints; the other vertices of $M_1$ can only partially guard the
prongs that they are not adjacent to. Moreover, the shaded regions of
$M_1$ can only be guarded by $u_1$ or $u_n$. Suppose, now, we can
guard $M_1$ with less than $\lfloor\frac{n}{2}\rfloor+1$ vertex
guards. Then either two consecutive vertices $u_i$ and $u_{i+1}$ of
$M_1$, $1\le{}i\le{}n-1$, will not belong to $G_1$, or $u_1$ and $u_n$
will not belong to $G_1$. In the former case, the prong that has $u_i$
and $u_{i+1}$ as endpoints is only partially guarded by the vertices
in $G_1$, a contradiction. In the latter case, the shaded regions of
$M_1$ are not guarded by the vertices in $G_1$, again a contradiction.

Consider now the polygon $M_2$ (see Fig. \ref{fig:monotonevertexguardsevenlb}). 
The number of vertices of $M_2$ between $x_1$ and $x_2$ is equal to
the number of vertices between $x_7$ and $x_8$, and even in
number. Every prong of $M_2$ between $x_1$ and $x_2$ (\resp between
$x_7$ and $x_8$) can be guarded by its two endpoints only; all other
vertices of $M_2$ guard each such prong only partially. The shaded
region $s_1$ (\resp $s_5$) is guarded only if either $x_1$ or $x_3$
(\resp either $x_6$ or $x_8$) belongs to $G_2$. The prong with
endpoints $x_2$ and $x_4$ can be guarded by either both $x_2$ and
$x_4$, or by $x_3$. If $x_2$ is the only vertex in $G_2$ among $x_2$,
$x_3$ and $x_4$, then the shaded region $s_4$ is not guarded.
Similarly, if $x_4$ is the only vertex in $G_2$ among $x_2$, $x_3$ and
$x_4$, then the shaded region $s_2$ is not guarded. Finally, if
neither $x_4$ nor $x_5$ belong to $G_2$, then the shaded prong $s_3$
is not guarded. Let us suppose now that $M_2$ can be guarded by less
than $\lfloor\frac{n}{2}\rfloor+1$ vertex guards. By our observations
above, it is not possible that two consecutive vertices $u_i$ and
$u_{i+1}$ of $M_2$, $1\le{}i\le{}n-1$, do not belong to $G_2$. Hence
$G_2$ will be a subset of the set
$G_2'=\{u_{2m-1},1\le{}m\le{}\lfloor\frac{n}{2}\rfloor\}$ or a subset
of the set
$G_2''=\{u_{2m},1\le{}m\le{}\lfloor\frac{n}{2}\rfloor\}$. In the
former case, \ie if $G_2\subseteq{}G_2'$, neither $x_6$ nor $x_8$
belong to $G_2$, and thus the region $s_5$ is not guarded, a
contradiction. Similarly, if $G_2\subseteq{}G_2''$, neither $x_1$ nor
$x_3$ belong to $G_2$, and thus the region $s_1$ is not guarded, again
a contradiction. We thus conclude that
$|G_2|\ge{}\lfloor\frac{n}{2}\rfloor+1$.

\myparagraph{Point guards.}
We now turn our attention to guarding $P$ with point guards (refer
again to Fig. \ref{fig:monotonepiececonvex}). Define $G_{even}$ to be
the vertex set $G_{even}=\{u_{2m},1\le{}m\le{}\lfloor\frac{n}{2}\rfloor\}$.
If $u_0\ne{}u_1$, \ie if $\kappa_0\ne\emptyset$, let $e_f$ be the
first (left-most) edge of $P$, and $u_\mu$, $\mu>1$, the right-most
endpoint of $e_f$ (the left-most endpoint of $e_f$ is necessarily $u_1$).
If $u_{n+1}\ne{}u_n$, \ie if $\kappa_{n+1}\ne\emptyset$, let $e_l$
be the last (right-most) edge of $P$, and $u_\nu$, $\nu<n$, the left-most
endpoint of $e_l$ (the right-most endpoint of $e_l$ is necessarily $u_n$).
Finally, let $u_i'$, $2\le{}i\le{}n-1$ be the projection along
$L^\perp$ of $u_i$ on the opposite monotone chain of $P$. Define the
set $G$ according to the following procedure:
\begin{enumerate}
\item[1.] Set $G$ equal to $G_{even}$.
\item[2.] If $u_0\ne{}u_1$ and $\mu>2$, replace $u_2$ in $G$ by $u_2'$.
\item[3.] If $u_{n+1}\ne{}u_n$ and $n$ is odd and $\nu<n-1$, replace
  $u_{2\lfloor\frac{n}{2}\rfloor}$ by $u_{2\lfloor\frac{n}{2}\rfloor}'$.
\end{enumerate}

As in the case of vertex guards, the set $G$ can be computed in linear
time and space: $G_{even}$ can be computed in linear time and space,
whereas determining if $u_2$ (\resp $u_{2\lfloor\frac{n}{2}\rfloor}$)
is to be replaced in $G$ by $u_2'$ (\resp
$u_{2\lfloor\frac{n}{2}\rfloor}'$) takes $O(1)$ time. The following
lemma establishes that $G$ is indeed a point guarding set for $P$.

\begin{lemma}
  The set $G$ defined according to the procedure above is a point
  guarding set for $P$.
\end{lemma}

\begin{proof}
Every convex region $\kappa_i$, $3\le{}i\le{}n-3$ is guarded by either
$u_i$ or $u_{i+1}$, since one of the two is in $G$.

Now consider the convex regions $\kappa_0$, $\kappa_1$ and
$\kappa_2$. Both $u_2$ and $u_2'$ lie on the common boundary of
$\kappa_1$ and $\kappa_2$. Since either $u_2$ or $u_2'$ is in $G$, we
conclude that $\kappa_1$ and $\kappa_2$ are guarded.
If $\kappa_0=\emptyset$, \ie if $u_0\equiv{}u_1$, $\kappa_0$ is
vacuously guarded. Suppose $\kappa_0\ne\emptyset$, \ie $u_0\ne{}u_1$.
Let $r_f$ be the room of $P$ corresponding to the edge $e_f$. Clearly,
$\kappa_0\subseteq{}r_f$. We distinguish between the cases $\mu=2$ and
$\mu>2$. If $\mu=2$, then $u_2\in{}G$ guards $r_f$ and thus $\kappa_0$.
If $\mu>2$, the point $u_2'\in{}G$ is a point on $e_f$. Therefore,
$u_2'$ guards $r_f$ and thus $\kappa_0$.

Finally, we consider the convex regions $\kappa_{n-2}$, $\kappa_{n-1}$
and $\kappa_n$. If $\kappa_n=\emptyset$, \ie $u_{n+1}\equiv{}u_n$,
$\kappa_n$ is vacuously guarded. Suppose, now, that 
$\kappa_n\ne\emptyset$, \ie $u_{n+1}\ne{}u_n$. Let $r_l$ be the room
of $P$ corresponding to the edge $e_l$. Clearly,
$\kappa_n\subseteq{}r_l$. We distinguish between the cases
``$n$ even'' and ``$n$ odd''.
If $n$ is even, then both
$u_{n-2}\equiv{}u_{2\lfloor\frac{n}{2}\rfloor-2}$ and
$u_n\equiv{}u_{2\lfloor\frac{n}{2}\rfloor}$ belong to $G$. This
immediately implies that all three $\kappa_{n-2}$, $\kappa_{n-1}$ and
$\kappa_n$ are guarded: $\kappa_{n-2}$ is guarded by $u_{n-2}$, whereas
$\kappa_{n-1}$ and $\kappa_n$ are guarded by $u_n$.
If $n$ is odd, either $u_{n-1}\equiv{}u_{2\lfloor\frac{n}{2}\rfloor}$ or
$u_{n-1}'\equiv{}u_{2\lfloor\frac{n}{2}\rfloor}'$ belongs to $G$. Since both
$u_{n-1}$ and $u_{n-1}'$
lie on the common boundary of $\kappa_{n-2}$ and $\kappa_{n-1}$, we
conclude that both $\kappa_{n-2}$ and $\kappa_{n-1}$ are guarded.
To prove that $\kappa_n$ is guarded, we further distinguish between
the cases $\nu=n-1$ and $\nu<n-1$. If $\nu=n-1$, then
$u_{n-1}\in{}G$ is an endpoint of $r_l$, and thus guards
$\kappa_n$. If $\nu<n-1$, the point $u_{n-1}'\in{}G$ is a point on
$e_l$. Therefore, $u_{n-1}'$ guards $r_l$ and thus $\kappa_n$.
\proofbox
\end{proof}

\begin{figure}[t]
\begin{center}
\includegraphics[width=0.99\textwidth]{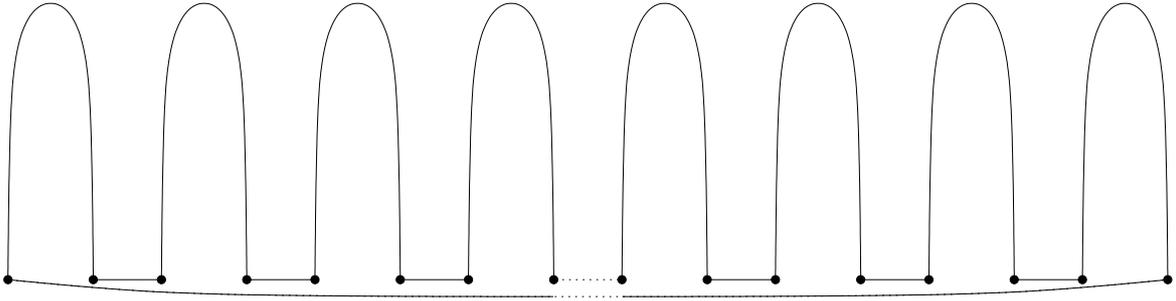}
\caption{A comb-like monotone \pconvex polygon that
  requires $\lfloor\frac{n}{2}\rfloor$ point guards in order to be
  guarded: one point guard is required per prong.}
\label{fig:monotonepointguardslb}
\end{center}
\end{figure}

As far as the minimum number of point guards required to guard a
monotone \pconvex polygon is concerned, the polygon $M$,
shown in Fig. \ref{fig:monotonepointguardslb}, yields the sought for
lower bound.
Notice that is very similar to the well known comb-like linear polygon
that establishes the lower bound on the number of point or vertex
guards required to guard a linear polygon. In our case it is easy to
see that we need at least one point guard per prong of the polygon,
and since there are $\lfloor\frac{n}{2}\rfloor$ prongs we conclude
that we need at least $\lfloor\frac{n}{2}\rfloor$ point guards in
order to guard $M$.

We are now ready to state the following theorem that summarizes the
results of this section.

\begin{theorem}
Given a monotone \pconvex polygon $P$ with $n\ge{}2$
vertices, $\lfloor\frac{n}{2}\rfloor+1$ vertex (\resp
$\lfloor\frac{n}{2}\rfloor$ point) guards are always sufficient and
sometimes necessary in order to guard $P$. Moreover, we can compute a
vertex (\resp point) guarding set for $P$ of size
$\lfloor\frac{n}{2}\rfloor+1$ (\resp $\lfloor\frac{n}{2}\rfloor$)
in $O(n)$ time and $O(n)$ space.
\end{theorem}


\section{\Pconcave polygons}
\label{sec:piececoncave}

In this section we deal with the problem of guarding \pconcave
polygons using point guards. Guarding a \pconcave
polygon with vertex guards may be impossible even for very simple
configurations (see Fig. \ref{fig:concave_no_vertex_guards}). In
particular we prove the following:

\begin{theorem}\label{thm:pconcave}
Let $P$ be a \pconcave polygon with $n$ vertices. $2n-4$ point
guards are always sufficient and sometimes necessary in order to guard
$P$.
\end{theorem}

\begin{proof}
To prove the sufficiency of $2n-4$ point guards we essentially apply the
technique in \cite{f-icd-77} for illuminating disjoint compact convex
sets --- please refer to Fig. \ref{fig:pconcaveub}. We denote by $A_i$
the convex object delimited by $a_i$ and the chord $v_iv_{i+1}$ of
$a_i$. Let $t_i(v_j)$ be the tangent line to $a_i$ at $v_j$,
$j=i,i+1$, and let $b_{i+1}$ be the bisecting ray of $t_i(v_{i+1})$,
$t_{i+1}(v_{i+1})$ pointing towards the interior of $P$. 

Construct a set of locally convex arcs
$\mathcal{K}=\{\kappa_1,\kappa_2,\ldots,\kappa_n\}$ that lie entirely
inside $P$ as such that (cf. \cite{f-icd-77}):
\begin{enumerate}
\item[(a)] the endpoints of $\kappa_i$ are $v_i$, $v_{i+1}$,
\item[(b)] $\kappa_i$ is tangent to $b_i$ (\resp $b_{i+1}$) at $v_i$
  (\resp $v_{i+1}$),
\item[(c)] if $S_i$ is the locally convex object defined by $\kappa_i$
  and its chord $v_iv_{i+1}$, then $A_i\subseteq{}S_i$, $1\le{}i\le{}n$,
\item[(d)] the arcs $\kappa_i$ are pairwise non-crossing, and
\item[(e)] the number of tangencies between the elements of
  $\mathcal{K}$ is maximized.
\end{enumerate}
Let $Q$ be the \pconcave polygon defined by the sequence of
the arcs in $\mathcal{K}$.

\begin{figure*}[t]
\begin{center}
  \psfrag{v1}[][]{\small$v_1$}
  \psfrag{v2}[][]{\small$v_2$}
  \psfrag{v3}[][]{\small$v_3$}
  \psfrag{v4}[][]{\small$v_4$}
  \psfrag{v5}[][]{\small$v_5$}
  \psfrag{v6}[][]{\small$v_6$}
  \psfrag{v7}[][]{\small$v_7$}
  \psfrag{v8}[][]{\small$v_8$}
  \psfrag{v9}[][]{\small$v_9$}
  \psfrag{v10}[][]{\small$v_{10}$}
  \psfrag{v11}[][]{\small$v_{11}$}
  \psfrag{r1}[][]{\textcolor{red}{\small$u_1$}}
  \psfrag{r2}[][]{\textcolor{red}{\small$u_2$}}
  \psfrag{r3}[][]{\textcolor{red}{\small$u_3$}}
  \psfrag{r4}[][]{\textcolor{red}{\small$u_4$}}
  \psfrag{r5}[][]{\textcolor{red}{\small$u_5$}}
  \psfrag{r6}[][]{\textcolor{red}{\small$u_6$}}
  \psfrag{r7}[][]{\textcolor{red}{\small$u_7$}}
  \psfrag{r8}[][]{\textcolor{red}{\small$u_8$}}
  \psfrag{r9}[][]{\textcolor{red}{\small$u_9$}}
  \psfrag{r10}[][]{\textcolor{red}{\small$u_{10}$}}
  \psfrag{r11}[][]{\textcolor{red}{\small$u_{11}$}}
  \includegraphics[width=0.9\textwidth]{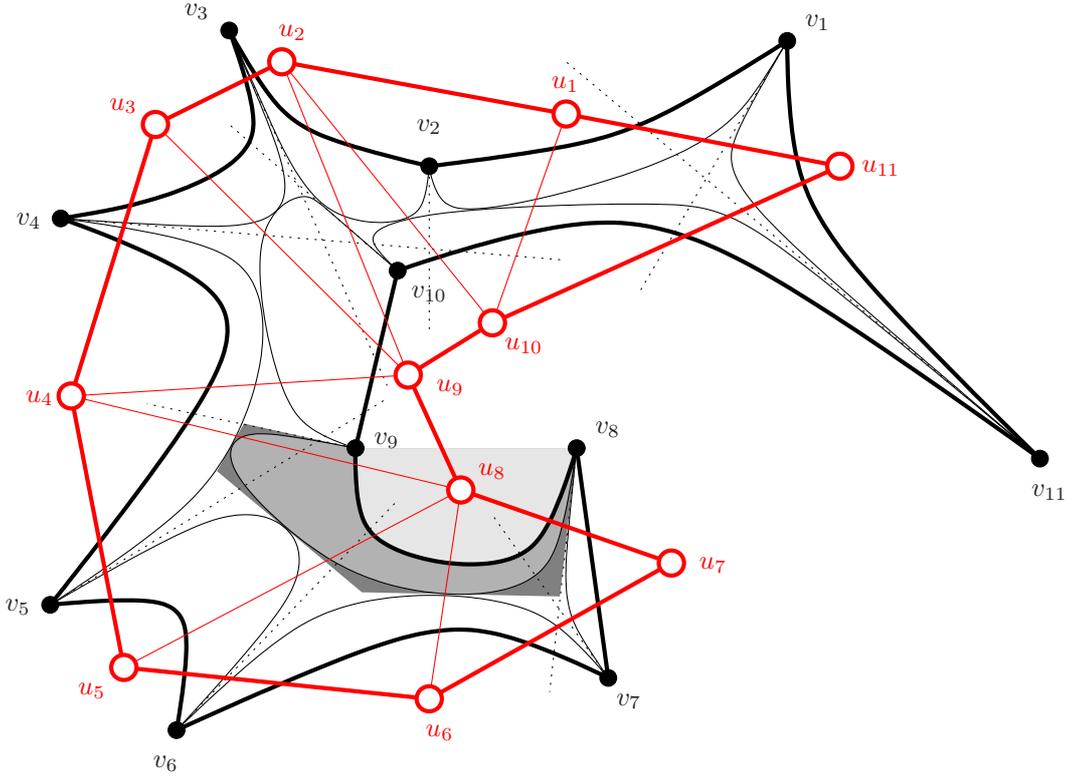}
  \caption{The proof for the upper bound of Theorem
    \ref{thm:pconcave}. The polygon $P$ is shown with thick solid
    curvilinear arcs. The arcs $\kappa_i$ are shown as thin solid
    arcs. The dotted rays are the bisecting rays $b_i$, whereas the
    dashed ray is the ray $r_8(v_9)$. The regions $A_8$,
    $S_8\setminus{}A_8$ and $\Pi_8\setminus{}S_8$ are also shown using
    three levels of gray; note that $\Pi_8$ has one reflex vertex at
    $v_9$. The graph $\Gamma$ (\ie the triangulation graph $\tr{R}$)
    is shown in red: the node $u_i$ corresponds to the arc $a_i$ and
    the polygon $R$ is depicted via thick segments.}
  \label{fig:pconcaveub}
\end{center}
\end{figure*}

Suppose now that $\kappa_i$ and $\kappa_{\sigma(j)}$ are tangent,
$1\le{}j\le{}m$, and let $\ell_{i,\sigma(j)}$ be the common tangent to
$\kappa_i$ and $\kappa_{\sigma(j)}$. Let $s_{i,\sigma(j)}$ be the
line segment on $\ell_{i,\sigma(j)}$ between the points of
intersection of $\ell_{i,\sigma(j)}$ with $\ell_{i,\sigma(j-1)}$ and
$\ell_{i,\sigma(j+1)}$. Let $\Pi_i$ be the polygonal region defined by
the chord $v_iv_{i+1}$ and the line segments
$s_{i,\sigma(j)}$. $\Pi_i$ is a linear polygon with at most two reflex
vertices (at $v_i$ and/or $v_{i+1}$).
It is easy to see that placing guards on the vertices of
the $\Pi_i$'s guards both $P$ and $Q$. Let $G_Q$ be the guard set
of $P$ constructed this way. Construct, now, a planar graph $\Gamma$
with vertex set $\mathcal{K}$. Two vertices $\kappa_i$ and $\kappa_j$
of $\Gamma$ are connected via an edge if
$\kappa_i$ and $\kappa_j$ are tangent. The graph $\Gamma$ is a planar
graph combinatorially equivalent to the triangulation graph $\tr{R}$
of a polygon $R$ with $n$ vertices. The edges of $\Gamma$ connecting
the arcs $\kappa_i$, $\kappa_{i+1}$, $1\le{}i\le{}n$, are the boundary
edges of $R$, whereas all other edges of $\Gamma$ correspond to
diagonals in $\tr{R}$. Let $\intrr{Q}$ denote the interior of
$Q$. Observing that $\intrr{Q}$ consists of a number of faces that are
in 1--1 correspondence with the triangles in $\tr{R}$, we conclude
that $\intrr{Q}$ consists of $n-2$ faces, each containing three
guards of $G_Q$. It fact, each face of $\intrr{Q}$ can actually be
guarded by only two of the three guards it contains and thus we can
eliminate one of them per face of $\intrr{Q}$. The new guard set
$G$ of $Q$ constructed above is also a guard set for $P$ and
contains $2(n-2)$ point guards.

To prove the necessity, refer to the \pconcave polygon $P$ in
Fig. \ref{fig:concavelb}. Each one of the pseudo-triangular regions in
the interior of $P$ requires exactly two point guards in order to be
guarded. Consider for example the pseudo-triangle $\tau$ shown in
gray in Fig. \ref{fig:concavelb}. We need one point along each one of
the lines $l_1$, $l_2$ and $l_3$ in order to guard the
regions near the corners of $\tau$, which implies that we need at
least two points in order to guard $\tau$ (two out of the three
points of intersection of the lines $l_1$, $l_2$ and $l_3$).
The number of such pseudo-triangular regions is exactly
$n-2$, thus we need a total of $2n-4$ point guards to guard $P$.
\proofbox
\end{proof}

\begin{figure*}[t]
\begin{center}
  \subfigure[\label{fig:concave_no_vertex_guards}]
  {\includegraphics[width=0.17\textwidth]{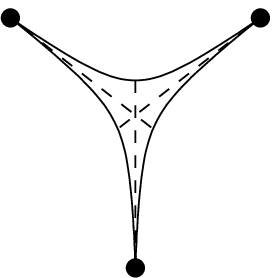}}
  \hfil
  \subfigure[\label{fig:concavelb}]
  {\psfrag{l1}[][]{\small$l_1$}
   \psfrag{l2}[][]{\small$l_2$}
   \psfrag{l3}[][]{\small$l_3$}
   \includegraphics[width=0.75\textwidth]{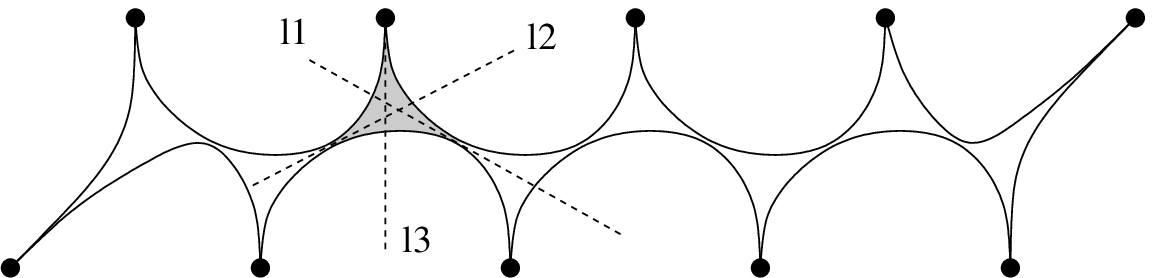}}
  \caption{(a) A \pconcave polygon $P$ that cannot be guarded
    solely by vertex guards. Two consecutive edges of $P$ have a common
    tangent at the common vertex and as a result the three vertices of
    $P$ see only the points along the dashed segments. (b) A
    \pconcave polygon $P$ that requires $2n-4$ point guards in
    order to be guarded.}
  \label{fig:concave_polygons}
\end{center}
\end{figure*}


\section{Locally convex and general polygons}
\label{sec:otherlb}

We have so far been dealing with the cases of \pconvex and
\pconcave polygons. In this section we will present results
about locally convex, monotone locally convex and general polygons.

\myparagraph{Locally convex polygons.}
The situation for locally convex polygons is much less interesting, as
compared to \pconvex polygons, in
the sense that there exist locally convex polygons that require $n$
vertex guards in order to be guarded. Consider for example the locally
convex polygon of Fig. \ref{fig:locallyconvexall}. Every room in this
polygon cannot be guarded by a single guard, but rather it requires
both vertices of every locally convex edge to be in any guarding set
in order for the corresponding room to be guarded. As a result it
requires $n$ vertex guards. Clearly, these $n$ guards are also
sufficient, since any one of them guards also the central convex part
of the polygon.
More interestingly, even if we do not restrict ourselves to vertex
guards, but rather allow guards to be any point in the interior or the
boundary of the polygon, then the locally convex polygon in
Fig. \ref{fig:locallyconvexall} still requires $n$ guards. This stems
from the fact that the rooms of this polygon have been constructed in
such a way so that the kernel of each room is the empty set (\ie
they are not star-shaped objects). However, we can guard each room
with two guards, which can actually be chosen to be the endpoints of
the locally convex arcs.

\begin{figure}[!t]
\begin{center}
\psfrag{v1}[][]{$v_1$}
\psfrag{v2}[][]{$v_2$}
\psfrag{a1}[][]{$a_1$}
\psfrag{a2}[][]{$a_2$}
\psfrag{l}[][]{$\ell$}
\subfigure[\label{fig:locallyconvexall}]
          {\includegraphics[width=0.3\textwidth]{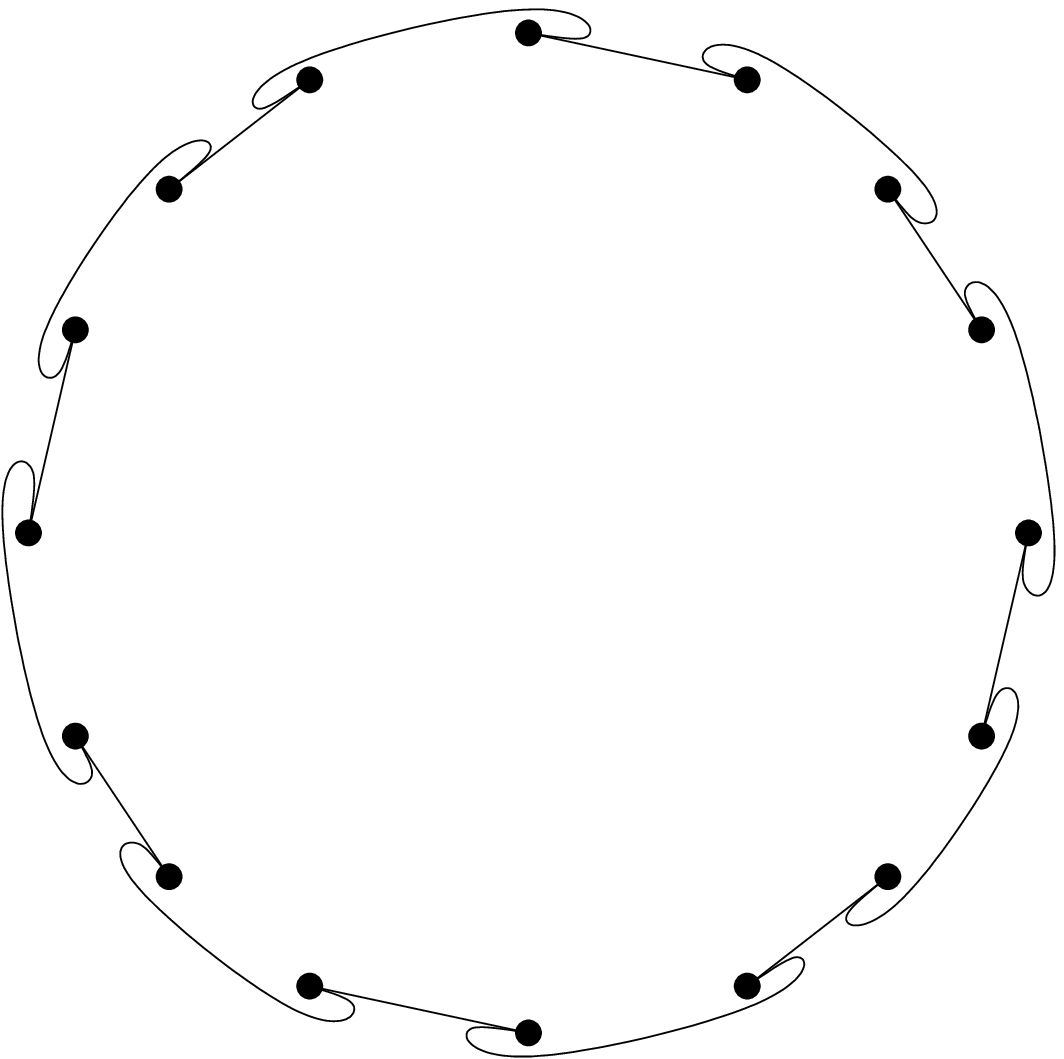}}\hfil%
\subfigure[\label{fig:nonconvexinfinite}]%
          {\includegraphics[width=0.3\textwidth]{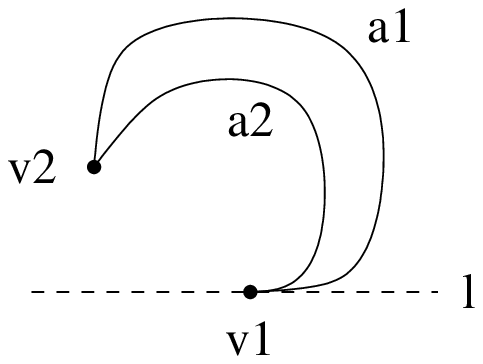}}
\caption{(a) A locally convex polygon with $n$ vertices that requires $n$
  vertex or point guards in order to be guarded. (b) A non-convex
  polygon that cannot be guarded by vertex guards, and which requires
  an infinite number of point guards.}
\label{fig:extensions}
\end{center}
\end{figure}

In fact the $n$ vertices of a locally convex polygon are not only
necessary (in the worst case), but also always sufficient.
Consider a point $q$ inside a locally convex polygon $P$ and let $\rho_q$
be an arbitrary ray emanating from $q$. Let $w_q$ be the first point
of intersection of $\rho_q$ with the boundary of $P$ as we walk on $\rho_q$
away from $q$. If $w_q$ is a vertex of $P$ we are done: $q$ is visible
by one of the vertices of $P$. Otherwise, rotate $\rho_q$ around $q$ in
the, say, counterclockwise direction, until the line segment $qw_q$
hits a feature $f$ of $P$ (if multiple features of $P$ are hit at the
same time, consider the one closest to $q$ along $\rho_q$). $f$ cannot
be a point in the interior of an edge of $P$ since then $P$ would have
to be locally concave at $f$. Therefore, $f$ has to be a vertex of
$P$, \ie $q$ is guarded by $f$. We can thus state the following
theorem:

\begin{theorem}
Let $P$ be a locally convex polygon with $n\ge{}2$ vertices. Then $n$
vertex (the $n$ vertices of $P$) or point guards are always
sufficient and sometimes necessary in order to guard $P$.
\end{theorem}

\myparagraph{Monotone locally convex polygons.}
As far as monotone locally convex polygons are concerned, it easy to
see that $\lfloor\frac{n}{2}\rfloor+1$ vertex or point guards are
always sufficient. Let $P$ be a locally convex polygon. As in the case
of \pconvex polygons, assume without loss of generality that
$P$ is monotone with respect to the $x$-axis. Let $u_1,\ldots,u_n$
be the vertices of $P$ sorted with respect to their $x$-coordinate.
To prove our sufficiency result, it suffices to consider the vertical
decomposition of $P$ into at most $n+1$ convex regions $\kappa_i$,
$0\le{}i\le{}n$. Corollary \ref{cor:vertical_decomp} remains valid.
As a result, the vertex set
$G=\{u_{2m-1}, 1\le{}m\le{}\lfloor\frac{n}{2}\rfloor\}\cup\{u_n\}$ 
is a guarding set for $P$ of size $\lfloor\frac{n}{2}\rfloor+1$: every
convex region $\kappa_i$, $1\le{}i\le{}n-1$ is guarded by either $u_i$ or
$u_{i+1}$, since at least one of $u_i$, $u_{i+1}$ is in $G$; moreover,
$u_1$ and $u_n$ guard $\kappa_0$ and $\kappa_n$, respectively. As in
the case of \pconvex polygons, $G$ can be computed in linear
time and space.

In fact, the upper bound on the number of vertex/point guards for $P$
just presented is also a worst case lower bound. Consider the locally
convex polygons $T_1$ and $T_2$ of Fig.
\ref{fig:monotonelocallypointguardslb}, each consisting of $n$
vertices. $T_1$ has an odd number
of vertices, while the number of vertices of $T_2$ is even. It is
readily seen that both $T_1$ and $T_2$ need at least one point guard
per prong (including the right-most prong of $T_1$ and both the
left-most and right-most prongs of $T_2$). Since the number of prongs
in either $T_1$ or $T_2$ is $\lfloor\frac{n}{2}\rfloor+1$, we conclude
that $T_1$ and $T_2$ require at least $\lfloor\frac{n}{2}\rfloor+1$
point guards in order to be guarded. Summarizing our results about
monotone locally convex polygons:

\begin{theorem}
Given a monotone locally convex polygon $P$ with $n\ge{}2$
vertices, $\lfloor\frac{n}{2}\rfloor+1$ vertex or point guards are
always sufficient and sometimes necessary in order to guard
$P$. Moreover, we can compute a vertex guarding set for $P$ of size
$\lfloor\frac{n}{2}\rfloor+1$ in $O(n)$ time and $O(n)$ space.
\end{theorem}

\begin{figure}[t]
\begin{center}
\includegraphics[width=0.99\textwidth]{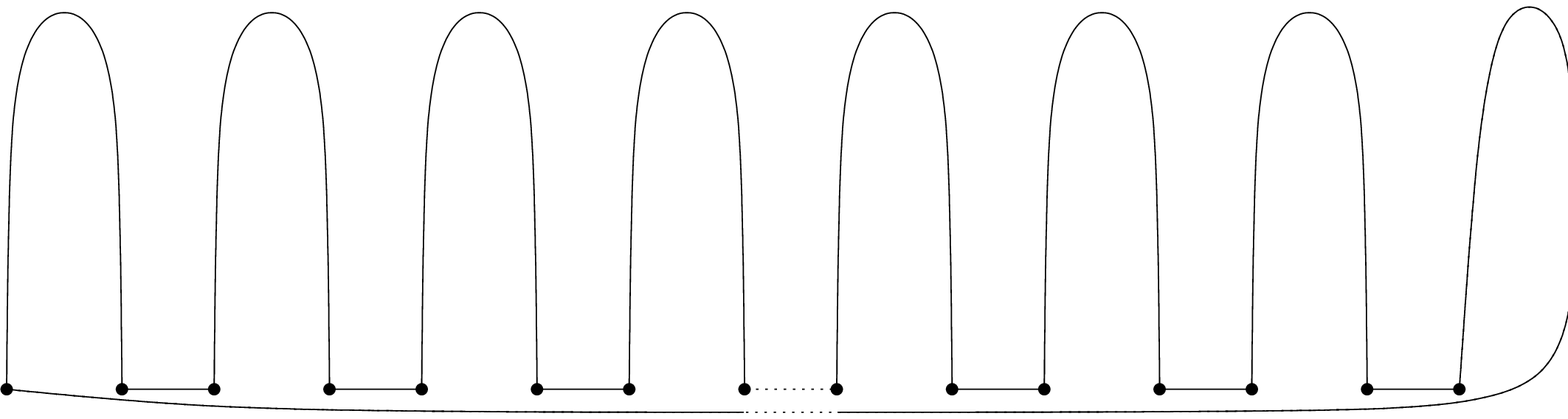}\\[5pt]
\includegraphics[width=0.99\textwidth]{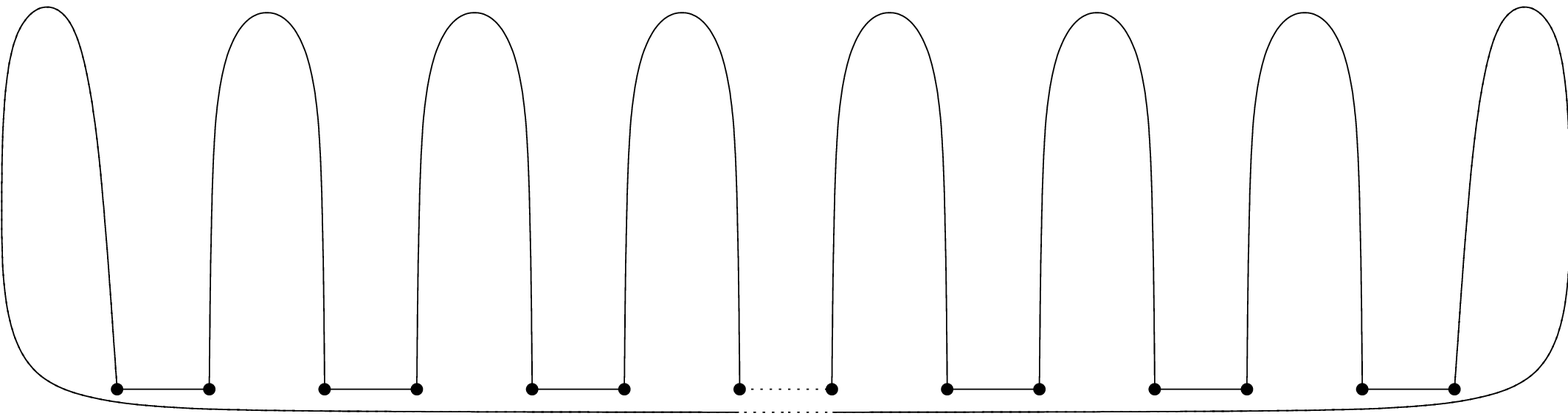}
\caption{Two comb-like monotone locally convex polygons $T_1$ (top)
  and $T_2$ (bottom) with an odd and even number of vertices,
  respectively. Both polygons require $\lfloor\frac{n}{2}\rfloor+1$ point
  guards in order to be guarded: one point guard is required per prong.}
\label{fig:monotonelocallypointguardslb}
\end{center}
\end{figure}

\begin{remark}\sl
  The results presented in this section about locally convex polygons
  are in essence the same with known results on the number of reflex
  vertices required to guard linear polygons. In particular, it is
  known that if a linear polygon $P$ has $r\ge{}1$ reflex vertices,
  $r$ vertex guards placed on these vertices are always sufficient and
  sometimes necessary in order to guard $P$ \cite{o-agta-87}, whereas
  if $P$ is a monotone linear polygon, $\lfloor\frac{r}{2}\rfloor+1$
  among its $r$ reflex vertices are always sufficient and sometimes
  necessary in order to guard $P$ \cite{a-agpiv-84}. In our setting,
  the $r$ reflex vertices of the linear polygon $P$ are the $n$
  vertices of our locally convex polygons, and the locally convex
  polylines connecting the reflex vertices of $P$ are our locally
  convex edges. Clearly, the analogy only refers to the combinatorial
  complexity of guarding sets, since for our algorithmic analysis we
  have assumed that the polygon edges have constant complexity.

  In the context we have just described, \ie seeing linear polygons
  as locally convex polygons the vertices of which are the reflex
  vertices of the linear polygons, it also possible to ``translate''
  the results of Section \ref{sec:piececonvex} as follows:
  \begin{quote}
    Consider a linear polygon $P$ with $r\ge{}2$ reflex vertices. If
    $P$ can be decomposed into $c\ge{}r$ convex polylines pointing
    towards the exterior of the polygon, then $P$ can be
    guarded with at most $\lfloor\frac{2c}{3}\rfloor$ vertex guards.
  \end{quote}
  The analogous ``translation'' for the results of Section
  \ref{sec:piececoncave} is as follows:
  \begin{quote}
    Consider a linear polygon $P$ with $n$ vertices, $r$ of
    which are reflex. If $P$ can be decomposed into $c\ge{}n-r$ convex
    polylines pointing towards the interior of the polygon, then $P$
    can be guarded with at most $2c-4$ point guards.
  \end{quote}

\end{remark}

\myparagraph{General polygons.}
The class of general polygons poses difficulties. Consider
the non-convex polygon $N$ of Fig. \ref{fig:nonconvexinfinite}, which
consists of two vertices $v_1$ and $v_2$ and two convex arcs $a_1$ and
$a_2$. The two arcs are tangent to a common line $\ell$ at $v_1$. It
is readily visible that $v_1$ and $v_2$ cannot guard the interior of
$N$. In fact, $v_1$ cannot guard any point of $N$ other than itself.
Even worse, any finite number of guards, placed anywhere in $N$,
cannot guard the polygon. To see that, consider the vicinity of
$v_1$. Assume that $N$ can be guarded by a finite number of guards,
and let $g\ne v_1$ be the guard closest to $v_1$ with respect to
shortest paths within $N$. Consider the line $\ell_g$ passing through $g$
that is tangent to $a_2$ (among the two possible tangents we are
interested in the one the point of tangency of which is closer to
$v_1$). Let $s_g$ be the sector of $N$ delimited by $a_1$, $a_2$
and $\ell_g$. $s_g$ cannot contain any guarding point, since such a
vertex would be closer to $v_1$ than $g$. Since $s_g$ is not guarded
by $v_1$, we conclude that $s_g$ is not guarded at all, which
contradicts our assumption that $N$ is guarded by a finite set of
guards.


\section{Summary and future work}
\label{sec:summary}

In this paper we have considered the problem of guarding a polygonal
art gallery, the walls of which are allowed to be arcs of curves
(our results are summarized in Table \ref{tbl:results}).
We have demonstrated that if we allow these arcs to be locally convex
arcs, $n$ (vertex or point) guards are always sufficient and sometimes
necessary. If these arcs are allowed to be non-convex, then an
infinite number of guards may be required.
In the case of \pconvex polygons with $n$ vertices, we have
shown that it is always possible to guard the polygon with
$\lfloor\frac{2n}{3}\rfloor$ vertex guards, whereas
$\lfloor\frac{4n}{7}\rfloor-1$ vertex guards are sometimes
necessary. Furthermore, we have described an $O(n\log n)$ time and
$O(n)$ space algorithm for computing a vertex guarding set of size at
most $\lfloor\frac{2n}{3}\rfloor$.
For \pconcave polygons, we have shown that $2n-4$
point guards are always sufficient and sometimes necessary. Finally,
in the special case of monotone \pconvex polygons,
$\lfloor\frac{n}{2}\rfloor+1$ vertex or $\lfloor\frac{n}{2}\rfloor$
point guards are always sufficient and sometimes necessary, whereas
for monotone locally convex polygons $\lfloor\frac{n}{2}\rfloor+1$
vertex or point guards are always sufficient and sometimes necessary.

Up to now we have not found a \pconvex polygon that requires
more than $\lfloor\frac{4n}{7}\rfloor+O(1)$ vertex guards, nor have we
devised a polynomial time algorithm for guarding a \pconvex
polygon with less than $\lfloor\frac{2n}{3}\rfloor$ vertex guards. Closing
the gap between then two complexities remains an open problem. Another
open problem is the worst case maximum number of point guards required to
guard a \pconvex polygon. In this case our lower bound
construction fails, since it is possible to guard the corresponding
polygon with $\lfloor\frac{3n}{7}\rfloor+O(1)$ point guards. On the
other hand, the comb-like polygon shown in
Fig. \ref{fig:monotonepointguardslb}, requires
$\lfloor\frac{n}{2}\rfloor$ point guards. Clearly, our algorithm that
computes a guarding set of at most $\lfloor\frac{2n}{3}\rfloor$ vertex
guards is still applicable.

\begin{table}[t]
\begin{center}
\begin{tabular}{|c|c|c|c|c|}\hline
&\multicolumn{4}{c|}{\sl{}Bounds by guard type}\\\cline{2-5}
\sl{}Polygon type&\multicolumn{2}{c|}{\sl{}Vertex}&
\multicolumn{2}{c|}{\sl{}Point}\\\cline{2-5}
&\sl{}Upper&\sl{}Lower&\sl{}Upper&\sl{}Lower\\\hline
\hline
\Pconvex&\small$\lfloor\frac{2n}{3}\rfloor$&
\small$\lfloor\frac{4n}{7}\rfloor-1$&
\small$\lfloor\frac{2n}{3}\rfloor$&
\small$\lfloor\frac{n}{2}\rfloor$\\\hline
Monotone \pconvex&
\multicolumn{2}{c|}{\small$\lfloor\frac{n}{2}\rfloor+1$}&
\multicolumn{2}{c|}{\small$\lfloor\frac{n}{2}\rfloor$}\\\hline
Locally convex&\multicolumn{4}{c|}{$n$}\\\hline
Monotone locally convex&\multicolumn{4}{c|}{$\lfloor\frac{n}{2}\rfloor+1$}
\\\hline
\Pconcave&\multicolumn{2}{c|}{\sc{}not always possible}&
\multicolumn{2}{c|}{$2n-4$}\\\hline
General&\multicolumn{2}{c|}{\sc{}not always possible}&
\multicolumn{2}{c|}{$\infty$}\\\hline
\end{tabular}
\caption{The results in this paper: worst case upper and lower bounds
  on the number of vertex or point guards needed in order to guard
  different types of curvilinear polygons.}
\label{tbl:results}
\end{center}
\end{table}

Other types of guarding problems have been studied in the literature,
which either differ on the type of guards (e.g., edge or mobile
guards), the topology of the polygons considered (e.g., polygons with
holes) or the guarding model (e.g., the fortress problem or 
the prison yard problem, mentioned in Section \ref{sec:intro}); see
the book by O'Rourke \cite{o-agta-87}, the survey paper by Shermer
\cite{s-rrag-92} of the book chapter by Urrutia \cite{u-agip-00} for
an extensive list of the variations of the art gallery problem with
respect to the types of guards or the guarding model. It would be
interesting to extend these results to the families of curvilinear
polygons presented in this paper.

Last but not least, in the case of general polygons, is it possible
to devise an algorithm for computing a guarding set of finite
cardinality, if the polygon does not contain cusp-like
configurations such as the one in Fig. \ref{fig:nonconvexinfinite}?


\section*{Acknowledgements}
The authors wish to thank Ioannis Z. Emiris, Hazel Everett and
G{\"u}nter Rote for useful discussions about the
problem. \ACSacknowledgement.

\bibliographystyle{abbrv}
\bibliography{art_gallery}

\end{document}